\documentclass[11pt]{amsart}
\usepackage{amsmath}
\usepackage[dvips]{epsfig}
\theoremstyle{plain}
\newtheorem{theorem}{Theorem}[section]
\newtheorem{lemma}[theorem]{Lemma}
\newtheorem{prop}[theorem]{Proposition}

\newtheorem{question}[theorem]{Question}

\theoremstyle{definition}

\numberwithin{equation}{section}
\newcommand{\volg}{1+ u^{2}| \nabla f |_{g}^{2}} %volume form of \bg%
\newcommand{\volbarg}{1-u^{2}| \overline{\nabla} f |_{\overline{g}}^{2}}  %volume form of  g%
\newcommand{\barG}{\overline{\Gamma}}  %barred Gamma- Christoffel symbol for barred g%
\newcommand{\tilG}{\widetilde{\Gamma}} %Christoffel symbol for lorentzian 4-metric%
\newcommand{\barna}{\overline{\nabla}} %barred nabla - covariant derivative w.r.t. barred g%
\newcommand{\tilna}{\widetilde{\nabla}} %tilde nabla - covariant derivative w.r.t. tilde{g}-4metric%
\newcommand{\bg}{\overline{g}}  %barred g%
\newcommand{\tg}{\widetilde{g}} %tilde g - Lorentzian 4-metric%
\newcommand{\by}{\overline{Y}}  %barred Y%
\newcommand{\byp}{Y^{\phi}}  % Y^\phi%
\newcommand{\bw}{\overline{w}} %barred w- spacelike component of normal vector to t=f%
\newcommand{\bk}{\overline{k}} %barred k - second fundamental form for t=0%
\newcommand{\bdiv}{div_{\overline{g}}} %barred divergence - divergence with respect to barred g%
\newcommand{\brho}{\overline{\rho}} %barred rho in Brill's coordinate%
\newcommand{\br}{\overline{r}} %barred rho in Brill's%
\newcommand{\bth}{\overline{\theta}} %barred theta in Brill's%
\newcommand{\bu}{\overline{U}} %barred U in Brill's coordinate %
\newcommand{\bz}{\overline{z}} %barred z in Brill's coordinate%
\newcommand{\bal}{\overline{\alpha}} %barred alpha in Brill's coordinate%

\begin{document}

\title[Deformations of Axially Symmetric Initial Data] {Deformations of Axially Symmetric Initial Data and the Mass-Angular Momentum Inequality}

\author[Cha]{Ye Sle Cha}
\address{Department of Mathematics\\
Stony Brook University\\
Stony Brook, NY 11794, USA}
\email{ycha@math.sunysb.edu}

\author[Khuri]{Marcus A. Khuri}
\address{Department of Mathematics\\
Stony Brook University\\
Stony Brook, NY 11794, USA}
\email{khuri@math.sunysb.edu}

\thanks{The second author acknowledges the support of
NSF Grants DMS-1007156 and DMS-1308753. This paper is also based upon work supported by NSF under
Grant No. 0932078 000, while the authors were in residence at
the Mathematical Sciences Research Institute in Berkeley, California, during
the fall of 2013.}

\begin{abstract}
We show how to reduce the general formulation of the mass-angular momentum inequality, for axisymmetric initial data
of the Einstein equations, to the known maximal case whenever a geometrically motivated system of equations admits a solution.
This procedure is based on a certain deformation of the initial data which preserves the relevant geometry, while achieving the maximal
condition and its implied inequality (in a weak sense) for the scalar curvature; this answers a question posed by R. Schoen. The primary equation involved, bears a strong resemblance to the Jang-type equations studied in the context of the positive mass theorem and the Penrose inequality. Each equation in the system is analyzed in detail individually, and it is shown that appropriate existence/uniqueness results hold with the solution satisfying desired asymptotics. Lastly, it is shown that the same reduction argument applies to the basic inequality yielding a lower bound for the area of black holes in terms of mass and angular momentum.

\end{abstract}
\maketitle

\section{Introduction}
\label{sec1} \setcounter{equation}{0}
\setcounter{section}{1}

The standard picture of gravitational collapse \cite{Choquet-Bruhat}, \cite{ChruscielGallowayPollack} asserts that generically, an asymptotically flat spacetime should eventually
settle down to a stationary final state, consisting of (possibly multiple) disconnected black hole spacetimes. The black hole uniqueness theorem implies that, in vacuum, each of these solutions must be the Kerr spacetime; note that there are still important
unresolved technical aspects associated with this uniqueness result \cite{ChruscielCosta0}. It is also conceivable that these black holes are coupled to matter fields. In any event, as in
Kerr, the following inequality holds between mass and angular momentum $m_{f}\geq\sqrt{|\mathcal{J}_{f}|}$ for each of the connected components of the final state, and hence for the final state itself. Moreover, as gravitational radiation carries positive energy, the mass of any initial state should not be smaller than that of the final state $m\geq m_{f}$. If auxiliary conditions are imposed,
one of which usually includes axisymmetry, in order to ensure the conservation of angular momentum, then $\mathcal{J}=\mathcal{J}_{f}$ where $\mathcal{J}$, $\mathcal{J}_{f}$ denote the (ADM) angular momentums of the initial and final state.
This leads to the mass-angular momentum inequality \cite{Dain}
\begin{equation}\label{1}
m\geq\sqrt{|\mathcal{J}|}
\end{equation}
for any initial state. A counterexample to \eqref{1} would pose a serious challenge to this standard picture of collapse, whereas a verification of \eqref{1} would only lend credence to this model.

Consider an initial data set $(M, g, k)$ for the Einstein equations. This consists of a 3-manifold $M$, Riemannian metric $g$, and symmetric 2-tensor $k$ representing the
extrinsic curvature (second fundamental form) of the embedding into spacetime, which satisfy the constraint equations
\begin{align}\label{2}
\begin{split}
16\pi\mu &= R+(Tr_{g}k)^{2}-|k|_{g}^{2},\\
8\pi J &= div_{g}(k-(Tr_{g}k)g).
\end{split}
\end{align}
Here $\mu$ and $J$ are the energy and momentum densities of the matter fields, respectively, and $R$ is the scalar curvature of $g$. The following inequality will be referred to as the dominant energy condition
\begin{equation}\label{2.1}
\mu\geq|J|_{g}.
\end{equation}
Suppose that $M$ has at least two ends, with one designated end being asymptotically flat, and the remainder being either asymptotically flat or asymptotically cylindrical. Recall that a domain $M_{\text{end}}\subset M$ is an
asymptotically flat end if it is diffeomorphic to $\mathbb{R}^{3}\setminus\text{Ball}$, and in the coordinates given by the asymptotic
diffeomorphism the following fall-off conditions hold
\begin{equation}\label{3}
g_{ij}=\delta_{ij}+o_{l}(r^{-\frac{1}{2}}),\text{ }\text{ }\text{ }\text{ }\partial g_{ij}\in L^{2}(M_{\text{end}}),\text{
}\text{ }\text{ }
\text{ }k_{ij}=O_{l-1}(r^{-\lambda}),\text{
}\text{ }\text{ }
\text{ }\lambda>\frac{5}{2},
\end{equation}
for some $l\geq 5$\footnote{The notation $f=o_{l}(r^{-a})$ asserts that $\lim_{r\rightarrow\infty}r^{a+n}\partial^{n}f=0$
for all $n\leq l$, and
$f=O_{l}(r^{-a})$ asserts that $r^{a+n}|\partial^{n}f|\leq C$ for all $n\leq l$. The assumption $l\geq 6$ is needed for
the results in \cite{Chrusciel}.}. In the context of the mass-angular momentum inequality, these asymptotics may be weakened, see for example \cite{SchoenZhou}. The asymptotics for cylindrical ends is
most easily described in Brill coordinates, to be given in the next section.

We say that the initial data are axially symmetric if the group of isometries of the Riemannian manifold $(M,g)$ has a subgroup
isomorphic to $U(1)$, and that the remaining quantities defining the initial data are invariant under the $U(1)$ action. In particular, if $\eta$ denotes the Killing field associated with this
symmetry, then
\begin{equation}\label{3.1}
\mathfrak{L}_{\eta}g=\mathfrak{L}_{\eta}k=0,
\end{equation}
where $\mathfrak{L}_{\eta}$ denotes Lie differentiation.
If $M$ is simply connected and the data are axially symmetric, it is shown in \cite{Chrusciel} that the analysis reduces to the study
of manifolds diffeomorphic to $\mathbb{R}^{3}$ minus a finite number of points. Each point represents a black hole, and has the geometry of an asymptotically flat or cylindrical end.
The fall-off conditions in the designated asymptotically flat end guarantee that the ADM mass and angular momentum are well-defined by the following limits
\begin{equation}\label{4}
m=\frac{1}{16\pi}\int_{S_{\infty}}(g_{ij,i}-g_{ii,j})\nu^{j},
\end{equation}
\begin{equation}\label{5}
\mathcal{J}=\frac{1}{8\pi}\int_{S_{\infty}}(k_{ij}-(Tr_{g} k)g_{ij})\nu^{i}\eta^{j},
\end{equation}
where $S_{\infty}$ indicates the limit as $r\rightarrow\infty$ of integrals over coordinate spheres $S_{r}$, with unit outer normal $\nu$. Note that \eqref{3} implies that the ADM linear momentum vanishes.

Angular momentum is conserved \cite{DainKhuriWeinsteinYamada} if
\begin{equation}\label{6}
J_{i}\eta^{i}=0.
\end{equation}
Moreover, when $M$ is simply connected, this is a necessary and sufficient condition \cite{DainKhuriWeinsteinYamada} for the existence of a twist potential $\omega$:
\begin{equation}\label{7}
2\epsilon_{ijl}(k^{jn}-(Tr_{g}k)g^{jn}))\eta^{l}\eta_{n}dx^{i}=d\omega
\end{equation}
where $\epsilon_{ijl}$ is the volume form for $g$.

In \cite{Dain0} Dain has confirmed \eqref{1} under the hypotheses that the
initial data have two ends, are maximal ($Tr_{g}k=0$), vacuum ($\mu=|J|_{g}=0$), and admit a global Brill coordinate system. He also established the rigidity statement, which asserts that equality occurs in \eqref{1} if and only if the initial data arise as the $t=0$ slice of the extreme Kerr spacetime. Chrusciel, Li, and Weinstein \cite{Chrusciel}, \cite{ChruscielLiWeinstein} improved these results by showing that global Brill coordinates exist under general conditions, and by replacing the vacuum assumption with the hypotheses that $\mu\geq 0$ and a twist potential exists; they also studied the case of multiple black holes. Later Schoen and Zhou \cite{SchoenZhou} gave a simplified proof for more general asymptotics, still assuming the maximal condition, and Zhou \cite{Zhou} treated the near maximal case. It should be noted that such results are false without the assumption of axial symmetry \cite{HuangSchoenWang}.

The focus of this paper is on the general case without the maximal or near maximal hypothesis. We will exhibit a reduction argument by which the general case is reduced to the maximal case, assuming that a canonical
system of elliptic PDEs possesses a solution. The procedure is motivated by, and bears a resemblance to, previous reduction arguments that have been applied to other geometric inequalities such as the positive mass theorem and the Penrose inequality \cite{BrayKhuri1}, \cite{BrayKhuri2}, \cite{DisconziKhuri}, \cite{Jang}, \cite{KhuriWeinstein}, \cite{SchoenYau}. Moreover, the primary equation is related to the Jang-type equations that appear in each of
these procedures. The end result yields a natural deformation of the initial data, in which the geometry relevant to the mass-angular momentum inequality is preserved, while achieving the maximal condition. In particular, this answers a question posed by R. Schoen \cite{Zhou}:

\begin{question}\label{q1}
Is there a canonical way to deform a non-maximal, axisymmetric, vacuum data to a
unique maximal, vacuum data with the same physical quantities, i.e. the mass and
angular momentum, which also preserves the axial symmetry?
\end{question}

This paper is organized as follows. In the next section we describe the deformation in detail, while in
Section \ref{sec3} the reduction argument is established and the case of equality is treated. Section \ref{sec4} contains an application to the basic inequality yielding a lower bound for the area of black holes in terms of mass and angular momentum. In Sections \ref{sec5} and \ref{sec6} we give an initial analysis of the canonical system of PDEs, and finally four appendices are added to include several important but lengthy calculations.\medskip

\textbf{Acknowledgements.} The authors would like to thank Lars Andersson, Piotr Chru\'{s}ciel, Sergio Dain, Marc Mars, Martin Reiris, Richard Schoen, and Xin Zhou for discussions related to this work.

\section{Deformation of Initial Data}
\label{sec2} \setcounter{equation}{0}
\setcounter{section}{2}

In this section we will describe the deformation procedure which leads to the reduction argument for the mass-angular momentum inequality. It will be assumed that $(M,g,k)$ is a simply connected,
axially symmetric initial data set with multiple ends as described in the previous section. Simple connectedness and axial symmetry imply \cite{Chrusciel} that $M\cong\mathbb{R}^{3}\setminus\sum_{n=1}^{N}i_{n}$, where $i_{n}$ are points in $\mathbb{R}^{3}$
and represent asymptotic ends (in total there are $N+1$ ends). Moreover there exists a global (cylindrical) Brill coordinate system $(\rho,\phi,z)$ on $M$, where the points $i_{n}$ all lie on the $z$-axis,
and in which the Killing field is given by $\eta=\partial_{\phi}$. In these coordinates the metric takes a simple form
\begin{equation}\label{9}
g=e^{-2U+2\alpha}(d\rho^{2}+dz^{2})+\rho^{2}e^{-2U}(d\phi+A_{\rho}d\rho+A_{z}dz)^{2},
\end{equation}
where $\rho e^{-U}(d\phi+A_{\rho}d\rho+A_{z}dz)$ is the dual 1-form to $|\eta|^{-1}\eta$ and all coefficient functions are independent of $\phi$. Let $M_{end}^{0}$ denote the end associated with limit $r=\sqrt{\rho^{2}+z^{2}}\rightarrow\infty$.
The asymptotically flat fall-off conditions \eqref{3} will be satisfied if
\begin{equation}\label{10}
U=o_{l-3}(r^{-\frac{1}{2}}),\text{ }\text{ }\text{ }\text{ }\alpha=o_{l-4}(r^{-\frac{1}{2}}),\text{ }\text{ }\text{ }\text{ }A_{\rho},A_{z}=o_{l-3}(r^{-\frac{3}{2}}).
\end{equation}
The remaining ends associated with the points $i_{n}$ will be denoted by $M_{end}^{n}$, and are associated with the limit $r_{n}\rightarrow 0$, where $r_{n}$ is the Euclidean distance to $i_{n}$. The asymptotics for
asymptotically flat and cylindrical ends are given, respectively, by
\begin{equation}\label{11}
U=2\log r_{n}+o_{l-4}(r_{n}^{\frac{1}{2}}),\text{ }\text{ }\text{ }\text{ }\alpha=o_{l-4}(r_{n}^{\frac{1}{2}}),\text{ }\text{ }\text{ }\text{ }A_{\rho},A_{z}=o_{l-3}(r_{n}^{\frac{3}{2}}),
\end{equation}
\begin{equation}\label{12}
U=\log r_{n}+o_{l-4}(r_{n}^{\frac{1}{2}}),\text{ }\text{ }\text{ }\text{ }\alpha=o_{l-4}(r_{n}^{\frac{1}{2}}),\text{ }\text{ }\text{ }\text{ }A_{\rho},A_{z}=o_{l-3}(r_{n}^{\frac{3}{2}}).
\end{equation}

It will also be assumed that the dominant energy condition \eqref{2.1} is satisfied, and that
\begin{equation}\label{13}
div_{g} k(\eta)=0,
\end{equation}
which is equivalent to \eqref{6}. Equation \eqref{13} gives rise to a twist potential $\omega$ \eqref{7} that is constant on each connected component of the axis of rotation. Let $I_{n}$ denote the
interval of the $z$-axis between $i_{n+1}$ and $i_{n}$, where $i_{0}=-\infty$ and $i_{N+1}=\infty$. Then a standard formula (see Appendix D) yields the angular momentum for each black hole
\begin{equation}\label{14}
\mathcal{J}_{n}=\frac{1}{8}(\omega|_{I_{n}}-\omega|_{I_{n-1}}).
\end{equation}
According to \eqref{5} and \eqref{13}, the total angular momentum is given by
\begin{equation}\label{15}
\mathcal{J}=\sum_{n=1}^{N}\mathcal{J}_{n}.
\end{equation}

We seek a deformation of the initial data $(M,g,k)\rightarrow(\overline{M},\overline{g},\overline{k})$ such that the manifolds are diffeomorphic $M\cong\overline{M}$, the geometry of the ends is preserved, and
\begin{equation}\label{16}
\overline{m}=m,\text{ }\text{ }\text{ }\text{ }\overline{\mathcal{J}}=\mathcal{J},\text{ }\text{ }\text{ }\text{ }Tr_{\overline{g}}\overline{k}=0,\text{ }\text{ }\text{ }\text{ }\overline{J}(\eta)=0,\text{ }\text{ }\text{ }\text{ }
\overline{R}\geq|\overline{k}|_{\overline{g}}^{2}\text{ }\text{ weakly, }\text{ }\text{ }
\end{equation}
where $\overline{m}$, $\overline{\mathcal{J}}$, $\overline{J}$, and $\overline{R}$ are the mass, angular momentum, momentum density, and scalar curvature of the new data. The inequality in \eqref{16} is said to hold `weakly' if it is valid when
integrated against an appropriate test function. The validity of this inequality plays a central role in the proof of the mass-angular momentum inequality in the maximal case, and it is precisely the lack of this inequality
in the non-maximal case which prevents the proof from generalizing. Thus, the primary goal of the deformation is to obtain such a lower bound for the scalar curvature, while preserving all other aspects of the geometry.

With intuition from previous work \cite{BrayKhuri1}, \cite{BrayKhuri2}, \cite{SchoenYau} we search for the deformation in the form of a graph inside a stationary 4-manifold
\begin{equation}\label{17}
\overline{M}=\{t=f(x)\}\subset(M\times\mathbb{R}, g+2Y_{i}dx^{i}dt+\varphi dt^{2}),
\end{equation}
where the 1-form $Y=Y_{i}dx^{i}$ and functions $\varphi$ and $f$ are defined on $M$ and satisfy
\begin{equation}\label{18}
\mathfrak{L}_{\eta}f=\mathfrak{L}_{\eta}\varphi=\mathfrak{L}_{\eta}Y=0.
\end{equation}
Define
\begin{equation}\label{19}
\overline{g}_{ij}=g_{ij}+f_{i}Y_{j}+f_{j}Y_{i}+\varphi f_{i}f_{j},\text{ }\text{ }\text{ }\text{ }
\overline{k}_{ij}=\frac{1}{2u}\left(\overline{\nabla}_{i}Y_{j}+\overline{\nabla}_{j}Y_{i}\right),
\end{equation}
where $f_{i}=\partial_{i}f$, $\overline{\nabla}$ is the Levi-Civita connection with respect to $\overline{g}$, and
\begin{equation}\label{20}
u^{2}=\varphi+|Y|_{\overline{g}}^{2}.
\end{equation}
In the `Riemannian' setting \eqref{17}, $\overline{g}$ arises as the induced metric on the graph $\overline{M}$. However in the `Lorentzian' setting
\begin{equation}\label{21}
M=\{t=f(x)\}\subset(\overline{M}\times\mathbb{R}, \overline{g}-2Y_{i}dx^{i}dt-\varphi dt^{2}),
\end{equation}
the deformed data arise as the induced metric and second fundamental form of the $t=0$ slice. Notice that
\begin{equation}\label{22}
\partial_{t}=un-\overline{Y},
\end{equation}
where $n$ is the unit normal to the $t=0$ slice and $\overline{Y}$ is the vector field dual to $Y$ with respect to $\overline{g}$. Thus $(u,-\overline{Y})$ comprise the lapse and shift of this stationary spacetime. Based on the structure of the Kerr spacetime, we make the following simplifying
assumption that $\overline{Y}$ has only one component
\begin{equation}\label{23}
\overline{Y}^{i}\partial_{i}:=\overline{g}^{ij}Y_{j}\partial_{i}=Y^{\phi}\partial_{\phi}.
\end{equation}

\begin{lemma}\label{lemma1}
Under the hypothesis \eqref{23}, $\overline{g}$ is a Riemannian metric, $Tr_{\overline{g}}\overline{k}=0$, and $\varphi=u^{2}-g_{\phi\phi}(Y^{\phi})^{2}$. Moreover if $\{e_{i}\}_{i=1}^{3}$ is an orthonormal frame for $\overline{g}$ with $e_{3}=|\eta|^{-1}\eta$, then
\begin{equation}\label{23.1}
\overline{k}(e_{i},e_{j})=\overline{k}(e_{3},e_{3})=0,\text{ }\text{ }\text{ }\text{ }
\overline{k}(e_{i},e_{3})=\frac{|\eta|}{2u}e_{i}(Y^{\phi}),
\text{ }\text{ }\text{ }\text{ }i,j\neq 3.
\end{equation}
Lastly
\begin{equation} \label{1001}
(\volg)(\volbarg) = 1,
\end{equation}
where $\nabla^{i} f=g^{ij}f_{j}$ and $\overline{\nabla}^{i}f=\overline{g}^{ij}f_{j}$.
\end{lemma}

\begin{proof}
From \eqref{18} it follows that $\overline{g}_{\phi\phi}=g_{\phi\phi}$, and so $|Y|_{\overline{g}}^{2}=g_{\phi\phi}(Y^{\phi})^{2}$. This yields the formula for $\varphi$.
Next observe that
\begin{align}\label{24}
\begin{split}
uTr_{\overline{g}}\overline{k} &= \overline{\nabla}_{i}Y^{i}\\
&= \partial_{i}Y^{i}-\overline{\Gamma}_{ij}^{i}Y^{j}\\
&=-\overline{\Gamma}_{i\phi}^{i}Y^{\phi}\\
&=-\left(\frac{1}{\sqrt{\det\overline{g}}}\partial_{\phi}\sqrt{\det\overline{g}}\right)Y^{\phi}\\
&=0,
\end{split}
\end{align}
where $\overline{\Gamma}_{ij}^{l}$ are Christoffel symbols.

We now show that $\overline{g}$ is Riemannian. Equations \eqref{18} and \eqref{23} imply that
\begin{equation}\label{25}
Y_{\phi}=g_{\phi\phi}Y^{\phi},\text{ }\text{ }\text{ }\text{ } Y_{i}=\overline{g}_{ij}Y^{j}=\overline{g}_{i\phi}Y^{\phi}=(g_{i\phi}+f_{i}Y_{\phi})Y^{\phi}=(g_{i\phi}+f_{i}g_{\phi\phi}Y^{\phi})Y^{\phi}.
\end{equation}
Inserting this into \eqref{19} produces
\begin{equation}\label{26}
\overline{g}_{ij}=g_{ij}+(f_{i}g_{j\phi}+f_{j}g_{i\phi})Y^{\phi}+(u^{2}+g_{\phi\phi}(Y^{\phi})^{2})f_{i}f_{j}.
\end{equation}
Take a $g$-orthonormal frame $(d_{1},d_{2},d_{3}=|\eta|^{-1}\eta)$ at a point, and express $\overline{g}$ as a matrix with respect to this frame
\begin{equation}\label{27}
\overline{g}=
\begin{pmatrix}
1+(u^{2}+g_{\phi\phi}(Y^{\phi})^{2})f_{1}^{2} & (u^{2}+g_{\phi\phi}(Y^{\phi})^{2})f_{1}f_{2} & \sqrt{g_{\phi\phi}}Y^{\phi}f_{1} \\
                                              & 1+(u^{2}+g_{\phi\phi}(Y^{\phi})^{2})f_{2}^{2} & \sqrt{g_{\phi\phi}}Y^{\phi}f_{2} \\
                                              &                                               & 1
\end{pmatrix}.
\end{equation}
The determinant of the lower $2\times 2$ minor is $1+u^{2}f_{2}^{2}> 0$, and the full determinant is given by
\begin{equation}\label{29}
\det\overline{g} = (1+u^{2}|\nabla f|_{g}^{2})\det g> 0.
\end{equation}
It follows that $\overline{g}$ is positive definite. Observe also that an analogous computation in the Lorentzian setting produces
\begin{equation}\label{29.0}
\det g = (1-u^{2}|\overline{\nabla} f|_{\overline{g}}^{2})\det\overline{g}.
\end{equation}
Equations \eqref{29} and \eqref{29.0} together yield \eqref{1001}.

In order to establish \eqref{23.1}, observe that
\begin{equation}\label{29.1}
2u\overline{k}_{ij}=\overline{\nabla}_{i}Y_{j}+\overline{\nabla}_{j}Y_{i}
=\partial_{i}Y_{j}+\partial_{j}Y_{i}-2\overline{\Gamma}_{ij}^{a}Y_{a},
\end{equation}
and
\begin{equation}\label{29.2}
\partial_{i}Y_{j}=\partial_{i}(\overline{g}_{\phi j}Y^{\phi})
=(\partial_{i}\overline{g}_{\phi j})Y^{\phi}+\overline{g}_{\phi j}\partial_{i}Y^{\phi},
\end{equation}
\begin{align}\label{29.3}
\begin{split}
2\overline{\Gamma}_{ij}^{a}Y_{a}&=\overline{g}^{al}(\partial_{i}\overline{g}_{jl}
+\partial_{j}\overline{g}_{il}-\partial_{l}\overline{g}_{ij})Y_{a}\\
&=(\partial_{i}\overline{g}_{j\phi}+\partial_{j}\overline{g}_{i\phi})Y^{\phi}.
\end{split}
\end{align}
Therefore
\begin{equation}\label{29.4}
2u\overline{k}_{ij}=\overline{g}_{\phi i}\partial_{j}Y^{\phi}+\overline{g}_{\phi j}\partial_{i}Y^{\phi}.
\end{equation}
Clearly $\overline{k}(e_{3},e_{3})=0$, and if we express $e_{i}$, $i=1,2$ in coordinates \eqref{45}, then
for $i,j=1,2$
\begin{align}\label{29.5}
\begin{split}
2u\overline{k}(e_{i},e_{j})&=2ue^{2\overline{U}-2\overline{\alpha}}(\overline{k}_{ij}
-A_{i}\overline{k}_{j\phi}-A_{j}\overline{k}_{i\phi}+A_{i}A_{j}\overline{k}_{\phi\phi})\\
&=e^{2\overline{U}-2\overline{\alpha}}(\overline{g}_{\phi i}\partial_{j}Y^{\phi}
+\overline{g}_{\phi j}\partial_{i}Y^{\phi}-A_{i}\overline{g}_{\phi\phi}\partial_{j}Y^{\phi}
-A_{j}\overline{g}_{\phi\phi}\partial_{i}Y^{\phi})\\
&=0,
\end{split}
\end{align}
since $\overline{g}_{\phi i}=A_{i}\overline{g}_{\phi\phi}$ from \eqref{42}. Also
\begin{equation}\label{29.6}
2u\overline{k}(e_{i},e_{3})=\frac{\overline{g}_{\phi\phi}}{|\eta|}e_{i}(Y^{\phi})
=|\eta|e_{i}(Y^{\phi}).
\end{equation}
\end{proof}

This lemma shows that the deformed data set is maximal, satisfying one requirement of \eqref{16}. Furthermore, it shows that $\varphi$ is determined by the functions
$u$ and $Y^{\phi}$. Thus, the three functions $(u,Y^{\phi},f)$ completely determine the new data, and will be chosen to satisfy the remaining statements in \eqref{16}, so as
to yield a reduction argument for the mass-angular momentum inequality.

The next task is to show how to choose the three functions $(u,Y^{\phi},f)$. In order to apply the techniques from the maximal case, the existence of a twist potential for $(\overline{M},\overline{g},\overline{k})$
is needed. Therefore we require
\begin{equation}\label{30}
div_{\overline{g}}\overline{k}(\eta)=0.
\end{equation}
This turns out to be a linear elliptic equation for $Y^{\phi}$ (if $u$ is independent of $Y^{\phi}$), as is shown in the appendix. As discussed in Section \ref{sec4}, the function $Y^{\phi}$ is uniquely determined
among bounded solutions of \eqref{30}, if the $r^{-3}$-fall-off rate is prescribed at $M_{end}^{0}$. In particular, we will choose the following boundary condition
\begin{equation}\label{31}
Y^{\phi}=-\frac{2\mathcal{J}}{r^{3}}+o_{2}(r^{-\frac{7}{2}})\text{ }\text{ }\text{ as }\text{ }\text{ }r\rightarrow\infty.
\end{equation}

\begin{lemma}\label{lemma2}
If $\overline{g}$ is asymptotically flat and $u\rightarrow 1$ as $r\rightarrow\infty$, then the boundary condition \eqref{31} guarantees that $\overline{\mathcal{J}}=\mathcal{J}$.
\end{lemma}

\begin{proof}
Observe that since $g_{\phi\phi}\sim r^{2}\sin^{2}\theta$ as $r\rightarrow\infty$, where $\rho=r\sin\theta$ and $z=r\cos\theta$, we have
\begin{align}\label{32}
\begin{split}
\overline{\mathcal{J}}&=\lim_{r\rightarrow\infty}\frac{1}{8\pi}\int_{S_{r}}\overline{k}(\partial_{\phi},\partial_{r})\\
&=\lim_{r\rightarrow\infty}\frac{1}{16\pi}\int_{0}^{\pi}\int_{0}^{2\pi}g_{\phi\phi}\partial_{r}Y^{\phi}r^{2}\sin\theta d\phi d\theta\\
&=\frac{3\mathcal{J}}{4}\int_{0}^{\pi}\sin^{3}\theta d\theta\\
&=\mathcal{J}.
\end{split}
\end{align}
\end{proof}

Let us now show how to choose $f$. As with previous deformations arising from the positive mass theorem and Penrose inequality, $f$ is chosen to impart positivity properties to the
scalar curvature. With this in mind, it is instructive to calculate the scalar curvature for an arbitrary $f$. The following result requires a long and detailed computation, and is therefore relegated
to the appendix.

\begin{theorem}\label{thm1}
Suppose that \eqref{3.1}, \eqref{13}, \eqref{18}, \eqref{23}, and \eqref{30} are satisfied, then the scalar curvature of $\overline{g}$ is given by
\begin{align}\label{33}
\begin{split}
\overline{R}-|\overline{k}|_{\overline{g}}^{2}=& 16\pi(\mu-J(v))+|k-\pi|_{g}^{2}+2u^{-1}div_{\overline{g}}(uQ)\\
&+(Tr_{g}\pi)^{2}-(Tr_{g}k)^{2}+2v(Tr_{g}\pi-Tr_{g}k),
\end{split}
\end{align}
where
\begin{equation}\label{35}
\pi_{ij}=\frac{u\nabla_{ij}f+u_{i}f_{j}+u_{j}f_{i}+\frac{1}{2u}(g_{i\phi}Y^{\phi}_{,j}+g_{j\phi}Y^{\phi}_{,i})}{\sqrt{1+u^{2}|\nabla f|_{g}^{2}}}
\end{equation}
is the second fundamental form of the graph $M$ in the Lorentzian setting,
\begin{equation}\label{34}
v^{i}=\frac{uf^{i}}{\sqrt{1+u^{2}|\nabla f|_{g}^{2}}},\text{ }\text{ }\text{ }\text{ }w^{i}=\frac{uf^{i}+u^{-1}\overline{Y}^{i}}{\sqrt{1+u^{2}|\nabla f|_{g}^{2}}},
\end{equation}
and
\begin{equation}\label{36}
Q_{i}=\overline{Y}^{j}\overline{\nabla}_{ij}f-u\overline{g}^{jl}f_{l}\overline{k}_{ij}+w^{j}(k-\pi)_{ij}+uf_{i}w^{l}w^{j}(k-\pi)_{lj}\sqrt{1+u^{2}|\nabla f|_{g}^{2}}.
\end{equation}
Furthermore, if $Y\equiv 0$ then the same conclusion holds without any of the listed hypotheses.
\end{theorem}

This theorem, together with the dominant energy condition \eqref{2.1}, make it clear that in order to obtain the inequality $\overline{R}\geq|\overline{k}|_{\overline{g}}^{2}$ at least weakly,
$f$ should be chosen to solve the equation
\begin{equation}\label{37}
Tr_{g}(\pi-k)=0.
\end{equation}
It follows that
\begin{equation}\label{38}
\overline{R}-|\overline{k}|_{\overline{g}}^{2}= 16\pi(\mu-J(v))+|k-\pi|_{g}^{2}+2u^{-1}div_{\overline{g}}(uQ),
\end{equation}
which yields the inequality in \eqref{16} after multiplying by $u$ and applying the divergence theorem; it is assumed that appropriate asymptotic conditions are imposed (see below) in order
to ensure that the boundary integrals vanish in each of the ends. Equation \eqref{37} is similar to previous Jang-type equations that have been used in connection with deformations of initial data, in particular
for the positive mass theorem \cite{SchoenYau} and the Penrose inequality \cite{BrayKhuri2}. These previous equations have the form
\begin{equation}\label{39}
Tr_{\overline{g}}(\pi-k)=0,
\end{equation}
where it is assumed that $u=1$ and $Y=0$ \cite{SchoenYau}, and $Y=0$ \cite{BrayKhuri2}. Note that \eqref{39} does not reduce to \eqref{37} even in the setting of \cite{SchoenYau} or \cite{BrayKhuri2}.
This suggests that there is a significant difference between these two equations. In fact, solutions of \eqref{37} do not blow-up, while solutions of \eqref{39}
typically blow-up at apparent horizons or can be prescribed to blow-up at these surfaces \cite{HanKhuri}. This separate behavior arises from the fact that the trace in \eqref{37}
is taken with respect to $g$, whereas the trace in \eqref{39} is taken with respect to $\overline{g}$. As a result, the analysis of \eqref{37} is much more simple than that of \eqref{39}.  Lastly, in order to ensure that $\overline{m}=m$, we will impose the following asymptotics
\begin{equation}\label{40}
|f|+r|\nabla f|_{g}+r^{2}|\nabla^{2} f|_{g}\leq cr^{-\varepsilon}\text{ }\text{ }\text{ in }\text{ }\text{ }M_{end}^{0},
\end{equation}
for some $0<\varepsilon<1$. A bounded solution may be obtained by prescribing the following asymptotics at the remaining ends
\begin{equation}\label{41}
r_{n}^{-1}|\nabla f|_{g}+r_{n}^{-2}|\nabla^{2} f|_{g}\leq c\text{ }\text{ }\text{ in asymptotically flat }\text{ }\text{ }M_{end}^{n},
\end{equation}
\begin{equation}\label{41.0}
|\nabla f|_{g}+|\nabla^{2} f|_{g}\leq cr_{n}^{\frac{1}{2}}\text{ }\text{ }\text{ in asymptotically cylindrical }\text{ }\text{ }M_{end}^{n}.
\end{equation}

At this point we have shown how to choose $f$ and $Y$, in order to produce a deformation of the initial data which satisfies \eqref{16}. It remains to choose $u$, in such a way as to facilitate
a proof of the mass-angular momentum inequality. This shall be accomplished in the next section.

\section{The Reduction Argument and Case of Equality}
\label{sec3} \setcounter{equation}{0}
\setcounter{section}{3}

Here we shall follow the maximal case proof of the mass-angular momentum inequality, within the setting of the deformed initial data $(\overline{M},\overline{g},\overline{k})$.
The primary stumbling block is a lack of the pointwise scalar curvature inequality as appearing in \eqref{16}. However a judicious choice of $u$ will overcome this difficulty.

Assuming that the functions $(u,Y^{\phi},f)$ are chosen to possess the appropriate asymptotics, the geometry of the ends will be preserved in the deformation. Since the deformed
data are also simply connected and axially symmetric, the results of \cite{Chrusciel} apply to yield a global Brill coordinate system $(\overline{\rho},\phi,\overline{z})$ such that
\begin{equation}\label{42}
\overline{g}=e^{-2\overline{U}+2\overline{\alpha}}(d\overline{\rho}^{2}+d\overline{z}^{2})
+\overline{\rho}^{2}e^{-2\overline{U}}(d\phi+A_{\overline{\rho}}d\overline{\rho}+A_{\overline{z}}d\overline{z})^{2}.
\end{equation}
Next, recall that \eqref{30} implies the existence of a twist potential $\overline{\omega}$.
An important property of the Brill coordinates is that they yield a simple formula for the mass (\cite{Brill}, \cite{Dain0})
\begin{equation}\label{43}
\overline{m}-\mathcal{M}(\overline{U},\overline{\omega})
=\frac{1}{32\pi}\int_{\mathbb{R}^{3}}\left(2e^{-2\overline{U}+2\overline{\alpha}}\!\text{ }\overline{R}
+\overline{\rho}^{2}e^{-2\overline{\alpha}}(A_{\overline{\rho},\overline{z}}-A_{\overline{z},\overline{\rho}})^{2}
-\overline{g}_{\phi\phi}^{\!\text{ }-2}|\partial\overline{\omega}|^{2}\right)dx,
\end{equation}
where $|\partial\overline{\omega}|$ and $dx$ denote the Euclidean norm and volume element, and
\begin{equation}\label{44}
\mathcal{M}(\overline{U},\overline{\omega})
=\frac{1}{32\pi}\int_{\mathbb{R}^{3}}\left(4|\partial\overline{U}|^{2}+\overline{g}_{\phi\phi}^{\!\text{ }-2}|\partial\overline{\omega}|^{2}\right)dx.
\end{equation}

Let
\begin{equation}\label{45}
e_{\overline{\rho}}=e^{\overline{U}-\overline{\alpha}}(\partial_{\overline{\rho}}-A_{\overline{\rho}}\partial_{\phi}),
\text{ }\text{ }\text{ }\text{ } e_{\overline{z}}=e^{\overline{U}-\overline{\alpha}}(\partial_{\overline{z}}-A_{\overline{z}}\partial_{\phi}),
\text{ }\text{ }\text{ }\text{ } e_{\phi}=\frac{1}{\sqrt{\overline{g}_{\phi\phi}}}\partial_{\phi},
\end{equation}
be an orthonormal frame. Then according to \eqref{7} and $\overline{g}_{\phi\phi}=g_{\phi\phi}$,
\begin{equation}\label{46}
\overline{k}(e_{\overline{\rho}},e_{\phi})=-\frac{1}{2|\eta|_{\overline{g}}^{2}}e_{\overline{z}}(\overline{\omega})=-\frac{e^{\overline{U}-\overline{\alpha}}}{2g_{\phi\phi}}\partial_{\overline{z}}\overline{\omega},
\text{ }\text{ }\text{ }\text{ } \overline{k}(e_{\overline{z}},e_{\phi})=\frac{1}{2|\eta|_{\overline{g}}^{2}}e_{\overline{\rho}}(\overline{\omega})=\frac{e^{\overline{U}-\overline{\alpha}}}{2g_{\phi\phi}}\partial_{\overline{\rho}}\overline{\omega}.
\end{equation}
In light of Lemma \ref{lemma1} it follows that
\begin{equation}\label{47}
|\overline{k}|_{\overline{g}}^{2}=2(\overline{k}(e_{\overline{\rho}},e_{\phi})^{2}+\overline{k}(e_{\overline{z}},e_{\phi})^{2})
=\frac{e^{2\overline{U}-2\overline{\alpha}}}{2g_{\phi\phi}^{2}}|\partial\overline{\omega}|^{2},
\end{equation}
and hence with the help of Theorem \ref{thm1} and the dominant energy condition
\begin{align}\label{48}
\begin{split}
\overline{m}-\mathcal{M}(\overline{U},\overline{\omega})
&\geq\frac{1}{32\pi}\int_{\mathbb{R}^{3}}2e^{-2\overline{U}+2\overline{\alpha}}(\overline{R}-|\overline{k}|_{\overline{g}}^{2})dx\\
&\geq\frac{1}{8\pi}\int_{\mathbb{R}^{3}}\frac{e^{-2\overline{U}+2\overline{\alpha}}}{u}div_{\overline{g}}(uQ)dx\\
&\geq\frac{1}{8\pi}\int_{\overline{M}}\frac{e^{\overline{U}}}{u}div_{\overline{g}}(uQ)dx_{\overline{g}},
\end{split}
\end{align}
where the volume element for $\overline{g}$ is given by $dx_{\overline{g}}=e^{-3\overline{U}+2\overline{\alpha}}dx$.

Inequality \eqref{48} suggests that we choose
\begin{equation}\label{49}
u=e^{\overline{U}}=\frac{\overline{\rho}}{\sqrt{\overline{g}_{\phi\phi}}}=\frac{\overline{\rho}}{\sqrt{g_{\phi\phi}}}.
\end{equation}
If $\overline{g}$ preserves the asymptotic geometry of $g$, then based on \eqref{10}, \eqref{11}, \eqref{12}
\begin{equation}\label{50}
u=1+o_{l-3}(r^{-\frac{1}{2}})\text{ }\text{ }\text{ as }\text{ }\text{ }r\rightarrow\infty\text{ }\text{ }
\text{ in }\text{ }\text{ }M_{end}^{0},
\end{equation}
\begin{equation}\label{51}
u=r_{n}^{2}+o_{l-4}(r_{n}^{\frac{5}{2}})\text{ }\text{ }\text{ as }\text{ }\text{ }r_{n}\rightarrow 0\text{ }\text{ }\text{ in asymptotically flat}\text{ }\text{ }M_{end}^{n},
\end{equation}
\begin{equation}\label{52}
u=r_{n}+o_{l-4}(r_{n}^{\frac{3}{2}})\text{ }\text{ }\text{ as }\text{ }\text{ }r_{n}\rightarrow 0\text{ }\text{ }\text{ in asymptotically cylindrical}\text{ }\text{ }M_{end}^{n},
\end{equation}
where $r_{n}$ is the Euclidean distance to the point $i_{n}$ defining the end. Therefore, with the help of
the asymptotics for $f$ \eqref{40}, \eqref{41} and $Y^{\phi}$ \eqref{31}, as well as the assumption
\begin{equation} \label{52.0}
|k|_{g}+|k(\partial_{\phi},\cdot)|_{g}+|k(\partial_{\phi},\partial_{\phi})|\leq c\text{ }\text{ }\text{ on }
\text{ }\text{ }M,
\end{equation}
the asymptotic boundary integrals arising from the right-hand side of \eqref{48} all vanish as long as $\mathcal{J}=\overline{\mathcal{J}}$. This is proven in Appendix C, Section \ref{sec9}. It follows that
\begin{equation}\label{53}
\overline{m}\geq\mathcal{M}(\overline{U},\overline{\omega}).
\end{equation}

\begin{theorem}\label{thm2}
Let $(M,g,k)$ be a smooth, simply connected, axially symmetric initial data set satisfying the dominant energy condition \eqref{2.1} and conditions \eqref{6}, \eqref{52.0}, and with two ends, one designated asymptotically flat and the other either asymptotically flat or asymptotically cylindrical. If the system of equations \eqref{30}, \eqref{37}, \eqref{49} admits a smooth solution $(u,Y^{\phi},f)$ satisfying the asymptotics \eqref{31}, \eqref{40}-\eqref{41.0}, \eqref{50}-\eqref{52}, then
\begin{equation}\label{54}
m\geq\sqrt{|\mathcal{J}|}
\end{equation}
and equality holds if and only if $(M,g,k)$ arises from an embedding into the extreme Kerr spacetime.
\end{theorem}

\begin{proof}
The existence of a solution $(u,Y^{\phi},f)$ ensures that we may apply the maximal case proof to the deformed
initial data $(\overline{M},\overline{g},\overline{k})$ as above, arriving at the inequality \eqref{53}. The results of \cite{ChruscielLiWeinstein}, \cite{Dain0}, \cite{SchoenZhou} then imply that
\begin{equation}\label{55}
\mathcal{M}(\overline{U},\overline{\omega})\geq\sqrt{|\overline{\mathcal{J}}|}.
\end{equation}
Moreover, according to \eqref{16} $\overline{m}=m$ and $\overline{\mathcal{J}}=\mathcal{J}$, and hence \eqref{53} yields the desired inequality \eqref{54}.

Consider now the case of equality in \eqref{54}. In the process of deriving \eqref{53}, several nonnegative
terms were left out from the right-hand side. These terms arise from \eqref{33} and \eqref{43}. In the current
situation, they must all vanish
\begin{equation}\label{56}
|\mu-J(v)|=|k-\pi|_{g}=|A_{\overline{\rho},\overline{z}}-A_{\overline{z},\overline{\rho}}|=0.
\end{equation}
Furthermore, in light of the dominant energy condition, the fact that $|v|_{g}<1$, and the identity
\begin{equation}\label{57}
\mu-J(v)=(\mu-|J|_{g})+(1-|v|_{g})|J|_{g}+(|J|_{g}|v|_{g}-J(v)),
\end{equation}
it follows that
\begin{equation}\label{58}
\mu=|J|_{g}=0.
\end{equation}

We claim that $(\overline{M},\overline{g},\overline{k})$ is now a vacuum initial data set. By Lemma \ref{lemma1}
$Tr_{\overline{g}}\overline{k}=0$, so that the momentum density is given by
\begin{equation}\label{59}
8\pi \overline{J}=div_{\overline{g}}\overline{k}.
\end{equation}
Let $\{e_{i}\}_{i=1}^{3}$ denote the orthonormal basis \eqref{45} with $e_{3}=e_{\phi}$, then
\begin{align}\label{60}
\begin{split}
(div_{\overline{g}}\overline{k})(e_{i})&=\sum_{i=1}^{3}(\overline{\nabla}_{e_{j}}\overline{k})(e_{i},e_{j})\\
&=\sum_{j=1}^{3}\left[e_{j}(\overline{k}(e_{i},e_{j}))
-\sum_{a=1}^{3}\langle\overline{\nabla}_{e_{j}}e_{i},e_{a}\rangle\overline{k}(e_{a},e_{j})
-\sum_{a=1}^{3}\langle\overline{\nabla}_{e_{j}}e_{j},e_{a}\rangle\overline{k}(e_{i},e_{a})\right].
\end{split}
\end{align}
Assume now that $i\neq 3$, then by Lemma \ref{lemma1}
\begin{equation}\label{61}
\sum_{j=1}^{3}e_{j}(\overline{k}(e_{i},e_{j}))=0
\end{equation}
and
\begin{equation}\label{62}
(div_{\overline{g}}\overline{k})(e_{i})
=-\sum_{j=1}^{2}\langle\overline{\nabla}_{e_{j}}e_{i},e_{3}\rangle\overline{k}(e_{3},e_{j})
-\sum_{a=1}^{2}\langle\overline{\nabla}_{e_{3}}e_{i},e_{a}\rangle\overline{k}(e_{a},e_{3})
-\sum_{j=1}^{3}\langle\overline{\nabla}_{e_{j}}e_{j},e_{3}\rangle\overline{k}(e_{i},e_{3}).
\end{equation}
The last sum is zero since $\partial_{\phi}$ is a Killing field. Moreover
\begin{equation}\label{63}
\langle\overline{\nabla}_{e_{3}}e_{i},e_{a}\rangle=-\langle e_{i},\overline{\nabla}_{e_{3}}e_{a}\rangle
=-\langle e_{i},\overline{\nabla}_{e_{a}}e_{3}\rangle
=\langle\overline{\nabla}_{e_{a}}e_{i},e_{3}\rangle,
\end{equation}
since
\begin{equation}\label{64}
[\partial_{\phi},e_{a}]=\mathcal{L}_{\partial_{\phi}}e_{a}=0.
\end{equation}
Thus, we need only show that the first sum in \eqref{62} vanishes. To accomplish this, observe that
\begin{equation}\label{65}
\langle\overline{\nabla}_{e_{j}}e_{i},e_{3}\rangle=-\langle\overline{\nabla}_{e_{i}}e_{j},e_{3}\rangle
\end{equation}
as $\partial_{\phi}$ is Killing. Furthermore a direct computation shows that
\begin{equation}\label{66}
\langle[e_{\overline{\rho}},e_{\overline{z}}],e_{3}\rangle
=|\eta|e^{2\overline{U}-2\overline{\alpha}}(A_{\overline{\rho},\overline{z}}-A_{\overline{z},\overline{\rho}})
=0,
\end{equation}
where \eqref{56} was used. Therefore
\begin{equation}\label{67}
\langle\overline{\nabla}_{e_{j}}e_{i},e_{3}\rangle=\langle\overline{\nabla}_{e_{i}}e_{j},e_{3}\rangle,
\end{equation}
and it follows that the first sum in \eqref{62} vanishes. Hence $|\overline{J}|_{\overline{g}}=0$.

Consider now the energy density
\begin{equation}\label{68}
16\pi \overline{\mu}=\overline{R}+(Tr_{\overline{g}}\overline{k})^{2}-|\overline{k}|_{\overline{g}}^{2}
=\overline{R}-|\overline{k}|_{\overline{g}}^{2}.
\end{equation}
A lengthy computation \eqref{0016} in Appendix A shows that
\begin{equation}\label{69}
\overline{R}-|\overline{k}|_{\overline{g}}^{2}
=-2(div_{\overline{g}}\overline{k})(u\overline{\nabla} f)
+16\pi(\mu-J(v))+|k|_{g}^{2}-|\pi|_{g}^{2}+2(div_{g}k)(v)-2(div_{g}\pi)(v),
\end{equation}
when equation \eqref{37} is satisfied. However, $|\overline{J}|_{\overline{g}}=0$ and \eqref{56} imply that the
right-hand side vanishes. Thus $\overline{\mu}=0$, and $(\overline{M},\overline{g},\overline{k})$
is a vacuum initial data set.

Next, since $\overline{m}=\overline{\mathcal{J}}$ we may now apply the results of \cite{Dain0} and \cite{SchoenZhou} to conclude that $(\overline{M},\overline{g},\overline{k})$ is isometric to the $t=0$ slice
$(\mathbb{R}^{3}-\{0\},g_{EK},k_{EK})$ of the extreme Kerr spacetime $\mathbb{EK}^{4}$. Consider the map $M\rightarrow\mathbb{EK}^{4}$ given by $x\mapsto(x,f(x))$. The induced metric on the graph is given by
\begin{equation}\label{70}
(g_{EK})_{ij}-f_{i}(Y_{EK})_{j}-f_{j}(Y_{EK})_{i}-(u_{EK}^{2}-|Y_{EK}|_{g_{EK}}^{2}) f_{i}f_{j},
\end{equation}
where
\begin{equation}\label{71}
(k_{EK})_{ij}=\frac{1}{2u_{EK}}\left(\nabla^{EK}_{i}(Y_{EK})_{j}+\nabla^{EK}_{j}(Y_{EK})_{i}\right),
\end{equation}
and $(u_{EK},-Y_{EK})$ are the lapse and shift. If $\partial_{\phi}$ denotes the spacelike Killing field
in this spacetime, then $g_{EK}^{ij}(Y_{EK})_{j}\partial_{i}=Y_{EK}^{\phi}\partial_{\phi}$, and $Y_{EK}^{\phi}$
satisfies equation \eqref{30} with $\overline{g}$ replaced by $g_{EK}$, as well as boundary condition \eqref{31}. Since there is a unique solution to \eqref{30}, \eqref{31}, and $\overline{g}\cong g_{EK}$, we have that $Y=Y_{EK}$. Moreover it is a direct calculation to find that $u_{EK}=e^{U_{EK}}=e^{\overline{U}}=u$, where
$U_{EK}$ arises from the Brill coordinate expression for $g_{EK}$. It now follows from \eqref{19} and \eqref{20}
that $g$ agrees with the induced metric \eqref{70}. Furthermore, from \eqref{56} $\pi=k$, showing that the
second fundamental form of the embedding $(M,g)\hookrightarrow\mathbb{EK}^{4}$ is given by $k$. Therefore
the initial data $(M,g,k)$ arise from the extreme Kerr spacetime.

Lastly, if $(M,g,k)$ arises from extreme Kerr, then by the properties of this spacetime, equality in
\eqref{54} holds.
\end{proof}

Theorem \ref{thm2} reduces the proof of the mass-angular momentum inequality, in the general non-maximal case, to  the existence of a solution $(u,Y^{\phi},f)$ to the system of equations \eqref{30}, \eqref{37}, and \eqref{49}.
Notice that this is in fact a coupled system, as the definition of $u$ depends on $\overline{g}$. The first task, which is addressed in the next section, is to analyze the given asymptotic boundary value problems associated with each equation \eqref{30} and \eqref{37}. Before doing so, however, we record the reduction statement for multiple black holes. Let
\begin{equation}\label{72}
\mathcal{F}(\mathcal{J}_{1},\ldots,\mathcal{J}_{N})
\end{equation}
denote the numerical value of the action functional \eqref{44} evaluated at the harmonic map, from
$\mathbb{R}^{3}-\{\overline{\rho}=0\}$ to the two-dimensional hyperbolic space, constructed in Proposition 2.1
of \cite{ChruscielLiWeinstein}. Whether the square of this value agrees with
\begin{equation}\label{73}
\mathcal{J}=\sum_{n=1}^{N}\mathcal{J}_{n}
\end{equation}
is an important open problem. The proof of the following theorem is analogous to that of Theorem \ref{thm2}.

\begin{theorem}\label{thm3}
Let $(M,g,k)$ be a smooth, simply connected, axially symmetric initial data set satisfying the dominant energy condition \eqref{2.1} and conditions \eqref{6}, \eqref{52.0}, and with $N+1$ ends, one designated asymptotically flat and the others either asymptotically flat or asymptotically cylindrical. If the system of equations \eqref{30}, \eqref{37}, \eqref{49} admits a smooth solution $(u,Y^{\phi},f)$ satisfying the asymptotics \eqref{31}, \eqref{40}-\eqref{41.0}, \eqref{50}-\eqref{52}, then
\begin{equation}\label{74}
m\geq\mathcal{F}(\mathcal{J}_{1},\ldots,\mathcal{J}_{N}).
\end{equation}
\end{theorem}

\section{A Lower Bound for Area in Terms of Mass and Angular Momentum}
\label{sec4} \setcounter{equation}{0}
\setcounter{section}{4}

In this section we observe that the reduction argument given above, immediately applies to another
geometric inequality for axisymmetric black holes. Let $(M,g,k)$ be as in the previous section, with the
restriction that it possesses only two ends denoted $M_{end}^{\pm}$, such that $M_{end}^{+}$ is asymptotically flat and $M_{end}^{-}$ is either asymptotically flat or asymptotically cylindrical. Based on the heuristic
arguments of Section \ref{sec1} leading to the mass-angular momentum inequality \eqref{1}, combined with
the Hawking area theorem \cite{Hawking}, the following upper and lower bounds are derived \cite{DainKhuriWeinsteinYamada}
\begin{equation}\label{4.1}
m^2-\sqrt{m^4-\mathcal{J}^2}\leq\frac{A_{min}}{8\pi}
\leq m^2+\sqrt{m^4-\mathcal{J}^2},
\end{equation}
where $A_{min}$ is the minimum area required to enclose $M_{end}^{-}$. In \cite{DainKhuriWeinsteinYamada} the lower bound is established in the maximal case, and it is also shown that equality occurs if and only if the
initial data set is isometric to the $t=0$ slice of the extreme Kerr spacetime. The proof relies upon the mass-angular momentum inequality and the area-angular momentum inequality $A_{min}\geq8\pi|\mathcal{J}|$ (\cite{ClementJaramilloReiris}, \cite{DainReiris}). In the non-maximal case, the area-angular momentum inequality has also been established when $A_{min}$ is replaced by the area of a stable, axisymmetric, marginally outer trapped surface (\cite{ClementJaramilloReiris}, \cite{DainJaramilloReiris}). Thus, since we have shown (in the previous section) how to reduce the non-maximal case of the mass-angular momentum inequality to the problem of solving a coupled system of elliptic equations, an analogous lower bound for area may also be reduced to the
same problem. More precisely, Theorem \ref{thm2} combined with Theorem 1.1 in \cite{ClementJaramilloReiris} and the proof of a Theorem 2.5 in \cite{DainKhuriWeinsteinYamada}, produces the following result.

\begin{theorem}\label{thm4}
Let $(M,g,k)$ be a smooth, simply connected, axially symmetric initial data set satisfying the dominant energy condition \eqref{2.1} and conditions \eqref{6}, \eqref{52.0}, and with two ends, one designated asymptotically flat and the other either asymptotically flat or asymptotically cylindrical. If the data possesses a stable axisymmetric marginally outer trapped surface with area $\mathcal{A}$, and the system of equations \eqref{30}, \eqref{37}, \eqref{49} admits a smooth solution $(u,Y^{\phi},f)$ satisfying the asymptotics \eqref{31}, \eqref{40}-\eqref{41.0}, \eqref{50}-\eqref{52}, then
\begin{equation}\label{4.2}
\frac{\mathcal{A}}{8\pi}\geq m^2-\sqrt{m^4-\mathcal{J}^2}
\end{equation}
and equality holds if and only if $(M,g,k)$ arises from an embedding into the extreme Kerr spacetime.
\end{theorem}

\section{The Equation for $f$}
\label{sec5} \setcounter{equation}{0}
\setcounter{section}{5}

Let $(M,g,k)$ be a simply connected, axisymmetric initial data set with two ends denoted $M_{end}^{\pm}$,
such that $M_{end}^{+}$ is asymptotically flat and $M_{end}^{-}$ is either asymptotically flat or asymptotically cylindrical. As discussed above there is a global Brill coordinate system $(\rho,\phi,z)$ in which the metric takes the form \eqref{9}. Here we make a change of coordinates to $(r,\phi,\theta)$, where $\rho=r\sin\theta$ and $z=r\cos\theta$. The metric may then be expressed by
\begin{equation}\label{75}
g=e^{-2U+2\alpha}(dr^{2}+r^{2}d\theta^{2})+e^{-2U}r^{2}\sin^{2}\theta(d\phi+A_{r}dr+A_{\theta}d\theta)^{2}.
\end{equation}
In addition to \eqref{10}-\eqref{12}, it is assumed that the initial data and $u$ satisfy the following asymptotics
\begin{equation}\label{76}
u=1+o_{2}(r^{-\frac{1}{2}}),\text{ }\text{ }\text{ }\text{ }Tr_{g}k=O_{2}(r^{-2-\varepsilon}),\text{ }\text{ }\text{ in }\text{ }\text{ }M_{end}^{+},
\end{equation}
for some $\varepsilon\in(0,1)$, and
\begin{equation}\label{77}
u=r^{2}+o_{2}(r^{\frac{5}{2}}),\text{ }\text{ }\text{ }\text{ }Tr_{g}k=O_{2}(r^{4}),\text{ }\text{ }\text{ in asymptotically flat }\text{ }\text{ }M_{end}^{-},
\end{equation}
\begin{equation}\label{78}
u=r+o_{2}(r^{\frac{3}{2}}),\text{ }\text{ }\text{ }\text{ }Tr_{g}k=O_{2}(r^{\frac{3}{2}}),\text{ }\text{ }\text{ in asymptotically cylindrical }\text{ }\text{ }M_{end}^{-}.
\end{equation}
Note that the asymptotics for $u$ are consistent with the choice \eqref{49} and the asymptotics
\eqref{10}-\eqref{12}, while the asymptotics for $Tr_{g}k$ are weaker in $M_{end}^{+}$, and stronger in asymptotically flat $M_{end}^{-}$, as compared with \eqref{3}.

In local coordinates, with the help of \eqref{35}, equation \eqref{37} is given by
\begin{equation}\label{79}
g^{ij}\left(\frac{u\nabla_{ij}f+u_{i}f_{j}+u_{j}f_{i}}{\sqrt{1+u^{2}|\nabla f|^{2}}}-k_{ij}\right)=0.
\end{equation}
Observe that this equation may also be expressed in divergence form
\begin{equation}\label{80}
div_{g}(u^{2}\nabla f)=u(Tr_{g}k)\sqrt{1+u^{2}|\nabla f|_{g}^{2}}.
\end{equation}
The desired asymptotics are
\begin{equation}\label{81}
|f|+r|\nabla f|_{g}+r^{2}|\nabla^{2} f|_{g}\leq cr^{-\varepsilon}\text{ }\text{ }\text{ in }\text{ }\text{ }M_{end}^{+},
\end{equation}
\begin{equation}\label{82}
r^{-1}|\nabla f|_{g}+r^{-2}|\nabla^{2} f|_{g}\leq c\text{ }\text{ }\text{ in asymptotically flat }\text{ }\text{ }M_{end}^{-},
\end{equation}
\begin{equation}\label{83}
|\nabla f|_{g}+|\nabla^{2} f|_{g}\leq cr^{\frac{1}{2}}\text{ }\text{ }\text{ in asymptotically cylindrical }\text{ }\text{ }M_{end}^{-},
\end{equation}
where $c$ is a constant.

We first solve \eqref{80} on the annular domain $\Omega_{\mathbf{r}}=\{(r,\phi,\theta)\mid \mathbf{r}^{-1}
<r<\mathbf{r}\}$, with zero Dirichlet boundary conditions
\begin{equation}\label{83.1}
\Delta_{g}f+2\nabla\log u\cdot\nabla f=u^{-1}(Tr_{g}k)\sqrt{1+u^{2}|\nabla f|_{g}^{2}}\text{ }\text{ }
\text{ in }\text{ }\text{ }\Omega_{\mathbf{r}}, \text{ }\text{ }f=0\text{ }\text{ }\text{ on }\text{ }\text{ }
\partial\Omega_{\mathbf{r}}.
\end{equation}

\begin{prop}\label{prop1}
Given initial data $(M,g,k)$ and a smooth positive function $u$, there exists a unique, smooth, uniformly bounded (independent of $\mathbf{r}$) solution $f$ of \eqref{83.1}.
\end{prop}

\begin{proof}
We will employ the continuity method. Thus, consider the family of equations
\begin{equation}\label{84}
div_{g}(u^{2}\nabla f_{s})=su(Tr_{g}k)\sqrt{1+u^{2}|\nabla f_{s}|_{g}^{2}}\text{ }\text{ }
\text{ in }\text{ }\text{ }\Omega_{\mathbf{r}}, \text{ }\text{ }f_{s}=0\text{ }\text{ }\text{ on }\text{ }\text{ }
\partial\Omega_{\mathbf{r}},
\end{equation}
and the set $\mathcal{S}=\{s\in[0,1]\mid \text{ there exists a unique solution $f_{s}\in C^{2,\beta}(\Omega_{\mathbf{r}})$ of \eqref{84}}\}$.  Clearly $\mathcal{S}$ is nonempty, since the case $s=0$ is solved by $f_{0}=0$. Moreover $\mathcal{S}$ is open by the implicit function theorem, since the linearized equation is strictly elliptic with no zeroth order term. It remains to show that $\mathcal{S}$ is closed. This
will be based on the construction of radial sub and super solutions.

Let $\overline{f}=\overline{f}(r)$ be a radial function. Then
\begin{align}\label{85}
\begin{split}
\Delta_{g}\overline{f}&=\frac{1}{\sqrt{\det g}}\partial_{i}\left(\sqrt{\det g}g^{ij}\partial_{j}\overline{f}\right)\\
&=\frac{1}{\sqrt{\det g}}\partial_{r}\left(\sqrt{\det g}g^{rr}\partial_{r}\overline{f}\right)\\
&=g^{rr}\left(\overline{f}''+\partial_{r}\log(r^{2}e^{-U})\overline{f}'\right),
\end{split}
\end{align}
where we have used
\begin{equation}\label{86}
g^{rr}=e^{2U-2\alpha},\text{ }\text{ }\text{ }g^{r\theta}=0,\text{ }\text{ }\text{ }
g^{\theta\theta}=r^{-2}e^{2U-2\alpha},\text{ }\text{ }\text{ }\det g=r^{4}e^{-6U+4\alpha}\sin^{2}\theta.
\end{equation}
Notice also that
\begin{equation}\label{87}
\nabla\log u\cdot\nabla\overline{f}=g^{rr}(\partial_{r}\log u)\overline{f}',\text{ }\text{ }\text{ }\text{ }
\nabla f\cdot\nabla\overline{f}=g^{rr}(\partial_{r}f)\overline{f}'.
\end{equation}
Let $\widetilde{u}=\widetilde{u}(r)>0$, $\widetilde{U}=\widetilde{U}(r)$, and $h=h(r)>0$ be radial functions such that
\begin{equation}\label{88}
u-\widetilde{u}=o_{1}(r^{\frac{7}{2}}),\text{ }\text{ }\text{ }\text{ }
e^{U}-e^{\widetilde{U}}=o_{1}(r^{\frac{7}{2}}),\text{ }\text{ }\text{ as }\text{ }\text{ }r\rightarrow 0,
\end{equation}
\begin{equation}\label{88.1}
\text{ }\text{ }\text{ }\text{ } u-\widetilde{u}=o_{1}(r^{-\frac{1}{2}}),\text{ }\text{ }\text{ }\text{ }
e^{U}-e^{\widetilde{U}}=o_{1}(r^{-\frac{1}{2}}),\text{ }\text{ }\text{ as }\text{ }\text{ }r\rightarrow\infty,
\end{equation}
\begin{equation}\label{88.2}
\frac{(g^{rr}u)^{-1}Tr_{g}k}{h}=o(1)\text{ }\text{ }\text{ as }\text{ }\text{ }
r\rightarrow 0\text{ }\text{ }\text{ and }\text{ }\text{ }r\rightarrow\infty,\text{ }\text{ }\text{ }\text{ }(g^{rr}u)^{-1}|Tr_{g}k|\leq h,
\end{equation}
and set $X_{s}=u|\nabla f_{s}|_{g}^{-1}\nabla f_{s}$, then
\begin{align}\label{89}
\begin{split}
&\Delta_{g}\overline{f}+2\nabla\log u\cdot\nabla\overline{f}-u^{-1}(Tr_{g}k)X_{s}\cdot\nabla\overline{f}\\
=&g^{rr}\left(\overline{f}''+\partial_{r}\log(r^{2}u^{2}e^{-U})\overline{f}'
-(Tr_{g}k)\frac{\partial_{r}f_{s}}{|\nabla f_{s}|_{g}}\overline{f}'\right)\\
=&g^{rr}\left(\overline{f}''+\partial_{r}\log(r^{2}\widetilde{u}^{2}e^{-\widetilde{U}})\overline{f}'
+\left(\partial_{r}\log\left(\frac{u^{2}e^{-U}}{\widetilde{u}^{2}e^{-\widetilde{U}}}\right)
-(Tr_{g}k)\frac{\partial_{r}f_{s}}{|\nabla f_{s}|_{g}}\right)\overline{f}'\right).
\end{split}
\end{align}

\begin{lemma}\label{lemma3}
Given initial data $(M,g,k)$ and a smooth positive function $u$ satisfying \eqref{10}-\eqref{12} and \eqref{76}-\eqref{78}, there exist negative $($positive$)$ radial sub $($super$)$ solutions $\underline{f}$ $(\overline{f})$ of \eqref{84}, which are independent of $s$ and uniformly bounded, and satisfy the asymptotics \eqref{81}-\eqref{83}.
\end{lemma}

\begin{proof}
The super solution $\overline{f}$ will be chosen as a solution of the ODE
\begin{equation}\label{90}
\overline{f}''+\left(\partial_{r}\log(r^{2}\widetilde{u}^{2}e^{-\widetilde{U}})-b\right)\overline{f}'
=-h,
\end{equation}
where $b=b(r)>0$ is a uniformly bounded function chosen so that
\begin{equation}\label{91}
\left|\partial_{r}\log\left(\frac{u^{2}e^{-U}}{\widetilde{u}^{2}e^{-\widetilde{U}}}\right)\right|
+|Tr_{g}k|\leq b,\text{ }\text{ }\text{ }\text{ }b(r)=O(r^{-\frac{3}{2}})\text{ }\text{ }\text{ as }\text{ }\text{ }r\rightarrow\infty.
\end{equation}
Consider the solution of \eqref{90} with $\lim_{r\rightarrow\infty}\overline{f}(r)=0$ and $\overline{f}'(0)=0$, that is
\begin{equation}\label{92}
\overline{f}(r)
=\int_{r}^{\infty}\left(r^{-2}\widetilde{u}^{-2}e^{\widetilde{U}}e^{\int_{0}^{r}b}
\int_{0}^{r}r^{2}\widetilde{u}^{2}e^{-\widetilde{U}}e^{-\int_{0}^{r}b}h\right).
\end{equation}
In order to see that \eqref{92} is indeed a super solution, use the fact that $\overline{f}'\leq 0$ along with \eqref{89}-\eqref{91} and the definition of $h$ to find
\begin{align}\label{93}
\begin{split}
\Delta_{g}\overline{f}+2\nabla\log u\cdot\nabla\overline{f}
=&g^{rr}\left(-h
+\left(b+\partial_{r}\log\left(\frac{u^{2}e^{-U}}{\widetilde{u}^{2}e^{-\widetilde{U}}}\right)\right)\overline{f}'\right)\\
\leq &g^{rr}(-h-|Tr_{g}k||\overline{f}'|)\\
\leq& -u^{-1}|Tr_{g}k|(1+u|\nabla \overline{f}|_{g})\\
\leq&-u^{-1}|Tr_{g}k|\sqrt{1+u^{2}|\nabla \overline{f}|_{g}^{2}}\\
\leq&su^{-1}(Tr_{g}k)\sqrt{1+u^{2}|\nabla \overline{f}|_{g}^{2}}.
\end{split}
\end{align}
Moreover, with the help of \eqref{76}-\eqref{78} and \eqref{88}-\eqref{88.2}, it is readily checked that $\overline{f}$ satisfies the desired asymptotics. Finally, analogous methods may also be used to construct a subsolution.
\end{proof}

We are now in a position to make uniform $C^{0}$ bounds for solutions of \eqref{84}. Thus, use the third line of
\eqref{93} to find that
\begin{align}\label{94}
\begin{split}
\Delta_{g}(f_{s}-\overline{f})+2\nabla\log u\cdot\nabla(f_{s}-\overline{f})
\geq& su^{-1}(Tr_{g}k)\sqrt{1+u^{2}|\nabla f_{s}|_{g}^{2}} +u^{-1}|Tr_{g}k|(1+u|\nabla \overline{f}|_{g})\\
\geq&su^{-1}|Tr_{g}k|\left(-1-X_{s}\cdot\nabla f_{s}+1+\frac{u\nabla f_{s}}{|\nabla f_{s}|_{g}}\cdot\nabla\overline{f}\right)\\
=&-su^{-1}|Tr_{g}k|X_{s}\cdot\nabla(f_{s}-\overline{f}).
\end{split}
\end{align}
Since $\overline{f}>0$, we have that $(f_{s}-\overline{f})|_{\partial\Omega_{\mathbf{r}}}<0$. Analogous but opposite inequalities hold when applying the sub solution. It now follows from a comparison argument that
\begin{equation}\label{95}
\underline{f}\leq f_{s}\leq\overline{f}.
\end{equation}

In order to produce higher order estimates, observe that the right-hand side of equation \eqref{84} is of linear growth in the first derivatives of $f_{s}$. Therefore, by slightly modifying standard techniques applied to the Dirichlet problem for linear elliptic equations, we obtain uniform $C^{2,\gamma}$ estimates for any $\gamma\in[0,1]$, with the help of \eqref{95}. It follows that $\mathcal{S}$ is closed, yielding a solution
$f=f_{1}\in C^{2,\beta}(\Omega_{\mathbf{r}})$ of \eqref{83.1}.  Elliptic regularity then implies that $f$ is
smooth.

In order to prove uniqueness, assume that two solutions $f_{1}$ and $f_{2}$ exist, then
\begin{align}\label{96}
\begin{split}
\Delta_{g}(f_{1}-f_{2})+2\nabla\log u\cdot\nabla(f_{1}-f_{2})
=& u^{-1}(Tr_{g}k)\left(\sqrt{1+u^{2}|\nabla f_{1}|^{2}} -\sqrt{1+u^{2}|\nabla f_{2}|^{2}}\right)\\
=&u(Tr_{g}k)\left(\frac{\nabla(f_{1}+f_{2})\cdot\nabla(f_{1}-f_{2})}
{\sqrt{1+u^{2}|\nabla f_{1}|^{2}} +\sqrt{1+u^{2}|\nabla f_{2}|^{2}}}\right).
\end{split}
\end{align}
Since $(f_{1}-f_{2})|_{\partial\Omega_{\mathbf{r}}}=0$, the maximum/minimum principle implies that
$f_{1}-f_{2}\equiv 0$.
\end{proof}

We are now ready to establish the main theorem of this section.

\begin{theorem}\label{thm5}
Given initial data $(M,g,k)$ and a smooth positive function $u$ satisfying \eqref{10}-\eqref{12} and \eqref{76}-\eqref{78}, there exists a smooth uniformly bounded solution $f$ of \eqref{79} satisfying the
asymptotics \eqref{81}-\eqref{83}.
\end{theorem}

Notice that the asymptotic behavior for $f$ on $M_{end}^{-}$ is not prescribed, but instead it is stated that the solution remains uniformly bounded and fall-off rates for its derivatives are given.

\begin{proof}
The same methods as in the proof of Proposition \ref{prop1}, yield uniform estimates for the solution $f_{\mathbf{r}}$ of \eqref{83.1} in $C^{2,\beta}_{loc}$. Thus, as $\mathbf{r}\rightarrow\infty$ a subsequence may be extracted which converges on compact subsets to a solution $f$ of \eqref{79}. By elliptic regularity, this solution is smooth. Moreover, in light of \eqref{95} we have the bound
\begin{equation}\label{97}
\underline{f}\leq f\leq\overline{f},
\end{equation}
showing that $f$ is uniformly bounded and has the appropriate asymptotics at $M_{end}^{+}$.

It remains to establish the asymptotics for derivatives. First consider the end $M_{end}^{+}$. Here we may follow the standard scaling argument in \cite{SchoenYau}. From the local $C^{2,\beta}$ estimates we have
\begin{equation}\label{98}
|f(x)|+|\nabla f(x)|_{g}+|\nabla^{2}f(x)|_{g}\leq c\text{ }\text{ }\text{ for }\text{ }\text{ }x\in M.
\end{equation}
Moreover it is also known from \eqref{97} that
\begin{equation}\label{99}
|f(x)|\leq c|x|^{-\varepsilon}\text{ }\text{ }\text{ as }\text{ }\text{ }|x|\rightarrow\infty\text{ }\text{ }\text{ in }\text{ }\text{ }M_{end}^{+}.
\end{equation}
Equation \eqref{83.1} may be viewed as the following linear equation
\begin{equation}\label{100}
\sum_{i,j=1}^{3}a^{ij}(x)\frac{\partial^{2}}{\partial x^{i}\partial x^{j}}f+\sum_{i=1}^{3}b^{i}(x)\frac{\partial}{\partial x^{i}}f=F(x),
\end{equation}
where
\begin{equation}\label{101}
a^{ij}=g^{ij},\text{ }\text{ }\text{ }\text{ }b^{i}=-g^{lm}\Gamma_{lm}^{i}+2(\log u)^{i},\text{ }\text{ }\text{ }\text{ }F=u^{-1}(Tr_{g}k)\sqrt{1+u^{2}|\nabla f|_{g}^{2}}.
\end{equation}
Now fix a point $x_{0}\in M_{end}^{+}$ and define coordinates $\overline{x} = (x-x_{0})/\sigma$, where $\sigma = |x_{0}|/2$. When written in these new coordinates, equation \eqref{100} becomes
\begin{equation}\label{102}
\sum_{i,j=1}^{3}a^{ij}(\overline{x})\frac{\partial^{2}}{\partial \overline{x}^{i}\partial \overline{x}^{j}}f+\sum_{i=1}^{3}\sigma b^{i}(\overline{x})\frac{\partial}{\partial \overline{x}^{i}}f=\sigma^{2}F(\overline{x})
\end{equation}
for $\overline{x}\in B_{1}(0)=\{|\overline{x}|<1\}$. Observe that $\partial_{\overline{x}^{i}}=\sigma\partial_{x^{i}}$, and therefore
\begin{equation}\label{103}
\left|\frac{\partial}{\partial\overline{x}^{l}}a^{ij}(\overline{x})\right|\leq \frac{c\sigma}{|x|^{\frac{3}{2}}}\leq\frac{c}{|x|^{\frac{1}{2}}},
\end{equation}
\begin{equation}\label{104}
\sigma|b^{i}(\overline{x})|+\sigma\left|\frac{\partial}{\partial\overline{x}^{l}}b^{i}(\overline{x})\right|\leq c\left(\frac{\sigma}{|x|^{\frac{3}{2}}}+\frac{\sigma^{2}}{|x|^{\frac{5}{2}}}\right)\leq \frac{c}{|x|^{\frac{1}{2}}},
\end{equation}
\begin{equation}\label{105}
\sigma^{2}|F(\overline{x})|+\sigma^{2}\left|\frac{\partial}{\partial\overline{x}^{l}}F(\overline{x})\right|\leq c\left(\frac{\sigma^{2}}{|x|^{2+\varepsilon}}+\frac{\sigma^{3}}{|x|^{3+\varepsilon}}\right)
\leq \frac{c}{|x|^{\varepsilon}},
\end{equation}
where $c$ represents a constant depending on the initial data and $u$.
These estimates show that we have control of the coefficients, and right-hand side of \eqref{102}, in $C^{0,\beta}$. Thus the interior Schauder estimates apply to yield
\begin{align}\label{106}
\begin{split}
|\overline{\partial} f(\overline{x})|+|\overline{\partial}^{2}f(\overline{x})|&\leq c(\sigma^{2}|F|_{C^{0,\alpha}(B_{1}(0))}+|f|_{C^{0}(B_{1}(0))})\\
& \leq c(\sigma^{2}|F|_{C^{1}(B_{1}(0))}+|f|_{C^{0}(B_{1}(0))})
\end{split}
\end{align}
for $\overline{x}\in B_{1/2}(0)$. It follows that
\begin{align}\label{107}
\begin{split}
\sigma|\partial f(x)|+\sigma^{2}|\partial^{2}f(x)|&\leq c(\sigma^{2}|F|_{C^{1}(B_{1}(0))}+|f|_{C^{0}(B_{1}(0))})\\
&\leq c|x|^{-\varepsilon}
\end{split}
\end{align}
for $x\in B_{\sigma/2}(x_{0})$, and hence the desired estimate holds
\begin{equation}\label{108}
|f(x)|+|x||\nabla f(x)|_{g}+|x|^{2}|\nabla^{2}f(x)|_{g}\leq c|x|^{-\varepsilon}
\text{ }\text{ }\text{ for }\text{ }\text{ }x\in M_{end}^{+}.
\end{equation}
Derivative estimates in the remaining end will be divided into two cases.\medskip

\textit{Case 1: $M_{end}^{-}$ is asymptotically flat.}
By performing an inversion $x\mapsto\frac{x}{|x|^{2}}$ near the origin in Brill coordinates, asymptotically flat
coordinates are obtained in $M_{end}^{-}$. We may now apply the same scaling argument as above. However here, $u\sim |x|^{-2}$ and $Tr_{g}k\sim |x|^{-4}$ so that
\begin{equation}\label{109}
\sigma|b^{i}(\overline{x})|+\sigma\left|\frac{\partial}{\partial\overline{x}^{l}}b^{i}(\overline{x})\right|\leq c\left(\frac{\sigma}{|x|}+\frac{\sigma^{2}}{|x|^{2}}\right)\leq c,
\end{equation}
\begin{equation}\label{110}
\sigma^{2}|F(\overline{x})|+\sigma^{2}\left|\frac{\partial}{\partial\overline{x}^{l}}F(\overline{x})\right|\leq c\left(\frac{\sigma^{2}}{|x|^{2}}+\frac{\sigma^{3}}{|x|^{3}}\right)
\leq c.
\end{equation}
It follows that
\begin{equation}\label{111}
\sigma|\partial f(x)|+\sigma^{2}|\partial^{2}f(x)|\leq c(\sigma^{2}|F|_{C^{1}(B_{1}(0))}+|f|_{B_{1}(0)})\leq c
\text{ }\text{ }\text{ for }\text{ }\text{ }x\in B_{\sigma/2}(x_{0}),
\end{equation}
and hence the desired estimate holds
\begin{equation}\label{112}
|x||\nabla f(x)|_{g}+|x|^{2}|\nabla^{2}f(x)|_{g}\leq c
\text{ }\text{ }\text{ for }\text{ }\text{ }x\in M_{end}^{-}.
\end{equation}

\textit{Case 2: $M_{end}^{-}$ is asymptotically cylindrical.}
According to the asymptotics \eqref{12}
\begin{equation}\label{113}
g= r^{-2}dr^{2}+g_{S^{2}}(\phi,\theta)+G(r,\phi,\theta),
\end{equation}
where $g_{S^{2}}$ is the round metric on the 2-sphere and the remainder satisfies
\begin{equation}\label{114}
G_{rr}=o_{2}(r^{-\frac{3}{2}}),\text{ }\text{ }\text{ }\text{ }G_{ij}=o_{2}(r^{\frac{1}{2}})
\text{ }\text{ }\text{ for }\text{ }\text{ }i,j\neq r.
\end{equation}
A further change of coordinates $\tau=\log r$, in which $d\tau=r^{-1}dr$, displays the canonical cylindrical form of the metric
\begin{equation}\label{115}
g= d\tau^{2}+g_{S^{2}}(\phi,\theta)+G(\tau,\phi,\theta),
\end{equation}
with
\begin{equation}\label{116}
|G_{ij}|+|\partial G_{ij}|+|\partial^{2}G_{ij}|=o(e^{\frac{1}{2}\tau})
\text{ }\text{ }\text{ for all }\text{ }\text{ }i,j.
\end{equation}

The interior Schauder estimates imply that
\begin{equation}\label{117}
|f|+|\nabla f|_{g}+|\nabla^{2}f|_{g}\leq c,
\end{equation}
and hence
\begin{equation}\label{118}
|\partial_{\tau} f|+|\partial_{\theta}f|+|\partial_{\tau}^{2}f|
+|\partial_{\tau}\partial_{\theta}f|+|\partial_{\theta}^{2}f|\leq c.
\end{equation}
It then follows, with the help of \eqref{78}, that
\begin{equation}\label{119}
\nabla\log u\cdot\nabla f=\partial_{\tau}f+O(e^{\frac{1}{2}\tau})
\end{equation}
and
\begin{equation}\label{120}
\Delta_{g}f=\partial_{\tau}^{2}f+\Delta_{S^{2}}f+O(e^{\frac{1}{2}\tau}).
\end{equation}

We now obtain derivative estimates using a separation of variables argument.
Let $\{\psi_{i}\}_{i=0}^{\infty}\subset L^{2}(S^{2})$ be an complete orthonormal set of eigenfunctions for $\Delta_{S^{2}}$, and let $\lambda_{i}=i(i+1)$ denote the corresponding eigenvalues. Since the eigenfunctions are complete, we may write
\begin{equation}\label{121}
f(\tau,\phi,\theta)=\sum_{i=0}^{\infty}d_{i}(\tau)\psi_{i}(\phi,\theta).
\end{equation}
Inserting this into equation \eqref{83.1} produces
\begin{equation}\label{122}
\sum_{i=0}^{\infty}(d_{i}''+2d_{i}'-\lambda_{i}d_{i})\psi_{i}=P,
\end{equation}
and thus
\begin{equation}\label{123}
d_{i}''+2d_{i}'-\lambda_{i}d_{i}=P_{i}(\tau):=\int_{S^{2}}P(\tau,\phi,\theta)\psi_{i}(\phi,\theta).
\end{equation}
The general solution to this ODE is given by the method of variation of parameters
\begin{equation}\label{124}
d_{i}(\tau)=(c_{1i}+p_{1i}(\tau))e^{\left(-1-\sqrt{1+\lambda_{i}}\right)\tau}+(c_{2i}+p_{2i}(\tau))e^{\left(-1+\sqrt{1+\lambda_{i}}\right)\tau},
\end{equation}
where $c_{1i}$ and $c_{2i}$ are constants and
\begin{equation}\label{125}
p_{1i}(\tau)=-\frac{1}{2\sqrt{1+\lambda_{i}}}\int_{-\infty}^{\tau}\!\!e^{\left(1+\sqrt{1+\lambda_{i}}\right)\tau}P_{i}(\tau)d\tau,\text{ }\text{ }
p_{2i}(\tau)=\frac{1}{2\sqrt{1+\lambda_{i}}}\int_{\tau_{0}}^{\tau}e^{\left(1-\sqrt{1+\lambda_{i}}\right)\tau}P_{i}(\tau)d\tau,
\end{equation}
for some $\tau_{0}$. Note that the boundedness of $f$ implies that $c_{1i}=0$. Moreover
\begin{equation}\label{126}
d_{0}(\tau)=c_{20}+p_{20}(\tau)+e^{-2\tau}p_{10}(\tau)=c_{20}+O(e^{\frac{1}{2}\tau}),
\end{equation}
and $P=O(e^{\frac{1}{2}\tau})$ implies that
\begin{equation}\label{127}
d_{i}(\tau)=O\left(\max\{e^{\frac{1}{2}\tau},e^{\left(-1+\sqrt{1+\lambda_{i}}\right)\tau}\}\right),\text{ }\text{ }\text{ }\text{ }i\geq 1.
\end{equation}
Since $-1+\sqrt{1+\lambda_{i}}>7/10$ for $i\geq 1$, it follows that
\begin{equation}\label{128}
|\partial_{\tau}f|+|\partial_{\theta}f|=O(e^{\frac{1}{2}\tau}),
\end{equation}
and hence
\begin{equation}\label{129}
|\nabla f|_{g}=O(r^{\frac{1}{2}}).
\end{equation}
By differentiating the expansion, similar considerations yield
\begin{equation}\label{130}
|\nabla^{2} f|_{g}=O(r^{\frac{1}{2}}).
\end{equation}
\end{proof}

\section{The Equation for $Y^{\phi}$}
\label{sec6} \setcounter{equation}{0}
\setcounter{section}{6}

Let $(M,g,k)$, $u$, and $f$ be as in the previous section, although $f$ is not required to satisfy an equation
here. In particular, $u$ and $f$ satisfy the asymptotics \eqref{76}-\eqref{78} and \eqref{81}-\eqref{83}.
In this section we solve the equation
\begin{equation}\label{131}
div_{\overline{g}}\text{ }\!\overline{k}(\eta)=0\text{ }\text{ }\text{ on }\text{ }\text{ }M,
\end{equation}
for solutions satisfying the following asymptotics
\begin{equation}\label{132}
Y^{\phi}=-\frac{2\mathcal{J}}{r^{3}}+o_{2}(r^{-\frac{7}{2}})\text{ }\text{ }\text{ in }\text{ }\text{ }M_{end}^{+},
\end{equation}
\begin{equation}\label{133}
Y^{\phi}=\mathcal{Y}+O_{2}(r^{5})\text{ }\text{ }\text{ in asymptotically flat }\text{ }\text{ }M_{end}^{-},
\end{equation}
\begin{equation}\label{134}
Y^{\phi}=\mathcal{Y}+O_{2}(r)\text{ }\text{ }\text{ in asymptotically cylindrical }\text{ }\text{ }M_{end}^{-},
\end{equation}
where $\mathcal{J}$ and $\mathcal{Y}$ are constants. In order to obtain a unique solution the value of $\mathcal{J}$ will be prescribed, and in this case the value of $\mathcal{Y}$ is determined by $\mathcal{J}$ and the initial data. Note that these asymptotics are consistent with those of the (sole component of the) shift vector field $Y_{EK}$ in the extreme Kerr spacetime, as is shown in Appendix B, Section \ref{sec8}.

The equation \eqref{131} may be expressed in a more revealing way as follows
\begin{equation}\label{135}
\Delta_{\overline{g}}Y^{\phi}+\overline{\nabla}\log(u^{-1}g_{\phi\phi})\cdot\overline{\nabla}Y^{\phi}=0.
\end{equation}
Since the metric $\overline{g}$ depends on $Y^{\phi}$, it appears that this should be a nonlinear equation.
However certain cancelations occur, and it turns out that when expressed in terms of the metric $g$, the equation is linear elliptic
\begin{align}\label{136}
\begin{split}
0=& \left(g^{ij}-\frac{u^{2}f^{i}f^{j}}{1+u^{2}|\nabla f|_{g}^{2}}\right)\left(\nabla_{ij}Y^{\phi}-\frac{u\pi_{ij}f^{l}}{\sqrt{1+u^{2}|\nabla f|_{g}^{2}}}\partial_{l}Y^{\phi}\right)\\
&+\left(g^{ij}-\frac{u^{2}f^{i}f^{j}}{1+u^{2}|\nabla f|_{g}^{2}}\right)\left(\partial_{i}\log g_{\phi\phi}-\frac{\partial_{i}\log u}{1+u^{2}|\nabla f|_{g}^{2}}\right)\partial_{j}Y^{\phi},
\end{split}
\end{align}
where
\begin{equation}\label{137}
\pi_{ij}=\frac{u}{\sqrt{1+u^{2}|\nabla f|_{g}^{2}}}\left(\nabla_{ij}f+(\log u)_{i}f_{j}+(\log u)_{j}f_{i}+\frac{g_{i\phi}Y^{\phi}_{,j}+g_{j\phi}Y^{\phi}_{,i}}{2u^{2}}\right)
\end{equation}
is the second fundamental form of the graph $\overline{M}=\{t=f(x)\}$ in the Lorentzian setting. It is important to note that the linear character of the equation, arises from the fact that
\begin{equation}\label{138}
\left(g^{ij}-\frac{u^{2}f^{i}f^{j}}{1+u^{2}|\nabla f|_{g}^{2}}\right)\pi_{ij}
\end{equation}
does not depend on $Y^{\phi}$. The equivalence of the three equations \eqref{131}, \eqref{135}, and \eqref{136}
will be proved in Appendix B.\medskip

We now prove existence and uniqueness. The first task is to construct a radial function $Y_{0}=Y_{0}(r)$ which
is an approximate solution in the asymptotic regions. Following the same procedure as in \eqref{89} yields
\begin{equation}\label{139}
\Delta_{g}Y_{0}+\nabla\log(u^{-1}g_{\phi\phi})\cdot\nabla Y_{0}
=g^{rr}\left(Y_{0}''+\partial_{r}\log(r^{4}\widetilde{u}^{-1}e^{-3\widetilde{U}})Y_{0}'
+\partial_{r}\log\left(\frac{u^{-1}e^{-3U}}{\widetilde{u}^{-1}e^{-3\widetilde{U}}}\right)Y_{0}'\right),
\end{equation}
where $\widetilde{u}$ and $e^{\widetilde{U}}$ are defined in \eqref{88} and \eqref{88.1}, and where we have used
$g_{\phi\phi}=e^{-2U}r^{2}\sin^{2}\theta$. The remaining terms may be computed in a similar way. If $L$ denotes
the differential operator on the right-hand side of \eqref{136}, then
\begin{eqnarray}\label{140}
%\begin{split}
LY_{0}&=& \left(g^{rr}-\frac{u^{2}(f^{r})^{2}}{1+u^{2}|\nabla f|_{g}^{2}}\right)Y_{0}''
+g^{rr}\partial_{r}\log(r^{4}\widetilde{u}^{-1}e^{-3\widetilde{U}})Y_{0}'\nonumber\\
& &+\left(g^{rr}\partial_{r}\log\left(\frac{u^{-1}e^{-3U}}{\widetilde{u}^{-1}e^{-3\widetilde{U}}}\right)
+\frac{u^{2}f^{i}f^{j}\Gamma_{ij}^{r}}{1+u^{2}|\nabla f|_{g}^{2}}
-\frac{u^{2}f^{r}f^{i}\partial_{i}\log g_{\phi\phi}}{1+u^{2}|\nabla f|_{g}^{2}}\right)Y_{0}'\\
& &+\left(\frac{u^{2}f^{r}f^{i}\partial_{i}\log u}{(1+u^{2}|\nabla f|_{g}^{2})^{2}}
+\frac{g^{rr}u^{2}|\nabla f|_{g}^{2}\partial_{r}\log u}{1+u^{2}|\nabla f|_{g}^{2}}
-\left(g^{ij}-\frac{u^{2}f^{i}f^{j}}{1+u^{2}|\nabla f|_{g}^{2}}\right)\frac{u\pi_{ij}f^{r}}{\sqrt{1+u^{2}|\nabla f|_{g}^{2}}}\right)Y_{0}'.\nonumber
%\end{split}
\end{eqnarray}
Note that the term $f^{i}\partial_{i}\log g_{\phi\phi}$ appears to cause a problem when $\theta=0,\pi$, since
it involves $f^{\theta}\partial_{\theta}\log\sin\theta=g^{\theta\theta}f_{\theta}\partial_{\theta}\log\sin\theta$
which could blow-up at those values of $\theta$. However, for axisymmetric smooth functions $f$, one must have
$\partial_{\theta}f=0$ on the axis of rotation.
%To see this, consider this derivative in
%xyz-coordinates:
%\begin{equation*}
%\partial_{\theta}f=r\cos\phi\cos\theta \partial_{x}f
%                           +r\sin\phi\cos\theta\partial_{y}f
%                            -r\sin\theta\partial_{z}f.
%\end{equation*}
%Due to axisymmetry $0=\partial_{\phi}f$ so that %$\sin\phi\partial_{x}f=\cos\phi\partial_{y}f$.
%From this we see that $\partial_{y}f=0$ when $\phi=0,\pi$ ie
%On x-axis. Since this is independent of z, we have that
%$\partial_{y}f=0$ on xz-plane, which includes the whole z-axis. Now us
%the above equation to rewrite
%\begin{equation*}
%\partial_{\theta}f=\frac{r\cos\theta}{\cos\phi}\partial_{y}f
%                           -r\sin\theta\partial_{z}f.
%\end{equation*}
%which then equals zero on the z-axis.]]]]
Observe that equation \eqref{140} may be written more simply as
\begin{equation}\label{141}
LY_{0}=\left(g^{rr}-\frac{u^{2}(f^{r})^{2}}{1+u^{2}|\nabla f|_{g}^{2}}\right)\left(Y_{0}''
+\left(\partial_{r}\log(r^{4}\widetilde{u}^{-1}e^{-3\widetilde{U}})+B\right)Y_{0}'\right),
\end{equation}
for some function $B$. With the help of \eqref{76}-\eqref{78}, \eqref{81}-\eqref{83}, \eqref{88}, \eqref{88.1}, and the calculation of Christoffel symbols in Appendix D, it can be shown that this function has the property that
\begin{equation}\label{142}
B=O(r^{-\frac{3}{2}})\text{ }\text{ }\text{ as }\text{ }\text{ }r\rightarrow\infty,
\text{ }\text{ }\text{ }\text{ }\text{ }\text{ }\text{ }\text{ }B=O(r^{\frac{1}{2}})\text{ }\text{ }\text{ as }\text{ }\text{ }r\rightarrow 0.
\end{equation}
We are motivated to choose $Y_{0}$ as the solution to the ODE
\begin{equation}\label{143}
Y_{0}''+\partial_{r}\log(r^{4}\widetilde{u}^{-1}e^{-3\widetilde{U}})Y_{0}'=0
\end{equation}
which satisfies the asymptotics \eqref{132}-\eqref{134}, namely
\begin{equation}\label{144}
Y_{0}(r)=c\int_{r}^{\infty}r^{-4}\widetilde{u}e^{3\widetilde{U}}
\end{equation}
where the constant $c$ is chosen in order to realize the desired $r^{-3}$-fall-off rate in \eqref{132}.

The desired solution of \eqref{136} will be constructed in the form $Y^{\phi}=Y_{0}+Y_{1}$, where $Y_{1}$ solves the equation
\begin{equation}\label{145}
LY_{1}=-LY_{0},
\end{equation}
and has the same asymptotics as in \eqref{132}-\eqref{134} with $\mathcal{J}=0$.

\begin{theorem}\label{thm6}
Given initial data $(M,g,k)$ and smooth functions $u>0$, $f$ satisfying \eqref{10}-\eqref{12}, \eqref{76}-\eqref{78}, and \eqref{81}-\eqref{83}, there exists a unique, smooth, uniformly bounded solution $Y^{\phi}$ of \eqref{136} satisfying the asymptotics \eqref{132}-\eqref{134}.
\end{theorem}

\begin{proof}
The first task is to construct radial sub and super solutions $\underline{Y}_{1}$ and $\overline{Y}_{1}$. The super solution will be chosen as a solution of the ODE
\begin{equation}\label{149}
\overline{Y}_{1}''+\left(\partial_{r}\log(r^{4}\widetilde{u}^{-1}e^{-3\widetilde{U}})-\overline{B}\right)\overline{Y}_{1}'
=-\overline{h},
\end{equation}
where the radial functions $\overline{B}>0$ and $\overline{h}>0$ are chosen so that
\begin{equation}\label{150}
|B|\leq \overline{B},\text{ }\text{ }\text{ }\text{ }\text{ }\text{ }
\left(g^{rr}-\frac{u^{2}(f^{r})^{2}}{1+u^{2}|\nabla f|_{g}^{2}}\right)^{-1}|LY_{0}|\leq\overline{h},
\end{equation}
with $\overline{B}$ satisfying the asymptotics \eqref{142}, whereas
\begin{equation}\label{151}
\overline{h}=O(r^{-\frac{11}{2}})\text{ }\text{ }\text{ in }\text{ }\text{ }M_{end}^{+},
\end{equation}
\begin{equation}\label{152}
\overline{h}=O(r^{4})\text{ }\text{ }\text{ in asymptotically flat }\text{ }\text{ }M_{end}^{-},
\end{equation}
\begin{equation}\label{153}
\overline{h}=O(1)\text{ }\text{ }\text{ in asymptotically cylindrical }\text{ }\text{ }M_{end}^{-}.
\end{equation}
To see that $\overline{Y}_{1}$ is a super solution, observe that \eqref{141} and \eqref{150} imply
\begin{align}\label{154}
\begin{split}
L\overline{Y}_{1}&=\left(g^{rr}-\frac{u^{2}(f^{r})^{2}}{1+u^{2}|\nabla f|_{g}^{2}}\right)\left(\overline{Y}_{1}''
+\left(\partial_{r}\log(r^{4}\widetilde{u}^{-1}e^{-3\widetilde{U}})+B\right)\overline{Y}_{1}'\right)\\
&=\left(g^{rr}-\frac{u^{2}(f^{r})^{2}}{1+u^{2}|\nabla f|_{g}^{2}}\right)\left(-\overline{h}+\left(\overline{B}+B\right)\overline{Y}_{1}'\right)\\
&\leq-LY_{0},
\end{split}
\end{align}
if $\overline{Y}_{1}'\leq 0$. Thus we choose
\begin{equation}\label{155}
\overline{Y}_{1}
=\int_{r}^{\infty}e^{\int_{0}^{r}\overline{B}}r^{-4}\widetilde{u}e^{3\widetilde{U}}
\int_{0}^{r}e^{-\int_{0}^{r}\overline{B}}r^{4}\widetilde{u}^{-1}e^{-3\widetilde{U}}\overline{h}.
\end{equation}
Clearly $\overline{Y}_{1}$ is positive, nonincreasing, and satisfies the asymptotics \eqref{132}-\eqref{134} with $\mathcal{J}=0$. Similar methods yield a negative,
nondecreasing subsolution $\underline{Y}_{1}$ satisfying the same asymptotics.

Since \eqref{145} is linear elliptic without a zeroth order term, there exists a unique smooth solution $Y_{1}^{\mathbf{r}}$, with zero Dirichlet boundary conditions, on the annular domain $\Omega_{\mathbf{r}}=\{(r,\phi,\theta)\mid \mathbf{r}^{-1}
<r<\mathbf{r}\}$. By the maximum/minimum principle
\begin{equation}\label{156}
\underline{Y}_{1}<Y_{1}^{\mathbf{r}}<\overline{Y}_{1}.
\end{equation}
The interior Schauder estimates now yield uniform $C^{2,\beta}_{loc}$ bounds, and thus after passing to
a subsequence, we obtain a $C^{2}$ solution $Y_{1}$ on $M$ as $\mathbf{r}\rightarrow\infty$. This solution is
smooth by elliptic regularity and satisfies the estimate
\begin{equation}\label{157}
\underline{Y}_{1}<Y_{1}<\overline{Y}_{1}.
\end{equation}
Asymptotics for the derivatives may be established as in the proof of Theorem \ref{thm5}. Namely, in $M_{end}^{+}$ a scaling argument is employed, and if $\mathcal{Y}_{1}$ denotes $\lim_{r\rightarrow 0}Y_{1}$ then the same is true for asymptotically flat $M_{end}^{-}$, although here one must apply the method to $Y_{1}-\mathcal{Y}_{1}$. Moreover, asymptotics for the derivatives in asymptotically cylindrical $M_{end}^{-}$ may be obtained through the use of eigenfunction expansions.

We now have a solution $Y^{\phi}=Y_{0}+Y_{1}$ of \eqref{136}, satisfying the asymptotics \eqref{132}-\eqref{134}. It remains to prove uniqueness. Thus, consider a solution of $LZ=0$ having the asymptotics \eqref{132}-\eqref{134} with $\mathcal{J}=0$. We must show that $Z\equiv 0$. A direct calculation shows that $\mathcal{K}:=u^{-1}g_{\phi\phi}\sqrt{1+u^{2}|\nabla f|_{g}^{2}}\in Ker_{g} L^{*}$, the kernel of the formal $L^{2}(M,g)$ adjoint of $L$. Alternatively this may be proved by observing, the immediately apparent fact, that $u^{-1}g_{\phi\phi}$ lies in the $L^{2}(\overline{M},\overline{g})$-kernel of the adjoint of the operator \eqref{135}, and using \eqref{29} as well as \eqref{807.1}. Upon multiplying the equation $LZ=0$ through by $Z\mathcal{K}$ and integrating by parts, several boundary terms cancel to yield
\begin{align}\label{158}
\begin{split}
0=&\int_{M}\mathcal{K}\left(g^{ij}-\frac{u^{2}f^{i}f^{j}}{1+u^{2}|\nabla f|_{g}^{2}}\right)\nabla_{i}Z\nabla_{j}Z\\
&+\lim_{r\rightarrow\infty}\int_{\partial B(r)}Z\mathcal{K}\left(g^{ij}-\frac{u^{2}f^{i}f^{j}}{1+u^{2}|\nabla f|_{g}^{2}}\right)\nu_{i}\nabla_{j}Z
-\lim_{r\rightarrow 0}\int_{\partial B(r)}Z\mathcal{K}\left(g^{ij}-\frac{u^{2}f^{i}f^{j}}{1+u^{2}|\nabla f|_{g}^{2}}\right)\nu_{i}\nabla_{j}Z,
\end{split}
\end{align}
where $\nu$ is the unit normal to $\partial B(r)$ pointing towards $M_{end}^{+}$. Note that the boundary term
at $M_{end}^{+}$ clearly vanishes, while the boundary term at $M_{end}^{-}$ becomes
\begin{equation}\label{159}
\lim_{r\rightarrow 0}\mathcal{Y}_{0}\int_{\partial B(r)}\mathcal{K}\left(g^{ij}-\frac{u^{2}f^{i}f^{j}}{1+u^{2}|\nabla f|_{g}^{2}}\right)\nu_{i}\nabla_{j}Z,
\end{equation}
where $\mathcal{Y}_{0}=\lim_{r\rightarrow 0}Y$. Now multiply the equation $LZ=0$ by $\mathcal{K}$ alone to find
\begin{equation}\label{160}
0=\lim_{r\rightarrow\infty}\int_{\partial B(r)}\mathcal{K}\left(g^{ij}-\frac{u^{2}f^{i}f^{j}}{1+u^{2}|\nabla f|_{g}^{2}}\right)\nu_{i}\nabla_{j}Z
-\lim_{r\rightarrow 0}\int_{\partial B(r)}\mathcal{K}\left(g^{ij}-\frac{u^{2}f^{i}f^{j}}{1+u^{2}|\nabla f|_{g}^{2}}\right)\nu_{i}\nabla_{j}Z.
\end{equation}
Again, the boundary term at $M_{end}^{+}$ vanishes, and hence the boundary term at $M_{end}^{-}$ vanishes, which implies that \eqref{159} vanishes. It then follows from \eqref{158} that $Z\equiv 0$.
\end{proof}

\section{Appendix A: The Scalar Curvature Formula}
\label{sec7} \setcounter{equation}{0}
\setcounter{section}{7}

The purpose of this section is to prove Theorem \ref{thm1}. We will follow the ideas in \cite{BrayKhuri2}. In particular, $(M,g)$ will be viewed as a graph $\{t=f(x)\}$ in the Lorentzian setting \eqref{21}. An alternative approach, in which the calculations are made by viewing $(\overline{M},\overline{g})$ as the same graph in the
Riemannian setting \eqref{17} is possible, although not pursued here; in the static case this approach was carried out in \cite{BrayKhuri1}.

Let $(\overline{M}\times\mathbb{R},\widetilde{g})$ denote the Lorentzian stationary spacetime with
\begin{equation} \label{170}
\widetilde{g}=\overline{g}-2Y_{i}dx^{i}dt-\varphi dt^{2}.
\end{equation}
The induced metric on the graph and its inverse are given by
\begin{equation} \label{171}
g_{ij}= \bg_{ij} - f_{i}Y_{j} - f_{j}Y_{i} - \varphi f_{i}f_{j},
\text{ }\text{ }\text{ }\text{ }\text{ }g^{ij}= \bg^{ij} - u^{-2}\by^{i}\by^{j} + w^{i}w^{j},
\end{equation}
where
\begin{equation} \label{172}
w^{i}= \frac{u\overline{g}^{ij}f_{j} + u^{-1}\by^{i}}{\sqrt{\volbarg}},
\text{ }\text{ }\text{ }\text{ }\text{ }u^{2} = \varphi + | \by |_{\bg}^{2}.
\end{equation}
The vector $w$ is in fact the spatial component of the unit normal to the graph. In particular, the unit normal
to the $t=0$ slice and the graph are given given, respectively, by
\begin{equation} \label{173}
n = \frac{\partial_{t}+\by}{u}, \text{ }\text{ }\text{ }\text{ }\text{ } N= \frac{u\barna f + n}{\sqrt{\volbarg}}.
\end{equation}
Note that $\partial_{t}$ is a killing vector on $(\overline{M}\times\mathbb{R},\widetilde{g})$. Thus there is an obvious one-to-one correspondence between $M=\{ t=f(x)\}$ and $\overline{M}=\{t=0\}$. In that sense,
decomposing $n$ into its normal and tangential components with respect to the graph, and decomposing $N$ into its normal and tangential components with respect to the $t=0$ slice, yields
\begin{equation} \label{173}
n = \sqrt{1+u^{2}|\nabla f|_{g}^{2}} N-u\nabla f, \text{ }\text{ }\text{ }\text{ }\text{ }
N= \sqrt{1+u^{2}|\nabla f|_{g}^{2}}(n+u\overline{\nabla}f).
\end{equation}
Here the identity \eqref{1001} was used.

Let $G=Ric_{\widetilde{g}}-\frac{1}{2}R_{\widetilde{g}}\widetilde{g}$ denote the Einstein tensor.
Using \eqref{173} and the Gauss-Codazzi relations $G(N,n)$ may be computed in two different ways, namely
\begin{align} \label{0011}
\begin{split}
G(N,n)
&= \sqrt{\volg} (G(n,n) + G(u\barna f, n))  \\
&= \sqrt{\volg}  \left[(\overline{R}+(Tr_{\bg}\bk)^{2}- | \bk |_{\bg}^{2})/2
 + \bdiv(\bk-(Tr_{\bg}\bk)\bg)(u\barna f)\right],
\end{split}
\end{align}
and
\begin{align} \label{0012}
\begin{split}
G(N,n)
&= \sqrt{\volg} G(N,N) - G(N, u\nabla f)  \\
&= \sqrt{\volg} (R +(Tr_{g}\pi)^{2}- | \pi |_{g}^{2})/2
 - div_{g}(\pi-(Tr_{g}\pi)g)(u \nabla f),
\end{split}
\end{align}
where $\overline{k}$ and $\pi$ denote the second fundamental forms of $\overline{M}$ and $M$, respectively.
It follows that
\begin{align} \label{0013}
\begin{split}
&\overline{R}+ (Tr_{\bg}\bk)^{2}- | \bk |_{\bg}^{2} + 2\bdiv(\bk-(Tr_{\bg}\bk)\bg)(u \barna f)  \\
=& R +(Tr_{g}\pi)^{2}- | \pi |_{g}^{2} - 2 div_{g}(\pi-(Tr_{g}\pi)g)(v)
\end{split}
\end{align}
where
\begin{equation} \label{0014}
v=\frac{u \nabla f}{ \sqrt{\volg}}.
\end{equation}
The energy density and momentum densities for the initial data $(M,g,k)$ are
\begin{align} \label{0015}
\begin{split}
8\pi  \mu = G(N,N)= (R+ (Tr_{g}k)^{2}- | k |_{g}^{2})/2 \\
8 \pi  J(\cdot) = G(N, \cdot) = div_{g}(k-(Tr_{g}k)g)(\cdot).
\end{split}
\end{align}
Thus, combining this with the fact that $Tr_{\bg}\bk = 0$ (Lemma \ref{lemma1}) produces
\begin{align} \label{0016}
\begin{split}
\overline{R}- | \bk |_{\bg}^{2}
+ 2\bdiv \bk (u\barna f)
=& 16 \pi(\mu-J(v)) - | \pi |_{g}^{2} +  | k |_{g}^{2}
-2div_{g}(\pi)(v)+ 2div_{g}(k)(v)  \\
& + (Tr_{g}\pi)^{2}-(Tr_{g}k)^{2} +2v(Tr_{g}\pi- Tr_{g}k).
\end{split}
\end{align}

In what follows, several identities involving significant computation will be proven, which when combined with \eqref{0016} will yield the desired result. The first task is to calculate the second fundamental form $\pi$.
Below, $i,j,l$ etc. will denote indices associated with local coordinates $x^{i}$ on $\overline{M}$, and
$\widetilde{\Gamma}_{ij}^{l}$, $\overline{\Gamma}_{ij}^{l}$, $\Gamma_{ij}^{l}$ will represent Christoffel symbols for the metrics $\widetilde{g}$, $\overline{g}$, and $g$, respectively. Observe that the inverse metric is
\begin{equation} \label{01}
\tg^{tt}= - u^{-2}, \text{ }\text{ }\text{ }\text{ } \tg^{ti}=-u^{-2}\by^{i}, \text{ }\text{ }\text{ }\text{ } \tg^{ij}=\bg^{ij} - u^{-2}\by^{i}\by^{j},
\end{equation}
from which we find
\begin{equation} \label{02}
\tilG^{t}_{tt}=0,  \text{ }\text{ }\text{ }\text{ }
\tilG^{i}_{tt}= \frac{1}{2}\bg^{ij}\varphi_{j}= \frac{1}{2}\varphi^{i}, \text{ }\text{ }\text{ }\text{ }
\tilG^{t}_{ij}= u^{-1} \bk_{ij},\text{ }\text{ }\text{ }\text{ }
\tilG^{k}_{ij} = \barG^{l}_{ij} + u^{-1}\by^{l} \bk_{ij},
\end{equation}
\begin{equation}\label{174}
\tilG^{t}_{it} = \frac{1}{2u^{2}}\left(\varphi_{i} + \by^{l}(Y_{l,i}-Y_{i,l})\right), \text{ }\text{ }\text{ }\text{ }
\tilG^{j}_{it} = -\frac{1}{2}\bg^{jl}(Y_{l,i}-Y_{i,l}) + \frac{\by^{j}}{2 u^{2}}\left(\varphi_{i} + \by^{l}(Y_{l,i}-Y_{i,l})\right).
\end{equation}
Let $X_{i}=\partial_{i}+f_{i}\partial_{t}$ denote tangent vectors to the graph, then
\begin{align} \label{05}
\begin{split}
\pi_{ij}
=& -\widetilde{g}(\tilna_{X_{i}} X_{j},N)   \\
=& - (\tilG^{l}_{ij} + f_{i}\tilG^{l}_{jt} +  f_{j}\tilG^{l}_{it} + f_{i}f_{j}\tilG^{l}_{tt})
 \widetilde{g}\left(\partial_{l}, \frac{u\barna f}{\sqrt{\volbarg}}\right)  \\
& - (\partial_{ij}f+\tilG^{t}_{ij} + f_{i}\tilG^{t}_{jt} +  f_{j}\tilG^{t}_{it} + f_{i}f_{j}\tilG^{t}_{tt}) \widetilde{g}\left(\partial_{t}, \frac{n}{\sqrt{\volbarg}} \right)   \\
=& \frac{u}{\sqrt{\volbarg}}
\left( \barna_{ij} f + u^{-1}\bk_{ij}+ f_{i}(\tilG^{t}_{jt}-\tilG^{l}_{jt}f_{l})
+f_{j}(\tilG^{t}_{it}-\tilG^{l}_{it}f_{l})
+f_{i}f_{j}(\tilG^{t}_{tt}-\tilG^{l}_{tt}f_{l}) \right).
\end{split}
\end{align}
Therefore
\begin{equation} \label{05-1}
\pi_{ij} = \frac{u \barna_{ij} f + \bk_{ij}+\frac{1}{2u}(f_{i} \varphi_{j}+f_{j} \varphi_{i})-\frac{u}{2}f_{i}f_{j} \barna f(\varphi)}{\sqrt{\volbarg}}
 + \frac{1}{2}f_{i}w^{l}(Y_{l,j}-Y_{j,l})
+ \frac{1}{2} f_{j}w^{l}(Y_{l,i}-Y_{i,l}).
\end{equation}

This formula for the second fundamental form involves quantities associated with $\overline{g}$. In order to obtain a formula involving only quantities associated with $g$, we will employ \textbf{Identity 1} below for the difference of Christoffel symbols. In particular
\begin{align} \label{0101}
\begin{split}
\Gamma^{l}_{ij}-\barG^{l}_{ij}
& = \frac{\by^{l}}{2u^{2}}(\barna_{j}Y_{i}+ \barna_{i}Y_{j})
-w^{l} \pi_{ij} + f_{i}f_{j}\tilG^{l}_{tt}
+ f_{i}\tilG^{l}_{jt}  + f_{j}\tilG^{l}_{it}  \\
&=\frac{\by^{l}}{2u^{2}}(\nabla_{j}Y_{i}+\nabla_{i}Y_{j})
+\frac{\by^{l}}{u^{2}}(\Gamma^{m}_{ij}-\barG^{m}_{ij})Y_{m}
-w^{l} \pi_{ij} + f_{i}f_{j}\tilG^{l}_{tt}
+ f_{i}\tilG^{l}_{jt}  + f_{j}\tilG^{l}_{it}.
\end{split}
\end{align}
Multiply by $Y_{l}$ and solve for $(\Gamma^{l}_{ij}-\barG^{l}_{ij})Y_{l}$ to obtain
\begin{equation} \label{0102}
\frac{\varphi}{u^{2}}(\Gamma^{l}_{ij}-\barG^{l}_{ij})Y_{l}
= \frac{| \by |_{\bg}^{2}}{2u^{2}}(\nabla_{j}Y_{i}+\nabla_{i}Y_{j})
-\pi_{ij}  w^{l}Y_{l} + f_{i}f_{j}\tilG^{l}_{tt} Y_{l}
+ f_{i}\tilG^{l}_{jt}Y_{l}  + f_{j}\tilG_{it}^{l}Y_{l}.
\end{equation}
For simplicity we temporarily assume that $\varphi$ does not vanish, however this will have no affect on the final result, which is valid without any such restriction on $\varphi$.
Substituting \eqref{0102} into the last line of \eqref{0101} produces
\begin{align}\label{0103}
\begin{split}
\Gamma^{l}_{ij}-\barG^{l}_{ij}
=& \frac{\by^{l}}{2\varphi}(\nabla_{j}Y_{i}+\nabla_{i}Y_{j})- \pi_{ij}(w^{l}+ \varphi^{-1}w^{m}Y_{m}\by^{l})  \\
&+f_{i}(\tilG^{l}_{jt}+ \varphi^{-1}\tilG^{m}_{jt}Y_{m}\by^{l})
+f_{j}(\tilG^{l}_{it}+ \varphi^{-1}\tilG^{m}_{it}Y_{m}\by^{l})
+f_{i}f_{j}(\tilG^{l}_{tt} +\varphi^{-1}\tilG^{m}_{tt}Y_{m}\by^{l}).
\end{split}
\end{align}
Use this and $\by^{i}f_{i}=0$ to compute the following expression from \eqref{05-1} in terms of $g$
\begin{align} \label{0104}
\begin{split}
\overline{\nabla}_{ij}f + u^{-1}\bk_{ij}
=& \nabla_{ij}f + \frac{1}{2u^{2}}(Y_{i;j}+Y_{j;i}) + (\Gamma^{l}_{ij}-\barG^{l}_{ij})(f_{l}+u^{-2}Y_{l}) \\
=& -\pi_{ij}(w^{l}f_{l}+\varphi^{-1}w^{l}Y_{l})
+ \nabla_{ij}f + \frac{1}{2\varphi}(\nabla_{j}Y_{i}+\nabla_{i}Y_{j}) \\
&+ f_{i}(\tilG^{l}_{jt}f_{l} + \varphi^{-1}\tilG^{l}_{jt}Y_{l})
+ f_{j}(\tilG^{l}_{it}f_{l} + \varphi^{-1} \tilG^{l}_{it}Y_{l})
+f_{i}f_{j}(\tilG^{l}_{tt}f_{l}+\varphi^{-1}\tilG^{l}_{tt}Y_{l}).
\end{split}
\end{align}
Also from \eqref{05}
\begin{equation} \label{0105}
\overline{\nabla}_{ij}f + u^{-1}\bk_{ij}
= u^{-1}\sqrt{\volbarg}\text{ }\! \pi_{ij}
- f_{i}(\tilG^{t}_{it} -  \tilG^{l}_{jt} f_{l} )
- f_{j}(\tilG^{t}_{it}-  \tilG^{l}_{jt} f_{l} )
+ f_{i}f_{j} \tilG^{l}_{tt}f_{l}.
\end{equation}
Now compare \eqref{0104} and \eqref{0105} and solve for $\pi$ to obtain
\begin{equation} \label{0106}
\begin{split}
 \frac{u}{\varphi \sqrt{\volbarg}} \pi_{ij}
=& \left( \frac{\sqrt{\volbarg}}{u} + w^{l}f_{l}+ \frac{w^{l}Y_{l}}{\varphi}\right) \pi_{ij}\\
=& \nabla_{ij}f + \frac{1}{2\varphi}(\nabla_{j}Y_{i}+\nabla_{i}Y_{j})+ f_{i}(\tilG^{t}_{jt} + \varphi^{-1}\tilG^{l}_{jt}Y_{l})
+ f_{j}(\tilG^{t}_{it} + \varphi^{-1}\tilG^{l}_{it}Y_{l}) \\
=&  \nabla_{ij}f + \frac{1}{2\varphi}(\nabla_{j}Y_{i}+\nabla_{i}Y_{j})+  \frac{1}{2}f_{i}(\log\varphi)_{j} + \frac{1}{2}f_{j}(\log\varphi)_{i}.
\end{split}
\end{equation}
Therefore
\begin{equation} \label{0107}
\pi_{ij} = \frac{1}{u\sqrt{\volg}}
\left( \varphi \nabla_{ij}f + \frac{1}{2}(\nabla_{j}Y_{i}+\nabla_{i}Y_{j}) +  \frac{1}{2} f_{i}\varphi_{j} + \frac{1}{2}f_{j}\varphi_{i}\right),
\end{equation}
where \eqref{1001} was used.
Notice that $Y_{i}$ depends on $f$ through the expression
\begin{equation} \label{0108}
Y_{i}= \by^{\phi}\bg_{\phi i} = \by^{\phi} (g_{\phi i}+f_{i}Y_{\phi})
= \by^{\phi} g_{\phi i} + f_{i} | \by |_{\bg}^{2},
\end{equation}
from which it follows that
\begin{equation} \label{0109}
\nabla_{j}Y_{i}+\nabla_{i}Y_{j}= 2 | \by |_{\bg}^{2} \nabla_{ij}f
+ \partial_{i}f \partial_{j}|\by |_{\bg}^{2}
+\partial_{j}f \partial_{i}| \by |_{\bg}^{2}
+ g_{\phi j}\partial_{i}\byp + g_{\phi i}\partial_{j}\byp.
\end{equation}
Inserting \eqref{0109} into \eqref{0107} then produces
\begin{equation}\label{175}
\pi_{ij}=\frac{u\nabla_{ij}f+u_{i}f_{j}+u_{j}f_{i}+\frac{1}{2u}(g_{i\phi}Y^{\phi}_{,j}+g_{j\phi}Y^{\phi}_{,i})}{\sqrt{1+u^{2}|\nabla f|_{g}^{2}}}.
\end{equation}

The remaining part of the proof consists of verifying several identities.

\subsection*{Identity 1}
\begin{equation}\label{06}
\begin{split}
\Gamma^{l}_{ij}-\barG^{l}_{ij}
=& -w^{l} \pi_{ij} + u^{-1}\by^{l}\bk_{ij} + \frac{1}{2}f_{i}f_{j}\varphi^{l} \\
&+ f_{i}\left(-\frac{1}{2}\bg^{lm}(\overline{\nabla}_{j}Y_{m}-\overline{\nabla}_{m}Y_{j}) + \frac{\by^{l}}{2 u^{2}}(\varphi_{j}+\by^{m}(\overline{\nabla}_{j}Y_{m}-\overline{\nabla}_{m}Y_{j})) \right)  \\
& + f_{j}\left(-\frac{1}{2}\bg^{lm}(\overline{\nabla}_{i}Y_{m}-\overline{\nabla}_{m}Y_{i}) + \frac{\by^{l}}{2 u^{2}}(\varphi_{i}+\by^{m}(\overline{\nabla}_{i}Y_{m}-\overline{\nabla}_{m}Y_{i})) \right)
\end{split}
\end{equation}

\begin{proof}
Observe that
\begin{equation} \label{176}
N=w+\frac{1}{u\sqrt{1-|\overline{\nabla} f|_{\overline{g}}^{2}}}\partial_{t},
\end{equation}
which yields
\begin{align} \label{07}
\begin{split}
\widetilde{\nabla}_{X_{i}} X_{j}
&=\nabla_{X_{i}} X_{j} - \widetilde{g}( \tilna_{X_{i}} X_{j}, N ) N   \\
&=\Gamma^{l}_{ij} X_{l} + \pi_{ij} N  \\
&= \Gamma^{l}_{ij} X_{l} + \pi_{ij} w^{l}X_{l}+ u^{-1}\sqrt{\volbarg}\text{ }\! \pi_{ij} \partial_{t}.
\end{split}
\end{align}
Alternatively, using \eqref{05} produces
\begin{align} \label{08}
\begin{split}
\tilna_{X_{i}} X_{j}
=& (\tilG^{l}_{ij} + f_{i}\tilG^{l}_{jt} +  f_{j}\tilG^{l}_{it} + f_{i}f_{j}\tilG^{l}_{tt}) \partial_{l}
+ (\partial_{ij}f+\tilG^{t}_{ij} + f_{i}\tilG^{t}_{jt} +  f_{j}\tilG^{t}_{it} + f_{i}f_{j}\tilG^{t}_{tt}) \partial_{t}     \\
=& (\tilG^{k}_{ij} + f_{,i}\tilG^{k}_{jt} +  f_{,j}\tilG^{k}_{it} + f_{,i}f_{,j}\tilG^{k}_{tt}) X_{k}  \\
& + \left(\partial_{ij}f+\tilG^{t}_{ij} + f_{i}\tilG^{t}_{jt} +  f_{j}\tilG^{t}_{it} + f_{i}f_{j}\tilG^{t}_{tt} - (\tilG^{l}_{ij} + f_{i}\tilG^{l}_{jt} +  f_{j}\tilG^{l}_{it} + f_{i}f_{j}\tilG^{}_{tt})f_{l}\right) \partial_{t}   \\
&= (\tilG^{l}_{ij} + f_{i}\tilG^{}_{jt} +  f_{j}\tilG^{l}_{it} + f_{i}f_{j}\tilG^{l}_{tt}) X_{l}
+ u^{-1}\sqrt{\volbarg}\text{ }\!\pi_{ij} \partial_{t}.
\end{split}
\end{align}
Therefore by comparing \eqref{07} and \eqref{08}
\begin{equation} \label{09}
\Gamma^{l}_{ij} - \barG^{l}_{ij} = -w^{l} \pi_{ij} + u^{-1}\by^{l}\bk_{ij}
+ f_{i}\tilG^{l}_{jt} + f_{j}\tilG^{l}_{it} + f_{i}f_{j}\tilG^{l}_{tt},
\end{equation}
from which the desired result follows with the help of \eqref{02} and \eqref{174}.
\end{proof}
%\begin{flushright}
%$\quad \square$
%\end{flushright}

It will be assumed that $k$ and $\pi$ are extended trivially to all of $\overline{M}\times\mathbb{R}$, in that $k(\partial_{t},\cdot)=\pi(\partial_{t},\cdot)=0$.

\subsection*{Identity 2}
\begin{equation} \label{010}
div_{g}k(w)
=u^{-1} \bdiv(uk(w, \cdot)) + w(k(w, w)) -\widetilde{g} (k,\pi) - 2\tg ( k(w, \cdot), \pi(w, \cdot)) + (Tr_{\tg}\pi)k(w, w)
\end{equation}

\begin{proof}
Direct calculation yields
\begin{align} \label{011}
\begin{split}
div_{g}k(w)
=& \left(\bg^{ij} -u^{-2}\by^{i}\by^{j} + w^{i}w^{j}\right)w^{l}\nabla_{j} k_{il}   \\
=& \left(\bg^{ij} -u^{-2}\by^{i}\by^{j} + w^{i}w^{j}\right)w^{l}
(\overline{\nabla}_{j}k_{il} - (\Gamma_{ij}^{m}-\barG_{ij}^{m}) k_{ml} -(\Gamma_{jl}^{m}-\barG_{jl}^{m}) k_{im}) \\
=& \bdiv k(w) - u^{-2}\barna_{\by}k(\by,w) +  \barna_{w}k(w,w)  \\
& -\left(\bg^{ij} -u^{-2}\by^{i}\by^{j} + w^{i}w^{j}\right)w^{l}
\left( (\Gamma_{ij}^{m}-\barG_{ij}^{m}) k_{ml} +(\Gamma_{jl}^{m}-\barG_{jl}^{m}) k_{im} \right).
\end{split}
\end{align}
Next, each term on the right-hand side of \eqref{011} will be computed separately. Let
\begin{equation}\label{177}
A_{ij}= u^{-1}\bk_{ij}
+ f_{i}(\tilG_{jt}^{t}-\tilG_{jt}^{l}f_{l})
+f_{j}(\tilG_{it}^{t}-\tilG_{it}^{l}f_{l})-f_{i}f_{j}\tilG_{tt}^{l}f_{l}.
\end{equation}

\subsection*{Identity 2-1}
\begin{align}\label{012}
\begin{split}
\overline{\nabla}_{j}w_{i}
&= \pi_{ij}+ \left(\pi(w, \partial_{j}) -  \pi \left(\frac{\by}{u\sqrt{\volbarg}}, \partial_{j}\right)\right) w_{i} - \frac{2 u_{j}} {u^{2}\sqrt{\volbarg}}Y_{i}\\
&  - \frac{u}{\sqrt{\volbarg}}\left(A_{ij}-u^{-2}\overline{\nabla}_{j}Y_{i}\right)
 + \frac{1}{\volbarg}\left((\log u)_{i} - u^{2}A(\barna f, \partial_{j})\right) w_{i}
\end{split}
\end{align}

\begin{proof}
Observe that
\begin{align} \label{013}
\begin{split}
\overline{\nabla}_{j}w_{i}
&= \overline{\nabla}_{j}\left(\frac{u\barna  f + u^{-1} \by}{\sqrt{\volbarg}}\right)_{i} \\
&= \frac{u(\overline{\nabla}_{ij}f + u^{-2} \overline{\nabla}_{j}Y_{i} + Y_{i}\partial_{j}u^{-2})}{\sqrt{\volbarg}}
+ \frac{1}{\volbarg}\left((\log u)_{j} + u^{2} \overline{g}^{lm}f_{m}\overline{\nabla}_{lj}f\right)w_{i}.
\end{split}
\end{align}
Now substitute the following expressions
\begin{equation}\label{178}
 \overline{\nabla}_{ij}f = u^{-1}\sqrt{\volbarg}\text{ }\! \pi_{ij} -A_{ij},
\end{equation}
\begin{equation}\label{179}
\overline{g}^{lm}f_{m}\overline{\nabla}_{lj}f
= \left(u^{-1}\sqrt{\volbarg}\text{ }\!w^{l}-u^{-2}\by^{l}\right)\left(u^{-1}\sqrt{\volbarg}\text{ }\! \pi_{lj}\right)- A(\overline{\nabla}f,\partial_{j}).
\end{equation}
\end{proof}

\subsection*{Identity 2-2}
\begin{equation} \label{014}
\begin{split}
 \bdiv k(w)
=& \bdiv (k(w, \cdot))- \bg ( k, \pi )
- \bg \left( k(w, \cdot), \pi\left(w-\frac{\by}{u\sqrt{\volbarg}}, \cdot \right) \right) \\
&  + \frac{u\overline{g}^{il}\overline{g}^{jm}k_{lm}\left(A_{ij}-u^{-2}\overline{\nabla}_{j}Y_{i} \right) }{\sqrt{\volbarg}}
+ \frac{u^{2}\bg ( k(\bw, \cdot), A(\barna f, \cdot))}{\volbarg}    \\
& + 2k \left(\by, \frac{\barna u}{u^{2}\sqrt{\volbarg}} \right)
 - k \left(\bw, \frac{\barna u}{u (\volbarg)} \right)
\end{split}
\end{equation}

\begin{proof}
Since
\begin{equation}\label{180}
\bdiv k(w) = \bdiv(k(w, \cdot)) - \overline{g}^{il}\overline{g}^{jm}k_{lm}\overline{\nabla}_{j}w_{i},
\end{equation}
the desired result follows from \textbf{Identity 2-1}.
\end{proof}

\subsection*{Identity 2-3}
\begin{equation} \label{015}
\barna_{\by}k(\by,w)
=-k(\barna_{\by}\by,w)
-\bg ( k(\by,\cdot), \pi(\by,\cdot) )
 +\frac{u\bg ( k(\by,\cdot), A(\by, \cdot))}{\sqrt{\volbarg}}
- \frac{k \left(\by,\barna_{\by}\by \right)}{u\sqrt{\volbarg}}
\end{equation}

\begin{proof}
Since $\by=\by^{\phi}\partial_{\phi}$
\begin{equation}\label{016}
\barna_{\by}k(\by,w) = \by (k(\by,w))-k\left(\barna_{\by}\by, w\right) - k\left(\barna_{\by} w,\by\right)
=-k\left(\barna_{\by}\by, w \right) - k\left(\barna_{\by} w,\by \right),
\end{equation}
and
\begin{equation} \label{017}
\left(\barna_{\by} w \right)_{i}
= \frac{u\by^{l}\overline{\nabla}_{il}f + u^{-1} \by^{l}\overline{\nabla}_{l}Y_{i}}{\sqrt{\volbarg}}
= \pi(\partial_{i}, \by)
 - \frac{u}{\sqrt{\volbarg}}\left(A(\partial_{i}, \by)- \frac{(\barna_{\by}\by)_{i}}{u^{2}}\right).
\end{equation}
Insert \eqref{017} into \eqref{016} to get \eqref{015}.
\end{proof}

\subsection*{Identity 2-4}
\begin{align}\label{018}
\begin{split}
\barna_{w}k(w,w)
&= w(k(w,w)) -2\bg ( k(w,\cdot), \pi(w,\cdot))
  -u^{-2}(\volbarg) k(w,w)w(\varphi) \\
& -2 k(w,w)\pi(w,w)+\frac{4k(w,\by)w(u)+2u^{3}\bg( k(w,\cdot), A(w,\cdot))
 - 2u k(w, \barna_{w}\by )}{u^{2}\sqrt{\volbarg}}
\end{split}
\end{align}

\begin{proof}
Observe that
\begin{equation} \label{019}
\barna_{w}k(w,w) = w(k(w,w))-2k \left(\barna_{w}w, w \right),
\end{equation}
and with the help of \textbf{Identity 2-1}
\begin{align}\label{020}
\begin{split}
\left(\barna_{w} w \right)_{i} = w^{l} \overline{\nabla}_{l}w_{i}
=& \pi(\partial_{i}, w)+\pi(w, w)w_{i} \\
&- \frac{u}{\sqrt{\volbarg}} \left(A(\partial_{i}, w)- u^{-2}\left(\barna_{w}\by\right)_{i} \right)
 -\frac{2w(u)}{u^{2}\sqrt{\volbarg}}Y_{i} \\
& + \left(\frac{u^{-1}w(u) - u^{2}A(\barna f, w)}{\volbarg}-  \frac{\pi(w,\by)}{u\sqrt{\volbarg}}\right) w_{i}.
\end{split}
\end{align}
The last line of \eqref{020} will now be computed directly using definition of $A$, $u$, and $\bk$, along with
$\by^{l}\overline{\nabla}_{il}f=-\overline{g}^{lj}f_{j}\overline{\nabla}_{i}Y_{l}$, $\overline{g}^{im}f_{m}\overline{g}^{jl}f_{l}\overline{\nabla}_{j}Y_{i}=0$. Namely
\begin{align} \label{021}
\begin{split}
\pi(w, \by) &= \frac{u}{\sqrt{\volbarg}}\left(w^{i}\by^{j}\overline{\nabla}_{ij}f + u^{-1}\bk(w,\by)
+ \frac{1}{2}w^{i}f_{i}\by^{i}\overline{g}^{jl}f_{l}(\overline{\nabla}_{i}Y_{j} - \overline{\nabla}_{j}Y_{i}) \right) \\
&= - \frac{\by^{i}\overline{g}^{jl}f_{l}\overline{\nabla}_{i}Y_{j} }{\volbarg}
+ \frac{\by^{i}\overline{g}^{jl}f_{l}(\overline{\nabla}_{j}Y_{i}+\overline{\nabla}_{i}Y_{j})}{2(\volbarg)}
+ \frac{u^{2}|\barna f|_{\overline{g}}^{2}\by^{i}\overline{g}^{jl}f_{l}(\overline{\nabla}_{i}Y_{j} - \overline{\nabla}_{j}Y_{i})}{2 (\volbarg)} \\
&= \frac{1}{2} \by^{i}\overline{g}^{jl}f_{l}(\overline{\nabla}_{j}Y_{i}-\overline{\nabla}_{i}Y_{j}),
\end{split}
\end{align}
\begin{equation} \label{021-1}
w(u) = \frac{w \left(\varphi + | \by |_{\bg}^{2} \right)}{2u},
\end{equation}
and
\begin{align} \label{021-2}
\begin{split}
A(\barna f, w)
=& \frac{\sqrt{\volbarg}}{u} A(w, w) - \frac{1}{u^{2}}A(\by, w)   \\
=& \frac{\sqrt{\volbarg}}{u} \left(u^{-1}\bk(w, w)
+  u^{-2}w^{i}f_{i}w(\varphi) - \frac{1}{2}(w^{i}f_{i})^{2}\barna f(\varphi) \right) \\
&- \frac{1}{u^{2}} \left(\bk\left(\frac{\barna f}{\sqrt{\volbarg}}, \by\right)
+ \frac{1}{2}w^{l}f_{l}\by^{i}\overline{g}^{jm}f_{m}(\overline{\nabla}_{i}Y_{j} - \overline{\nabla}_{j}Y_{i}) \right)
\\
=& \frac{w(\varphi)}{2} \left(\frac{2|\barna f|_{\overline{g}}^{2}}{u^{2}}-|\barna f|_{\overline{g}}^{4} \right)
+ \frac{\by^{i}\overline{g}^{jm}f_{m}(\overline{\nabla}_{j}Y_{i} + \overline{\nabla}_{i}Y_{j})}{2u^{3}\sqrt{\volbarg}}\\
&
+ \frac{|\barna f|_{\overline{g}}^{2}\by^{i}\overline{g}^{jm}f_{m}(\overline{\nabla}_{j}Y_{i} - \overline{\nabla}_{i}Y_{j}) }{2u \sqrt{\volbarg}}.
\end{split}
\end{align}
Upon using \eqref{021}, \eqref{021-1}, and \eqref{021-2}, the last line of \eqref{020} simplifies by
\begin{equation} \label{022}
\frac{\pi(w, \by)}{u\sqrt{\volbarg}} -\frac{w(u)}{u(\volbarg)}
+ \frac{u^{2} A(\barna f,w)}{\volbarg}
= - \frac{(\volbarg)}{2 u^{2}}w(\varphi).
\end{equation}
The desired result is now obtained by substituting \eqref{020} and \eqref{022} into \eqref{019}.
\end{proof}

The following term from \eqref{011} may be rewritten by combining \textbf{Identity 2-2}, \textbf{Identity 2-3}, and \textbf{Identity 2-4}
\begin{align}\label{023}
\begin{split}
&\bdiv k(w) - u^{-2} \barna_{\by}k(\by,w)  +  \barna_{w}k(w,w) \\
=& \bdiv(k(w, \cdot))+w(k(w,w))- 3\bg ( k(w, \cdot), \pi(w, \cdot))
-2 k(w,w)\pi(w,w)- \bg (k, \pi) \\
&+u^{-2}\bg ( k(\by,\cdot), \pi(\by,\cdot))
+u^{-2}k\left(\barna_{\by}\by, w \right)
-u^{-2} (\volbarg)k(w,w)w(\varphi)
  \\
&+\frac{4k(w,\by)w(u)}{u^{2}\sqrt{\volbarg}}
- \frac{2k \left(w, \barna_{w}\by \right)}{u\sqrt{\volbarg}}
+ \frac{k \left(\by,\barna_{\by}\by \right)}{u^{3}\sqrt{\volbarg}}
+  \frac{2k \left(\by,\barna u\right)}{u^{2}\sqrt{\volbarg}} \\
& - \frac{k \left(w, \barna u\right)}{u (\volbarg)}
 + \frac{u\overline{g}\left(k, \left(A-u^{-2}\overline{\nabla}Y \right)\right) }{\sqrt{\volbarg}}
 +\frac{3 u\bg ( k(w,\cdot), A(w,\cdot))}{\sqrt{\volbarg}}   \\
&+ \frac{\bg ( k(w, \cdot), \pi(\by, \cdot))}{u\sqrt{\volbarg}}
- \frac{\bg ( k(w, \cdot), A(\by, \cdot)) }{\volbarg}
-\frac{\bg ( k(\by,\cdot), A(\by , \cdot))}{u\sqrt{\volbarg}}.
\end{split}
\end{align}
Each term involving $A$ will now be computed with \eqref{02}, \eqref{174}, and
\begin{equation}\label{181}
\barna f = u^{-1}\sqrt{\volbarg}\text{ }\!w - u^{-2}\by.
\end{equation}
Namely
\begin{align} \label{024}
%\begin{split}
\frac{u\overline{g}\left(k, \left(A-u^{-2}\overline{\nabla}Y \right)\right) }{\sqrt{\volbarg}}
&=2k \left(w-\frac{\by}{u\sqrt{\volbarg}}, \bg^{jl}(\tilG^{t}_{jt}-\tilG^{m}_{jt}f_{m})\partial_{l} \right) \nonumber\\
&- \tilG^{l}_{tt}f_{l} k\left(w-\frac{\by}{u\sqrt{\volbarg}}, u^{-1}\sqrt{\volbarg}\text{ }\!w-u^{-2}\by \right)\nonumber\\
&=2k \left(w-\frac{\by}{u\sqrt{\volbarg}}, \frac{\barna \varphi}{2 u^{2}}
+\frac{\sqrt{\volbarg}}{2u}\bg^{jl}w^{i}(\overline{\nabla}_{j}Y_{i}-\overline{\nabla}_{i}Y_{j})\partial_{l} \right) \\
&- \frac{\barna f (\varphi)}{2} k\left(w-\frac{\by}{u\sqrt{\volbarg}}, u^{-1}\sqrt{\volbarg}\text{ }\!w-u^{-2}\by \right),\nonumber
%\end{split}
\end{align}
\begin{align} \label{024-1}
\begin{split}
\frac{u\bg ( k(w,\cdot), A(w,\cdot))}{\sqrt{\volbarg}}
=& \frac{\bg ( k(w,\cdot), \bk(w,\cdot)) }{\sqrt{\volbarg}}
+k \left(w, w- \frac{\by}{u\sqrt{\volbarg}} \right)
w^{j}(\tilG^{t}_{jt}-\tilG^{l}_{jt}f_{l}) \\
&
-(w^{m}f_{m})\tilG^{l}_{tt}f_{l}k\left(w, w- \frac{\by}{u\sqrt{\volbarg}} \right) \\
&+\frac{u^{2} |\barna f|_{\overline{g}}^{2}}{\volbarg} k \left(w, \bg^{jl}(\tilG^{t}_{jt}-\tilG^{m}_{jt}f_{m})\partial_{l} \right)\\
=&
\frac{u^{2} |\barna f|_{\overline{g}}^{2}}{\volbarg}k \left(w, \frac{\barna \varphi}{2u^{2}}
+ \frac{\sqrt{\volbarg}}{2u}\bg^{jl}w^{i}(\overline{\nabla}_{j}Y_{i}-\overline{\nabla}_{i}Y_{j})\partial_{l} \right) \\
&+k \left(w, w- \frac{\by}{u\sqrt{\volbarg}} \right) \frac{(\volbarg)w(\varphi)}{2u^{2}}\\
&+k \left(w, \frac{\bg^{jl}w^{i}(\overline{\nabla}_{j}Y_{i}+\overline{\nabla}_{i}Y_{j})\partial_{l}}{2u\sqrt{\volbarg}} \right),
\end{split}
\end{align}
\begin{equation} \label{024-2}
 \frac{\bg( k(w, \cdot), \pi (\by, \cdot ))}{u\sqrt{\volbarg}}
  - \frac{\bg ( k(w, \cdot), A(\by, \cdot))}{\volbarg}
= -\frac{k \left(w, \bg^{jl}\overline{g}^{im}f_{m}\overline{\nabla}_{j}Y_{i}\partial_{l} \right)}{\volbarg},
\end{equation}
and
\begin{align} \label{024-3}
\begin{split}
\frac{\bg ( k(\by,\cdot), A(\by, \cdot)) }{u\sqrt{\volbarg}}
=& \frac{k \left(\by, \bg^{jl}\by^{i}(\overline{\nabla}_{j}Y_{i}+\overline{\nabla}_{i}Y_{j})\partial_{l} \right)}{2u^{3}\sqrt{\volbarg}}\\
&
- k \left(\by, \frac{w}{u^{2}} - \frac{\by}{u^{3}\sqrt{\volbarg}} \right)\pi(w, \by).
\end{split}
\end{align}
Now use \eqref{024}-\eqref{024-3} as well as the formula
\begin{equation}\label{182}
u\barna^{l}{u}= \frac{1}{2}\barna^{l}\varphi + \bg^{jl}\by^{i}\overline{\nabla}_{j}Y_{i},
\end{equation}
to rewrite the last three lines of \eqref{023} together with the last term on the third line by
\begin{align} \label{025}
\begin{split}
 &- k \left (\by , \frac{\bg^{jl}f^{i}\overline{\nabla}_{j}Y_{i}\partial_{l}}{u\sqrt{\volbarg}} \right)
+k \left(\by ,\frac{\bg^{jl}\by^{i}(\overline{\nabla}_{j}Y_{i}+\overline{\nabla}_{i}Y_{j})\partial_{l}}{2u^{3}\sqrt{\volbarg}} \right)
\\
&+ u^{-2}k \left(\by , \barna_{w}\by \right)
+ k \left(w , \frac{1}{2}\bg^{jl}w^{i}(\overline{\nabla}_{j}Y_{i}-\overline{\nabla}_{i}Y_{j})\partial_{l} \right)
\left(\frac{3 + u^{2} |\barna f|_{\overline{g}}^{2} }{u\sqrt{\volbarg}} \right)\\
&+k(w, \by)\left(\frac{4w(u)}{u^{2}\sqrt{\volbarg}} - \frac{\sqrt{\volbarg}\text{ }\!w(\varphi)}{2u^{3}} +\frac{1}{u^{2}} \pi(w , \by) \right)\\
&-k(\by, \by) \left(\frac{w(\varphi)}{2u^{4}} + \frac{\pi(w, \by)}{u^{3}\sqrt{\volbarg}} \right)
+ k \left(w , \frac{1}{2}\barna \varphi \right)
\left(\frac{1+u^{2}|\barna f |_{\overline{g}}^{2}}{u^{2}(\volbarg)} \right).
\end{split}
\end{align}
Note also that
\begin{align}\label{183}
\begin{split}
\overline{g}^{il}f_{l}\overline{\nabla}_{j}Y_{i}=&-\by^{i}\overline{\nabla}_{ij}f\\
 =& -u^{-1}\sqrt{\volbarg}\text{ }\!\pi(\by, \partial_{j})+A(\by, \partial_{j})\\
 =& -u^{-1}\sqrt{\volbarg}\text{ }\!\pi(\by, \partial_{j}) + \frac{1}{2}u^{-2}\by^{i}(\overline{\nabla}_{j}Y_{i} + \overline{\nabla}_{i}Y_{j})\\
& -u \sqrt{1-u^{2} | \overline{\nabla}f|_{\overline{g}}^{2}}\pi(w, \by)\left(\frac{\overline{g}_{ij}w^{i}}{u^{2}}-\frac{Y_{j}}{u^{3}\sqrt{\volbarg}}\right),
\end{split}
\end{align}
and
\begin{equation}\label{184}
\overline{g}^{lm}f_{m}\overline{Y}^{j}\overline{\nabla}_{j}Y_{i}
= -\by^{l}\by^{j}\overline{\nabla}_{jl}f= -u^{-1}\sqrt{\volbarg}\text{ }\!\pi(\by, \by),
\end{equation}
so that \eqref{023} becomes
\begin{align}\label{026}
\begin{split}
&\bdiv k(w) - u^{-2} \barna_{\by}k(\by,w)  +  \barna_{w}k(w,w)   \\
=& \bdiv(k(w, \cdot)) + w(k(w, w))
-3\bg ( k(w, \cdot), \pi(w, \cdot) ) -2k(w,w)\pi(w,w)-\bg ( k,\pi )   \\
&+ 2u^{-2}\bg ( k(\by, \cdot),\pi(\by, \cdot))
-u^{-4}k(\by,\by)\pi(\by,\by) +u^{-2}k \left(w, \barna_{\by}\by \right)
+ u^{-2}k \left(\by, \barna_{w}\by \right) \\
&
-u^{-3}w(u)k(\by, \by)
+ \frac{1}{2}k(w, \barna \varphi)\frac{1+ u^{2} | \barna f |_{\overline{g}}^{2}}{u^{2}(\volbarg)}  \\
&+ k(w, \by) \left(\frac{4w(u)}{u^{2}\sqrt{\volbarg}}
- \frac{\sqrt{\volbarg}}{2 u^{3}}w(\phi)
+ \frac{2}{u^{2}}\pi(w, \by) \right)  \\
&+ \frac{1}{2}k \left(w, w^{l}(\overline{\nabla}_{i}Y_{l}-\overline{\nabla}_{l}Y_{i})\overline{g}^{ij}\partial_{j} \right)
\frac{3+u^{2} |\barna f |_{\overline{g}}^{2}}{u\sqrt{\volbarg}}.
\end{split}
\end{align}

The last line of \eqref{011} will now be computed.

\subsection*{Identity 2-5}
\begin{align}\label{027}
\begin{split}
&\left(\bg^{ij} -u^{-2}\by^{i}\by^{j} + w^{i}w^{j} \right)w^{l}
\left( (\Gamma_{ij}^{m}-\barG_{ij}^{m}) k_{ml} + (\Gamma_{jl}^{m}-\barG_{jl}^{m}) k_{im} \right) \\
=& - (Tr_{\tg}\pi)k(w,w) -2k(w,w)\pi(w,w) -\bg ( k(w, \cdot), \pi(w,\cdot)) \\
&-\frac{1}{2u^{2}}k \left(w, \barna | \by |_{\bg}^{2} \right)
+u^{-2}k \left(w, \barna_{\by}\by \right)
+ u^{-2}k \left(\by, \barna_{w}\by \right)
-u^{-3}w(u)k(\by, \by)\\
&+ k(w, \by) \left(\frac{4w(u)}{u^{2}\sqrt{\volbarg}}
- \frac{\sqrt{\volbarg}}{2 u^{3}}w(\varphi) \right) \\
&+\frac{1}{2}k \left(w, w^{l}(\overline{\nabla}_{i}Y_{l}-\overline{\nabla}_{l}Y_{i})\overline{g}^{ij}\partial_{j}\right)
\frac{3+u^{2} |\barna f |_{\overline{g}}^{2}}{u\sqrt{\volbarg}}
+ k\left(w, \barna \varphi \right) \frac{ | \barna f |_{\overline{g}}^{2}}{\volbarg}
\end{split}
\end{align}

\begin{proof}
Recall that
\begin{equation}\label{185}
\Gamma_{ij}^{l}-\barG_{ij}^{l}=-w^{l}\pi_{ij} + u^{-1}\by^{l}\bk_{ij} + f_{i}\tilG_{jt}^{l}+ f_{j}\tilG_{it}^{l} + f_{i}f_{j}\tilG_{tt}^{l},
\end{equation}
and so
\begin{equation}\label{028.1}
 \tg^{ij}(\Gamma_{ij}^{l}-\barG_{ij}^{l})
= -w^{l} Tr_{\tg}\pi +  2\overline{g}^{ij}f_{i}\tilG_{jt}^{l}
 + | \barna f |_{\overline{g}}^{2} \tilG_{tt}^{l},
\end{equation}
\begin{equation}\label{028.2}
 w^{i}w^{j}(\Gamma_{ij}^{l}-\barG_{ij}^{l})
=-w^{l}\pi(w,w) + u^{-1}\by^{l}\bk(w,w)
+2w^{m}f_{m}w^{j}\tilG_{jt}^{l} + (w^{m}f_{m})^{2}\tilG_{tt}^{l},
\end{equation}
\begin{align}\label{028.3}
\begin{split}
\bg^{ij}w^{l}(\Gamma_{jl}^{m}-\barG_{jl}^{m}) k_{im}
=& -\bg ( k(w, \cdot),\pi(w, \cdot)) + u^{-1}\bg ( k(\by, \cdot),\bk (w, \cdot)) \\
&+w^{l}f_{l}\bg^{ij}\tilG_{jt}^{m}k_{im}
+ w^{l}f_{l}k(\barna f, \tilG_{tt}^{m} \partial_{m})
+ k(\barna f, w^{l}\tilG_{lt}^{m} \partial_{m}),
\end{split}
\end{align}
and
\begin{equation}\label{028.4}
  \by^{i}\by^{j}w^{l}(\Gamma_{jl}^{m}-\barG_{jl}^{m}) k_{im}
= -k(w,\by)\pi(w,\by)+ u^{-1} k(\by,\by)\bk(w,\by)
+  w^{i}f_{i} k(\by, \by^{l}\tilG_{lt}^{m} \partial_{m}).
\end{equation}
Next, evaluate \eqref{028.1}-\eqref{028.4} with \eqref{02} and \eqref{174}.
Using also \eqref{181} and
\begin{equation}\label{186}
\tilG_{it}^{m}= \tilG_{it}^{t}\by^{m}
- \frac{1}{2}\bg^{ml}(\overline{\nabla}_{i}Y_{l}-\overline{\nabla}_{l}Y_{i}),
\end{equation}
leads to \eqref{027}.
\end{proof}

To complete the proof of \textbf{Identity 2}, subtract \textbf{Identity 2-5} from \eqref{026} to find
\begin{align}\label{029}
\begin{split}
 div_{g}k(w)
=& \bdiv(k(w, \cdot))
 +u^{-1}k \left(w, \barna u \right)
+ w (k(w, w))  +(Tr_{\tg}\pi)k(w,w)\\
& -\bg ( k,\pi )+ 2u^{-2} \bg ( k(\by, \cdot),\pi(\by, \cdot))
-u^{-4}k(\by,\by)\pi(\by,\by)  \\
&
-2\bg ( k(w, \cdot), \pi(w, \cdot))
+ 2u^{-2}k(w, \by)\pi(w, \by)  \\
=&u^{-1} \bdiv(u k(w, \cdot)) + w(k(w,w)) -\tg ( k,\pi ) - 2\tg ( k(w, \cdot), \pi(w, \cdot))
+ (Tr_{\tg}\pi)k(w, w).
\end{split}
\end{align}
\end{proof}

\subsection*{Identity 3}
\begin{align}\label{030}
\begin{split}
div_{g}k(\partial_{\phi})
=& \bdiv(k(\partial_{\phi}, \cdot)) + u^{-1}\bdiv(u k(\partial_{\phi}, w)w)
-g ( k(\partial_{\phi}, \cdot), \pi(v, \cdot) ) \\
&+ \bg ( k(\partial_{\phi}, \cdot), \bk (u\barna f, \cdot)) -u^{-1}\barna f(u) k(\partial_{\phi}, \by)
\end{split}
\end{align}

\begin{proof}
Observe that
\begin{equation}\label{187}
g^{ij}\Gamma_{j\phi}^{l}k_{li}=\frac{1}{2}k^{lj}(\partial_{j}g_{l \phi}- \partial_{l}g_{j \phi})=0,
\end{equation}
and so
\begin{align} \label{031}
\begin{split}
div_{g}k(\partial_{\phi})
=& g^{ij}\nabla_{j}k_{i\phi}   \\
=& g^{ij}(\partial_{j}(k_{i \phi}) - \Gamma_{ij}^{l}k_{l \phi}-\Gamma_{j\phi}^{l}k_{li})\\
=& \left(\bg^{ij} -u^{-2}\by^{i}\by^{j} + w^{i}w^{j} \right)
\left( \barna_{\partial_{j}}(k(\partial_{\phi}, \cdot))(\partial_{i})
- (\Gamma_{ij}^{l}-\barG_{ij}^{l}) k_{l\phi} \right) \\
=& \bdiv(k(\partial_{\phi}, \cdot))
+  u^{-2}k \left(\barna_{\by}\by,\partial_{\phi} \right)
+\barna_{w}(k(\partial_{\phi}, \cdot))(w) \\
& -\left(\bg^{ij} -u^{-2}\by^{i}\by^{j} + w^{i}w^{j} \right)
(\Gamma_{ij}^{l}-\barG_{ij}^{l}) k_{ l\phi} .
\end{split}
\end{align}

\subsection*{Identity 3-1}
\begin{equation} \label{032}
\begin{split}
\barna_{w}(k(\partial_{\phi}, \cdot))(w)
=& w(k(w,\partial_{\phi}))
  -g( k(\partial_{\phi},\cdot), \pi(v,\cdot))
+ \bg ( k(\partial_{\phi}, \cdot), \bk(u\barna f, \cdot)) \\
&- u^{-1}\barna f(u)k(\partial_{\phi}, \by)
 - \frac{k\left(\partial_{\phi}, \barna_{w}\by \right)}{u\sqrt{\volbarg}}
+ \frac{2w(u)k(\partial_{\phi}, \by)}{u^{2}\sqrt{\volbarg}} \\
&
+ \frac{|\barna f |_{\overline{g}}^{2}k \left(\partial_{\phi},  u\barna u \right)}{\volbarg}
+ \frac{k \left(\partial_{\phi}, f^{l}\overline{\nabla}^{i}Y_{l}\partial_{i} \right)}{\volbarg}
\end{split}
\end{equation}

\begin{proof}
First note that
\begin{equation} \label{033}
\barna_{w}(k(\partial_{\phi}, \cdot))(w) = w (k(w,\partial_{\phi}))-k \left(\barna_{w}w, \partial_{\phi}\right).
\end{equation}
Use \eqref{020} and \eqref{022} to evaluate $\barna_{w}w$ as in \textbf{Identity 2-4}, then
\begin{align} \label{035}
\begin{split}
\barna_{w}(k(\partial_{\phi}, \cdot))(w)
=& w(k(w,\partial_{\phi}))
-\bg ( k(\partial_{\phi},\cdot), \pi(w,\cdot)) - k(\partial_{\phi},w)\pi(w,w) \\
& -\frac{1}{2}u^{-2}w(\phi)(\volbarg)k(w,\partial_{\phi})
 +\frac{2w(u)k(\partial_{\phi},\by)}{u^{2}\sqrt{\volbarg}}
\\
& +\frac{u(\bg ( k(\partial_{\phi},\cdot), A(w,\cdot))
 - u^{-2} k(\partial_{\phi}, \barna_{w}\by))}{\sqrt{\volbarg}}.
\end{split}
\end{align}
Consider the first term on the last line of \eqref{035}. With help from \eqref{024} and the trivial identity
\begin{equation}\label{188}
w^{i}(\overline{\nabla}_{j}Y_{i}-\overline{\nabla}_{i}Y_{j})
= -w^{i}(\overline{\nabla}_{j}Y_{i}+\overline{\nabla}_{i}Y_{j})+2w^{i}\overline{\nabla}_{j}Y_{i},
\end{equation}
it follows that
\begin{align}\label{036}
\begin{split}
\frac{u\bg ( k(\partial_{\phi},\cdot), A(w,\cdot))}{\sqrt{\volbarg}}
=&k \left(\partial_{\phi}, \frac{\bg^{jl}w^{i}(\overline{\nabla}_{j}Y_{i}+\overline{\nabla}_{i}Y_{j})\partial_{l}}{2u\sqrt{\volbarg}} \right)\\&
+k \left(\partial_{\phi}, w- \frac{\by}{u\sqrt{\volbarg}} \right)\frac{(\volbarg)w(\varphi)}{2u^{2}} \\
&+\frac{u^{2} |\barna f|_{\overline{g}}^{2}}{\volbarg} k\left(\partial_{\phi}, \frac{\barna \varphi}{2u^{2}}
+ \frac{\sqrt{\volbarg}}{2u}\bg^{jl}\bw^{i}(\overline{\nabla}_{j}Y_{i}-\overline{\nabla}_{i}Y_{j})\partial_{l} \right) \\
=&\frac{\sqrt{\volbarg}}{2u}k \left(\partial_{\phi}, \bar{g}^{jl}w^{i}(\overline{\nabla}_{j}Y_{i}+\overline{\nabla}_{i}Y_{j})\partial_{l} \right)
\\
&+k \left(\partial_{\phi}, w- \frac{\by}{u\sqrt{\volbarg}} \right) \frac{(\volbarg)w(\varphi)}{2u^{2}}\\
&+\frac{ |\barna f|_{\overline{g}}^{2}}{\volbarg}k(\partial_{\phi}, u \barna u)
+\frac{u^{2}|\barna f|_{\overline{g}}^{2}}{\volbarg}k \left(\partial_{\phi}, \bg^{jl}f^{i}\overline{\nabla}_{j}Y_{i}\partial_{l} \right).
\end{split}
\end{align}
Now use \eqref{021}, as well as
\begin{equation}\label{189}
v=w-\frac{\sqrt{\volbarg}}{u}\by,\text{ }\text{ }\text{ }\text{ }
\frac{\sqrt{\volbarg}}{u}\pi(\by, \partial_{j})=-f^{l}\overline{\nabla}_{j}Y_{l}+u^{-1}\bk(\by, \partial_{j}) - f_{j}\pi(w,\by),
\end{equation}
to find
\begin{equation} \label{037}
\begin{split}
 \bg ( k(\partial_{\phi},\cdot), \pi(w,\cdot)) + k(\partial_{\phi}, w)\pi(w,w)
&= g( k(\partial_{\phi},\cdot), \pi(v,\cdot)) + u^{-2} k(\partial_{\phi}, \by) \pi(v, \by) \\
&+ \frac{\sqrt{\volbarg}}{u}[k(\partial_{\phi}, w)\pi(w, \by)
+ \bg ( k(\partial_{\phi}, \cdot), \pi(\by, \cdot))] \\
&=g ( k(\partial_{\phi},\cdot), \pi(v,\cdot))
+ u^{-1} \bg ( k(\partial_{\phi},\cdot), \bk(\by,\cdot))\\
&-k \left(\partial_{\phi}, \bg^{jl}f^{i}\overline{\nabla}_{j}Y_{i}\partial_{l} \right)
+ \frac{1}{2u^{2}}k(\partial_{\phi}, \by) \barna f(| \by |_{\bg}^{2}).
\end{split}
\end{equation}
Substituting \eqref{036} and \eqref{037} into \eqref{035} yields \eqref{032}, with the help of \eqref{181}.
\end{proof}

\subsection*{Identity 3-2}
\begin{equation} \label{039}
\begin{split}
& \left(\bg^{ij} -u^{-2}\by^{i}\by^{j} + w^{i}w^{j} \right)
 \left(\Gamma_{ij}^{l}-\barG_{ij}^{l} \right) k_{l\phi}  \\
&= -k(w,\partial_{\phi})(Tr_{\tg}\pi + \pi(w,w))
+ u^{-2}k \left(\partial_{\phi}, \barna_{\by}\by \right)
- \frac{1}{u\sqrt{\volbarg}}
k(\partial_{\phi}, \barna_{w}\by)  \\
&+ 2k(\partial_{\phi}, \by) \left(\frac{w(u)}{u^{2}\sqrt{\volbarg}} \right)
+ k(\partial_{\phi},  u\barna u)\frac{| \barna f |_{\overline{g}} ^{2}}{\volbarg} + \frac{k \left(\partial_{\phi}, \bg^{ij}f^{l}\overline{\nabla}_{i}Y_{l} \partial_{j} \right)}{\volbarg}
\end{split}
\end{equation}

\begin{proof}
Proceed in the same way as in the proof of \textbf{Identity 2-5}. Namely use \eqref{185}-\eqref{028.2}, \eqref{186}, and substitute the expressions \eqref{02} and \eqref{174}.
\end{proof}

Combining \textbf{Identity 3-1} and \textbf{Identity 3-2} produces
\begin{align} \label{041}
\begin{split}
div_{g}k(\partial_{\phi})
=& \bdiv(k(\partial_{\phi}, \cdot))
+\bw(k(w,\partial_{\phi})) + k(w,\partial_{\phi})(Tr_{\tg}\pi + \pi(w,w)) \\
&-g( k(\partial_{\phi}, \cdot),\pi(v, \cdot)) + \bg( k(\partial_{\phi}, \cdot), \bk(u\barna f, \cdot)) -u^{-1}\barna f(u)k(\partial_{\phi}, \by).
\end{split}
\end{align}

\subsection*{Identity 3-3}
\begin{equation}\label{042}
\bg^{ij}\overline{\nabla}_{j}w_{i}= Tr_{\tg}\pi + \pi(w,w) -u^{-1}w(u)
\end{equation}

\begin{proof}
By \textbf{Identity 2-1}
\begin{align}\label{043}
\begin{split}
\bg^{ij}\overline{\nabla}_{j}w_{i}
=& Tr_{\bg}\pi + \pi(w,w)
- \pi \left(\frac{\by}{u \sqrt{\volbarg}}, w \right)+\frac{w(u)}{u(\volbarg)}  \\
& - \frac{u}{\sqrt{\volbarg}} \bg^{ij} \left(A_{ij}- u^{-2}\overline{\nabla}_{j}Y_{i} \right)
-\frac{u^{2}A(\barna f, w)}{\volbarg} \\
=& Tr_{\bg}\pi + \pi(w, w)
- \pi(\by, w)\left(\frac{u |\barna f|_{\overline{g}}^{2}}{ (\volbarg)^{\frac{3}{2}}}+\frac{3}{u\sqrt{\volbarg}} \right)
- \frac{u \bk(w, \barna f)}{\volbarg}
\\
&+\frac{w(u)}{u(\volbarg)}
-\frac{w(\varphi)}{u^{2}(\volbarg)}
+ \frac{|\barna f|_{\overline{g}}^{2} w(\varphi)}{2(\volbarg)} \\
=& Tr_{\bg}\pi + \pi(w, w) -u^{-2}\pi(\by, \by) - u^{-1}w(u) \\
=& Tr_{\tg}\pi + \pi(w, w) - u^{-1}w(u).
\end{split}
\end{align}
The last line holds since
\begin{equation}\label{190}
\pi(\by, \by)= -\frac{uf^{l}\by^{j}\overline{\nabla}_{j}Y_{l}}{\sqrt{\volbarg}}.
\end{equation}
\end{proof}

%In sum, by \eqref{042}, the first line of \eqref{041} will be replaced by the following :
%\begin{equation} \label{044}
%\begin{split}
%&\bw(k(\bw,\partial_{\phi})) + k(\bw,\partial_{\phi})(Tr_{\tg}\pi + \pi(\bw,\bw))\\
%&=\bw(k(\bw,\partial_{\phi}))+ k(\bw,\partial_{\phi})(\bg^{ij}\bw_{i\bar{;}j} +\frac{\bw(u)}{u})
%=\frac{1}{u}\bdiv(u k(\partial{\phi}, \bw)\bw)
%\end{split}
%\end{equation}

Inserting \eqref{042} into \eqref{041} yields \textbf{Identity 3}.
\end{proof}

\subsection*{Identity 4}
\begin{align} \label{045}
\begin{split}
\bdiv(\pi(\partial_{\phi}, \cdot))
=& \frac{1}{\sqrt{\volbarg}}\bdiv(u \overline{\nabla}^{2}f(\partial_{\phi}, \cdot))
- u^{-1} \bdiv(u \pi(\partial_{\phi}, w)w)\\
&+ g( \pi(\partial_{\phi}, \cdot),\pi(v, \cdot)) - \bg( \pi(\partial_{\phi}, \cdot), \bk(u\barna f, \cdot))
+ u^{-1}\barna f(u)\pi(\partial_{\phi}, \by)
\end{split}
\end{align}

\begin{proof}
The following condition will be used throughout this proof
\begin{equation} \label{047}
\bdiv \bk(\partial_{\phi})=0.
\end{equation}
A direct computation with \eqref{05-1} produces
\begin{align} \label{048}
\begin{split}
 \bdiv(\pi(\partial_{\phi}, \cdot))
=& \frac{u}{\sqrt{\volbarg}}
\left[\bdiv(\overline{\nabla}^{2}f(\partial_{\phi}, \cdot))
-\bk \left(\partial_{\phi}, u^{-2}\barna u \right)
- \barna f(\pi(w, \partial_{\phi}))\right]\\
&  - \frac{u\pi(w, \partial_{\phi}))\Delta_{\overline{g}}f}{\sqrt{\volbarg}}
+\pi \left(\partial_{\phi}, \barna \log \left(\frac{u}{\sqrt{\volbarg}} \right) \right).
\end{split}
\end{align}

\subsection*{Identity 4-1}
\begin{align} \label{049}
\begin{split}
&\pi \left(\partial_{\phi}, \barna \log \left(\frac{u}{\sqrt{\volbarg}} \right) \right)\\
=& g ( \pi(\partial_{\phi}, \cdot), \pi(v, \cdot))-
\bg ( \pi(\partial_{\phi}, \cdot), \bk (u \barna f, \cdot))
 +\pi \left(\partial_{\phi}, u^{-1}\barna u \right)  \\
& -\pi(\partial_{\phi}, w)\pi(w, w)
-  \frac{\barna f(\varphi)\sqrt{\volbarg}}{2u} \pi(\partial_{\phi}, w)
+ u^{-1}\barna f(u) \pi(\partial_{\phi}, \by)
\end{split}
\end{align}

\begin{proof}
Observe that
\begin{equation} \label{050}
\begin{split}
&\pi\left(\partial_{\phi}, \barna \log\left(\frac{u}{\sqrt{\volbarg}}\right)\right) \\
=& \pi(\partial_{\phi}, u^{-1}\barna u)
+ \frac{| \barna f |_{\overline{g}}^{2}}{\volbarg}\pi(\partial_{\phi}, u\barna u)
+ \frac{u^{2}}{2(\volbarg)} \pi(\partial_{\phi},\barna | \barna f |_{\overline{g}}^{2}) \\
=& \frac{\pi(\partial_{\phi}, \barna u)}{u(\volbarg)}
+ \bg\left(\pi\left(\frac{u\barna f}{\sqrt{\volbarg}} , \cdot\right), \pi(\partial_{\phi}, \cdot)\right)
- \frac{u^{2}\bg( A(\barna f, \cdot), \pi(\partial_{\phi}, \cdot))}{\volbarg}.
\end{split}
\end{equation}
In order to evaluate the last term in \eqref{050}, follow the computations in \eqref{024-1}, replacing $k(w,\cdot)$ by $\pi(\partial_{\phi}, \cdot)$ to find
\begin{align} \label{051}
\begin{split}
&\frac{u^{2}\bg( A(\barna f, \cdot), \pi(\partial_{\phi}, \cdot))}{\volbarg}\\
=& \frac{u\bg( A(w, \cdot), \pi(\partial_{\phi}, \cdot))}{\sqrt{\volbarg}}
-\frac{\bg( A(\by, \cdot), \pi(\partial_{\phi}, \cdot))}{\volbarg}  \\
=& \pi \left(\partial_{\phi}, \frac{\bg^{jl}w^{i}(\overline{\nabla}_{j}Y_{i}
+\overline{\nabla}_{i}Y_{j})\partial_{l}}{2u\sqrt{\volbarg}} \right)
+\pi \left(\partial_{\phi}, w- \frac{\by}{u\sqrt{\volbarg}} \right) \frac{(\volbarg)w(\varphi)}{2u^{2}} \\
&+\frac{u^{2} |\barna f|_{\overline{g}}^{2}}{\volbarg} \pi\left(\partial_{\phi}, \frac{\barna \varphi}{2u^{2}}
+ \frac{\sqrt{\volbarg}}{2u}\bg^{jl}w^{i}(\overline{\nabla}_{j}Y_{i}
-\overline{\nabla}_{i}Y_{j})\partial_{l} \right)
-\frac{\bg( A(\by, \cdot), \pi(\partial_{\phi}, \cdot)) }{\volbarg}.
\end{split}
\end{align}
Moreover
\begin{align} \label{052}
\begin{split}
&\bg\left( \pi \left(\frac{u\barna f}{\sqrt{\volbarg}}, \cdot \right),\pi(\partial_{\phi}, \cdot)\right) \\
=&\bg\left( \pi \left(v-\frac{u | \barna f |_{\overline{g}}^{2} \by }{\sqrt{\volbarg}}, \cdot \right), \pi(\partial_{\phi}, \cdot)\right) \\
=& g ( \pi(v, \cdot), \pi(\partial_{\phi}, \cdot)) +u^{-2}\pi(\partial_{\phi}, \by)\pi(v, \by)- \bg \left( \pi \left(\frac{u | \barna f |_{\overline{g}}^{2} \by }{\sqrt{\volbarg}}, \cdot \right), \pi(\partial_{\phi}, \cdot)\right)
\\
&- \pi(\partial_{\phi}, w)\pi \left(w - \frac{\sqrt{\volbarg}}{u}\by , w \right).
\end{split}
\end{align}
Substitute \eqref{051}, \eqref{052} into \eqref{050}, and use the following relations to get the desired result:
\begin{equation}\label{191}
\pi\left(\frac{u | \barna f |_{\overline{g}}^{2} \by}{\sqrt{\volbarg}}, \partial_{i} \right)= \frac{u^{2} | \barna f |_{\overline{g}}^{2} }{\volbarg}(-f^{l}\overline{\nabla}_{i}Y_{l} + A(\by, \partial_{i})),
\end{equation}
\begin{align} \label{054}
\begin{split}
&\frac{\barna \varphi}{2u^{2}}
+ \frac{\sqrt{\volbarg}}{2u}\bg^{jl}w^{i}(\overline{\nabla}_{j}Y_{i}
-\overline{\nabla}_{i}Y_{j})\partial_{l}
-\bg^{jl}f^{i}\overline{\nabla}_{j}Y_{i}\partial_{l} \\
=& \frac{\barna u}{u} -  \frac{\sqrt{\volbarg}}{2u}\bg^{jl}w^{i}(\overline{\nabla}_{j}Y_{i}
+\overline{\nabla}_{i}Y_{j})\partial_{l},
\end{split}
\end{align}
\begin{align} \label{054-1}
%\begin{split}
\bg(A(\by,\cdot ), \pi(\partial_{\phi}, \cdot))
=& \pi \left(\partial_{\phi}, \frac{\bg^{jl}\by^{i}(\overline{\nabla}_{j}Y_{i} +\overline{\nabla}_{i}Y_{j})\partial_{l}}{2u^{2}} \right)
- \pi \left(\partial_{\phi}, \frac{\sqrt{\volbarg}}{u}w - \frac{\by}{u^{2}} \right) \pi(w, \by) \nonumber\\
=& - \bg(\pi(\partial_{\phi}, \cdot), \bk(u \barna f, \cdot))
+ \frac{\sqrt{\volbarg}}{2u} \pi \left(\partial_{\phi}, \bg^{jl}w^{i}(\overline{\nabla}_{j}Y_{i} +\overline{\nabla}_{i}Y_{j} )\partial_{l} \right) \\
&- \pi \left(\partial_{\phi}, \frac{\sqrt{\volbarg}}{u}w - \frac{\by}{u^{2}} \right) \pi(w, \by),\nonumber
%\end{split}
\end{align}
and
\begin{equation} \label{054-2}
\pi(v, \by) = \frac{f^{l}\by^{j}(\overline{\nabla}_{j}Y_{i} +\overline{\nabla}_{i}Y_{j})}{2u}
= -\pi(w, \by) + \frac{\barna f (|\by|_{\overline{g}}^{2})}{2u^{2}}.
\end{equation}
\end{proof}

We now finish the proof of \textbf{Identity 4}. Employ \textbf{Identity 2} and
\begin{equation}\label{192}
\pi(w,\by)=\frac{\by^{i}f^{j}(Y_{i\bar{;}j}-Y_{j\bar{;}i})}{2},
\text{ }\text{ }\text{ }\text{ }
\pi(\by,\by) = \frac{u\by^{i}\by^{j}\overline{\nabla}_{ij}f}{\sqrt{\volbarg}} =
-\frac{u\by^{i}f^{j}\overline{\nabla}_{i}Y_{j}}{\sqrt{\volbarg}},
\end{equation}
to find
\begin{align} \label{055}
\begin{split}
\Delta_{\overline{g}}f =&  u^{-1}\sqrt{\volbarg} Tr_{\bg}\pi - Tr_{\bg}A \\
=&u^{-1}\sqrt{\volbarg} Tr_{\bg}\pi
-2u^{-2}\pi(w, \by)
-  \frac{1}{2}\left(2u^{-2}-| \barna f |_{\overline{g}}^{2} \right) \barna f(\varphi)  \\
=& u^{-1}\sqrt{\volbarg} Tr_{\tg}\pi
+ u^{-3}\sqrt{\volbarg}\text{ }\! \pi(\by, \by) \\
&-2u^{-2}\pi(w, \by)
- \frac{1}{2}\left(2u^{-2}-| \barna f |_{\overline{g}}^{2} \right)\barna f(\varphi)  \\
=& u^{-1}\sqrt{\volbarg} Tr_{\tg}\pi
- u^{-1}\barna f (u) - \frac{(\volbarg)\barna f(\varphi)}{2u^{2}}.
\end{split}
\end{align}
Substituting \textbf{Identity 4-1} and \eqref{055} into \eqref{048} produces
\begin{align}\label{056}
\begin{split}
\bdiv(\pi(\partial_{\phi}, \cdot))
=&\frac{u}{\sqrt{\volbarg}}
\left(\bdiv(\overline{\nabla}^{2}f(\partial_{\phi}, \cdot))
-u^{-2}\bk(\partial_{\phi},\barna u) \right)
+u^{-1}\pi(\partial_{\phi}, \barna u) \\
&+g( \pi(\partial_{\phi},\cdot), \pi(v, \cdot))
-\bg( \pi(\partial_{\phi},\cdot), \bk(u \barna f, \cdot))
+ u^{-1}\pi (\partial_{\phi}, \by)\barna f(u) \\
& -w(\pi(\partial_{\phi}, w))
-\pi(\partial_{\phi}, w) \left(Tr_{\tg}\pi + \pi(w, w)  -u^{-1}w(u) \right).
\end{split}
\end{align}
The desired result may now be achieved with the help of
\eqref{042} and
\begin{equation}\label{193}
\pi \left(\partial_{\phi}, \frac{\barna u}{u} \right)
= \frac{u}{\sqrt{\volbarg}} \left(\overline{\nabla}^{2}f\left(\frac{\barna u}{u}, \partial_{\phi}\right)
+ \frac{\bk(\barna u, \partial_{\phi})}{u^{2}} - \frac{\barna f (u)}{u}\pi(w, \partial_{\phi}) \right). \end{equation}
\end{proof}

\subsection*{Identity 5}
\begin{equation} \label{057}
div_{g} \pi(\partial_{\phi})
= \frac{\bdiv \left(u\overline{\nabla}^{2} f(\partial_{\phi}, \cdot) \right)}{\sqrt{\volbarg}}
\end{equation}

\begin{proof}
Replace $k$ by $\pi$ in \textbf{Identity 3} to obtain
\begin{equation} \label{058}
\begin{split}
div_{g} \pi(\partial_{\phi})
=& \bdiv(\pi(\partial_{\phi}, \cdot)) + u^{-1}\bdiv(u \pi(\partial_{\phi}, w)w)\\
&-g( \pi(\partial_{\phi}, \cdot),\pi(v, \cdot))
+ \bg(\pi(\partial_{\phi}, \cdot), \bk (u \barna f, \cdot))
-u^{-1}\barna f (u)\pi(\partial_{\phi}, \by).
\end{split}
\end{equation}
Now employ \textbf{Identity 4}.
\end{proof}

\subsection*{Identity 6}
\begin{align} \label{059}
%\begin{split}
div_{g} (k-\pi)(v)
=&  | \pi |_{g}^{2} - g ( \pi, k) - u\overline{g}(\bk,\overline{\nabla}^{2}f)
 \nonumber \\
& +u^{-1}\bdiv \left(u \left(\overline{\nabla}^{2}f(\by, \cdot) + (k-\pi)(w, \cdot)
+(k-\pi)(w, w)\frac{udf}{\sqrt{\volbarg}} \right) \right)
%\end{split}
\end{align}

\begin{proof}
We have
\begin{equation} \label{060}
div_{g} (k-\pi)(v)
= div_{g} (k-\pi)(w) - u^{-1}\sqrt{\volbarg}\text{ }\!div_{g}(k-\pi)(\by).
\end{equation}
Replace $k$ by $(k-\pi)$ in \textbf{Identity 2} to calculate $div_{g} (k-\pi)(w)$. Next use \textbf{Identity 5} to evaluate $\byp div_{g} \pi(\partial_{\phi})$, and note that $\byp div_{g}k(\partial_{\phi})=0$ as
well as $\bk(\barna f, \barna u)=0$ to find
\begin{align} \label{061}
\begin{split}
div_{g} (k-\pi)(v)
=& |\pi |_{\tg}^{2} - \tg(\pi, k)+ 2|\pi(w, \cdot) |_{\tg}^{2}
 - 2  \tg(\pi(w, \cdot), k(w, \cdot)) \\
& + u^{-1} \bdiv(u (k-\pi)(w, \cdot))
 + w((k-\pi)(w,w))+ (k-\pi)(w,w)(Tr_{\tg}\pi) \\
& + u^{-1}\bdiv(u\overline{\nabla}^{2}f(\by, \cdot)) - u\overline{g}(\bk,\overline{\nabla}^{2}f).
\end{split}
\end{align}
Next observe that
\begin{equation} \label{062}
| \pi |_{\tg}^{2}  + 2 | \pi(w, \cdot) |_{\tg}^{2}
 =  | \pi |_{g}^{2} - \pi(w, w)^{2}, \text{ }\text{ }\text{ }\text{ }
 \tg( \pi, k)  + 2  \tg( \pi(w, \cdot), k(w, \cdot))
= g(\pi, k ) - \pi(w,w)k(w,w),
\end{equation}
as well as the fact that \textbf{Identity 3-3} together with $\bg^{ij}\overline{\nabla}_{j}Y_{i}=0$ imply
\begin{align} \label{064}
%\begin{split}
 \bw((k-\pi)(w, w))+ (k-\pi)(w,w)(Tr_{\tg}\pi+ \pi(w,w))
&= u^{-1} \bdiv(u(k-\pi)(w, w) w) \nonumber\\
&= u^{-1}\bdiv \left( u(k-\pi)(w, w)\frac{u df}{\sqrt{\volbarg}} \right).
%\end{split}
\end{align}
Combining \eqref{061}-\eqref{064} yields the desired result.
\end{proof}

\subsection*{Identity 7}
\begin{equation} \label{065}
%\begin{split}
\overline{R} - | \bk |_{\bg}^{2}
= 16\pi(\mu - J(v)) +   | k-\pi |_{g}^{2} + \frac{2}{u} \bdiv(u Q(\cdot))
+ (Tr_{g}\pi)^{2}-(Tr_{g}k)^{2} +2v(Tr_{g}\pi- Tr_{g}k)
%\end{split}
\end{equation}
where $Q$ is the 1-form defined by
\begin{equation} \label{066}
Q (\cdot)= \overline{\nabla}^{2}f(\by, \cdot)-  \bk(u \barna f, \cdot)
+ (k-\pi)(w, \cdot)+  (k-\pi)(w,w) \frac{u  df}{\sqrt{\volbarg}}.
\end{equation}

\begin{proof}
Recall the formula for $\overline{R}$ in \eqref{0016}, and observe that
\begin{equation}\label{193}
 \bdiv \bk (u \barna f) = \frac{1}{u} \bdiv(u\bk (u\barna f, \cdot))- u \overline{g}(\bk,\overline{\nabla}^{2}f).
\end{equation}
Apply \textbf{Identity 6} and solve for $\overline{R}-| \bk |_{\bg}^{2} $.
\end{proof}

\section{Appendix B: Computations Related to $Y^{\phi}$}
\label{sec8} \setcounter{equation}{0}
\setcounter{section}{8}

The purpose of this section is twofold. Namely, to give several equivalent versions of the equation satisfied by $Y^{\phi}$, and to compare the prescribed asymptotics of $Y^{\phi}$ with examples from the extreme Kerr spacetime.

Recall the basic defining equation for $Y^{\phi}$
\begin{equation}\label{800}
div_{\overline{g}}\text{ }\!\overline{k}(\eta)=0.
\end{equation}
Let $(\overline{\rho},\phi,\overline{z})$ denote Brill coordinates \eqref{42} for $\overline{g}$, with corresponding Christoffel symbols $\overline{\Gamma}_{ij}^{l}$. By Lemma \ref{lemma1} and the fact that $\partial_{\phi}$ is a Killing field
\begin{equation}\label{801}
\bk_{i \phi} = \frac{g_{\phi \phi}}{2u}\partial_{i}\byp,\text{ }\text{ }\text{ }\text{ }\text{ }\text{ }
\bg^{ij} \barG^{l}_{j \phi} \bk_{li}
= \frac{1}{2}\bk^{jl}(\partial_{j}\bg_{l \phi}-\partial_{l}\bg_{j \phi})=0,
\end{equation}
so that
\begin{align} \label{302}
\begin{split}
div_{\overline{g}}\text{ }\!\overline{k}(\eta)=\bg^{ij}\overline{\nabla}_{j}\bk_{\phi i}
&=\bg^{ij} \left(\partial_{j}\bk_{\phi i} - \barG^{l}_{ij}\bk_{l \phi} - \barG^{l}_{j \phi} \bk_{li} \right)  \\
&=\bg^{ij}\left(\frac{1}{2}\partial_{j} (u^{-1}g_{\phi \phi}\partial_{i}Y^{\phi})- \barG^{l}_{ij}\bk_{l \phi} \right)  \\
&=-\frac{g_{\phi\phi}}{u^{2}} \overline{g}^{ij}\partial_{j}u\partial_{i}Y^{\phi}
+ \frac{1}{2 u} \bg^{ij}\left( \partial_{j}(g_{\phi \phi} \partial_{i}\byp )
-g_{\phi\phi}\barG^{l}_{ij}\partial_{l}\byp\right)   \\
&=-\frac{g_{\phi\phi}}{u^{2}} \overline{g}^{ij}\partial_{j}u\partial_{i}Y^{\phi}
+\frac{g_{\phi \phi}}{2u} \left(\Delta_{\overline{g}}\byp +  \barna \log g_{\phi\phi}\cdot \barna \byp \right)  \\
&=\frac{g_{\phi \phi}}{2u} \left(\Delta_{\overline{g}}\byp + \barna \log (u^{-1}g_{\phi\phi})\cdot \barna \byp  \right).
\end{split}
\end{align}

Equation \eqref{302} may also be expressed explicitly in Brill coordinates.
Observe that
\begin{equation}  \label{303-1}
\bg^{pq} = e^{2\bu-2\bal}\delta^{pq}, \text{ }\text{ }\text{ }\text{ }
\bg^{p \phi} =  -A_{p}e^{2\bu - 2 \bal}, \text{ }\text{ }\text{ }\text{ }p,q=\overline{\rho},\overline{z},
\text{ }\text{ }\text{ }\text{ }
\bg^{\phi \phi}
= \brho^{-2}e^{2\bu} + e^{2\bu- 2\bal}(A_{\brho}^{2} + A_{\bz}^{2}),
\end{equation}
and
\begin{equation}\label{802}
e_p = e^{\bu - \bal}(\partial_{p}-A_{p}\partial_{\phi}),\text{ }\text{ }\text{ }\text{ }p=\overline{\rho},\overline{z},
\text{ }\text{ }\text{ }\text{ }
 e_{\phi} = \brho^{-1}e^{\bu} \partial_{\phi},
\end{equation}
so that
\begin{align} \label{303}
\begin{split}
\bg (\barna_{e_{\overline{\rho}}}e_{\overline{\rho}}, e_{\overline{z}} )
&=  e^{\bu- \bal}(\barG^{\overline{z}}_{\overline{\rho}\overline{\rho}} - 2A_{\brho}\barG^{\overline{z}}_{\overline{\rho} \phi} + A_{\brho}^{2}\barG^{\overline{z}}_{\phi \phi})  \\
&=e^{\bu - \bal} \left(\frac{\bg^{\overline{z}\overline{z}}}{2}(2\partial_{\overline{\rho}}\bg_{\overline{\rho} \overline{z}}-\partial_{\overline{z}}\bg_{\overline{\rho}\overline{\rho}}) + \bg^{\overline{z} \phi}\partial_{\overline{\rho}}\bg_{\overline{\rho} \phi} \right) \\
&-e^{\bu-\bal}A_{\brho} (\bg^{\overline{z}\overline{z}}(\partial_{\overline{\rho}}\bg_{\phi \overline{z}}-\partial_{\overline{z}}\bg_{\overline{\rho} \phi})+\bg^{\overline{z} \phi}\partial_{\overline{\rho}}g_{\phi \phi}) - \frac{1}{2}e^{\bu-\bal}A_{\brho}^{2}\bg^{\overline{z}\overline{z}}\partial_{\overline{z}}g_{\phi \phi} \\
&= \partial_{\overline{z}}(\bu - \bal)e^{\bu -\bal}.
\end{split}
\end{align}
Similarly
\begin{equation}\label{803}
\bg ( \barna_{e_{\overline{z}}}e_{\overline{z}}, e_{\overline{\rho}} )
= \partial_{\overline{\rho}}(\bu -\bal)e^{\bu-\bal}
\end{equation}
and
\begin{equation}\label{804}
\bg (\barna_{e_{\phi}}e_{\phi}, e_{p} )
= g_{\phi \phi}^{-1}e^{-\bu + \bal}\barG^{p}_{\phi \phi}
= -\frac{1}{2}e^{\bu -\bal}  \partial_{p}\log g_{\phi \phi},\text{ }\text{ }\text{ }\text{ }p=\overline{\rho},\overline{z}.
\end{equation}
It follows that
\begin{align} \label{304}
\begin{split}
\Delta_{\overline{g}}\byp
&= \sum_{p=\overline{\rho}, \overline{z}} \left( e_{p}(e_{p} \byp) - \barna_{e_{p}} e_{p} (\byp) - \barna_{e_{\phi}} e_{\phi} (\byp) \right) \\
&= \sum_{p=\overline{\rho}, \overline{z}} \left(e^{\bu-\bal}\partial_{p}(e^{\bu -\bal}\partial_{p}\byp)
- e^{2\bu - 2\bal} \partial_{p}(\bu - \bal)\partial_{p}\byp
+ \frac{1}{2}e^{2\bu -2\bal} \partial_{p}\log g_{\phi \phi} \partial_{p}\byp \right)  \\
&= \sum_{p=\overline{\rho}, \overline{z}} \left(e^{2\bu-2\bal}\partial_{p}^{2}\byp
+ \frac{1}{2}e^{2\bu -2\bal} \partial_{p} \log g_{\phi \phi} \partial_{p}\byp \right),
\end{split}
\end{align}
and
\begin{equation}\label{805}
\barna \log (u^{-1}g_{\phi\phi})\cdot \barna \byp
= \sum_{p=\overline{\rho}, \overline{z}} e^{2\bu - 2\bal} \left( g_{\phi\phi}^{-1}\partial_{p}g_{\phi\phi}-u^{-1}\partial_{p}u \right)\partial_{p}\byp.
\end{equation}
Hence
\begin{align} \label{306}
\begin{split}
div_{\overline{g}}\text{ }\!\overline{k}(\eta)
&=\frac{g_{\phi \phi}}{2u} \left(\Delta_{\overline{g}}\byp + \barna \log (u^{-1}g_{\phi\phi})\cdot \barna \byp  \right) \\
& = \frac{g_{\phi \phi}e^{2\bu -2\bal}}{2u} \sum_{p=\overline{\rho}, \overline{z}}
\left(\partial_{p}^{2} \byp + \left( \frac{3}{2}g_{\phi\phi}^{-1}\partial_{p}g_{\phi\phi}-u^{-1}\partial_{p}u \right)\partial_{p}\byp \right) \\
&=
\frac{e^{2\bu - 2\bal}}{2\sqrt{g_{\phi \phi}}}\sum_{p=\overline{\rho}, \overline{z}} \partial_{p} \left(u^{-1}g_{\phi \phi}^{3/2} \partial_{p}\byp \right).
\end{split}
\end{align}

We will now express \eqref{800} in terms of the metric $g$. Observe that
\begin{equation} \label{307}
\Delta_{\overline{g}}Y^{\phi}=\overline{g}^{ij}(\partial_{ij}Y^{\phi}-\Gamma_{ij}^{l}\partial_{l}Y^{\phi})
+\overline{g}^{ij}(\Gamma_{ij}^{l}-\overline{\Gamma}_{ij}^{l})\partial_{l}Y^{\phi}
\end{equation}
and
\begin{equation} \label{806}
\barna \log (u^{-1}g_{\phi\phi})\cdot \barna \byp
=\overline{g}^{ij}\partial_{i}\log (u^{-1}g_{\phi\phi})\partial_{j}Y^{\phi},
\end{equation}
where $\Gamma_{ij}^{l}$ are Christoffel symbols for $g$. In \textbf{Identity 1} of Section \ref{sec7} the difference between Christoffel symbols is computed, so that
\begin{equation} \label{308}
\bg^{ij}(\Gamma^{l}_{ij}- \barG^{l}_{ij})\partial_{l}\byp
= -w(\byp) Tr_{\bg}\pi
+  \frac{1}{2}\mid \barna f\mid_{\bg}^{2}\barna \varphi\cdot\barna \byp
- \bg^{ij}f^{l}(\overline{\nabla}_{l}Y_{j}-\overline{\nabla}_{j}Y_{l})\partial_{i} \byp.
\end{equation}
In order to proceed, we will also need
\begin{equation} \label{309}
Y^{i}=g^{ij}Y_{j}=g^{ij}(\byp g_{i \phi} +  | \by |^{2}_{\overline{g}}\partial_{i}f)
=\byp \delta^{i}_{\phi} + | \by |^{2}_{\overline{g}}f^{i},
\end{equation}
\begin{equation} \label{309-1}
\bg^{ij}
 = g^{ij} - \frac{u^{2}f^{i}f^{j}}{\volg}
- \frac{\byp (\delta^{i}_{\phi}f^{j}+f^{i}\delta^{j}_{\phi})}{\volg}
+  \frac{| \nabla f |^{2}_{g} (\byp)^{2} \delta^{i}_{\phi}\delta^{j}_{\phi}}{\volg}.
\end{equation}
It follows that
\begin{align} \label{310}
\begin{split}
&\bg^{ij}(\partial_{ij}\byp-\Gamma^{l}_{ij}\partial_{l}\byp) \\
=& \left(g^{ij} - \frac{u^{2}f^{i}f^{j}}{\volg} \right)\nabla^{2}\byp(\partial_{i}, \partial_{j})
+ \frac{2\byp f^{j}\Gamma^{l}_{\phi j}\partial_{l}\byp}{\volg}
-  \frac{| \nabla f |_{g}^{2} (\byp)^{2}\Gamma^{l}_{\phi \phi}\partial_{l}\byp}{\volg},
\end{split}
\end{align}
\begin{align}\label{807}
\begin{split}
\bg^{ij}(\Gamma^{l}_{ij}- \barG^{l}_{ij})\partial_{l}\byp
&= -\frac{u\nabla f(\byp)(Tr_{\bg}\pi)}{\sqrt{\volg}}
+ \frac{ |\nabla f |_{g}^{2}\nabla\varphi\cdot \nabla \byp}{2(\volg)}
-\frac{u^{2} | \nabla f |_{g}^{2}\nabla f(\varphi)\nabla f(Y^{\phi})}{2(\volg)^{2}} \\
&+ \frac{2\byp (\byp)^{l}\Gamma^{i}_{l \phi}\partial_{i}f}{\volg}
+ \frac{ Y_{\phi}| \nabla f |_{g}^{2} | \nabla \byp |_{g}^{2}}{\volg}
- \frac{Y_{\phi}(\nabla f(\byp))^{2}}{\volg},
\end{split}
\end{align}
and
\begin{align} \label{310-1}
\begin{split}
    \barna \log (u^{-1}g_{\phi\phi})\cdot \barna \byp
&= \left( g^{ij} - \frac{u^{2}f^{i}f^{j}}{\volg} \right)\partial_{i}\log (u^{-1}g_{\phi\phi})\partial_{j}\byp \\
&= \nabla \log (u^{-1}g_{\phi\phi})\cdot \nabla \byp
-\frac{u^{2}f^{l}\partial_{l}\log(u^{-1}g_{\phi\phi}) f^{j}\partial_{j}Y^{\phi}}{\volg}.
\end{split}
\end{align}
Next note that with the help of \eqref{35} and \eqref{309-1}
\begin{align} \label{311}
\begin{split}
Tr_{\bg} \pi
&=\left(g^{ij} - \frac{u^{2}f^{i}f^{j}}{\volg}\right)\pi_{ij} \\
&- \frac{g_{\phi \phi}\byp }{u^{2}(\volg)}
 \left(\frac{u \nabla f (\byp)}{\sqrt{\volg}} \right)
- \frac{(\byp)^{2} | \nabla f |_{g}^{2}}{\volg}
\left(\frac{u\Gamma^{l}_{\phi \phi}\partial_{l}f}{\sqrt{\volg}} \right).
\end{split}
\end{align}
Therefore employing \eqref{310}, \eqref{807}, \eqref{310-1}, \eqref{311}, and the identity
\begin{equation} \label{312}
\frac{1}{2} \varphi^{l} - (\byp)^{2}\Gamma^{l}_{\phi \phi} + Y_{\phi} (\byp)^{l}
= u u^{l},
\end{equation}
produces
\begin{align}\label{807.1}
\begin{split}
\Delta_{\overline{g}}\byp + \barna \log (u^{-1}g_{\phi\phi})\cdot \barna \byp&= \left(g^{ij}-\frac{u^{2}f^{i}f^{j}}{1+u^{2}|\nabla f|_{g}^{2}}\right)\left(\nabla_{ij}Y^{\phi}-\frac{u\pi_{ij}f^{l}}{\sqrt{1+u^{2}|\nabla f|_{g}^{2}}}\partial_{l}Y^{\phi}\right)\\
&+\left(g^{ij}-\frac{u^{2}f^{i}f^{j}}{1+u^{2}|\nabla f|^{2}}\right)\left(\partial_{i}\log g_{\phi\phi}-\frac{\partial_{i}\log u}{1+u^{2}|\nabla f|^{2}}\right)\partial_{j}Y^{\phi}.
\end{split}
\end{align}
We now record what has been shown.

\begin{lemma}\label{lemma4}
The following equations are equivalent:
\begin{equation} \label{808}
div_{\overline{g}}\text{ }\!\overline{k}(\eta)=0,
\end{equation}
\begin{equation} \label{809}
\Delta_{\overline{g}}\byp + \barna \log (u^{-1}g_{\phi\phi})\cdot \barna \byp=0,
\end{equation}
\begin{equation} \label{810}
\frac{e^{2\bu - 2\bal}}{2\sqrt{g_{\phi \phi}}}\sum_{p=\overline{\rho}, \overline{z}} \partial_{p} \left(u^{-1}g_{\phi \phi}^{3/2} \partial_{p}\byp \right)=0,
\end{equation}
\begin{align}\label{811}
\begin{split}
& \left(g^{ij}-\frac{u^{2}f^{i}f^{j}}{1+u^{2}|\nabla f|_{g}^{2}}\right)\left(\nabla_{ij}Y^{\phi}-\frac{u\pi_{ij}f^{l}}{\sqrt{1+u^{2}|\nabla f|_{g}^{2}}}\partial_{l}Y^{\phi}\right)\\
&+\left(g^{ij}-\frac{u^{2}f^{i}f^{j}}{1+u^{2}|\nabla f|^{2}}\right)\left(\partial_{i}\log g_{\phi\phi}-\frac{\partial_{i}\log u}{1+u^{2}|\nabla f|^{2}}\right)\partial_{j}Y^{\phi}=0.
\end{split}
\end{align}
\end{lemma}

Lastly, the prescribed asymptotics \eqref{132}-\eqref{134} for $Y^{\phi}$ will be compared with the example from
the (extreme) Kerr spacetime. Recall that in Boyer-Lindquist coordinates the Kerr metric takes the form
\begin{equation} \label{314}
 -\frac{\Delta - a^{2}\sin^{2}\theta}{\Sigma} dt^{2} + \frac{4ma\widetilde{r}\sin^{2}\theta}{\Sigma}dtd\phi
+  \frac{(\widetilde{r}^{2}+a^{2})^{2} - \Delta a^{2}\sin^{2}\theta}{\Sigma} \sin^{2}\theta d\phi^{2}
+ \frac{\Sigma}{\Delta} d\widetilde{r}^{2} + \Sigma d\theta^{2}
\end{equation}
where
\begin{equation}\label{812}
\Delta = \widetilde{r}^{2} + a^{2} -2m\widetilde{r},\text{ }\text{ }\text{ }\text{ }\text{ } \Sigma = \widetilde{r}^{2} + a^{2}\cos^{2}\theta.
\end{equation}
The event horizon is located at the larger of the two solutions to the quadratic equation
$\Delta=0$, namely $\widetilde{r}_{+}=m+\sqrt{m^{2}-a^{2}}$. For $\widetilde{r}>\widetilde{r}_{+}$ it holds that $\Delta>0$, so that a new radial coordinate may be defined by
\begin{equation}\label{813}
r=\frac{1}{2}(\widetilde{r}-m+\sqrt{\Delta}),
\end{equation}
or rather
\begin{equation} \label{315}
\begin{split}
\widetilde{r} &= r + m + \frac{m^{2}-a^{2}}{4r},\text{ }\text{ }\text{ }\text{ }\text{ } m^{2} \neq a^{2}  \\
\widetilde{r} &= r + m, \text{ }\text{ }\text{ }\text{ }\text{ } m^{2}=a^{2}.
\end{split}
\end{equation}
Note that the new coordinate is defined for $r>0$, and a critical point for the right-hand side of \eqref{315} ($m^{2}\neq a^{2}$) occurs at the horizon, so that two isometric copies of the outer region are encoded on this interval.  Moreover the $t=0$ slice of the metric takes the form \eqref{75}, showing that $(r,\phi,\theta)$ are
an appropriate set of Brill coordinates.

Observe that
\begin{equation} \label{317}
\byp= g^{\phi \phi} Y_{\phi} =  -\frac{2ma\widetilde{r}}{(\widetilde{r}^{2}+a^{2})^{2} -\Delta a^{2}\sin^{2}\theta}.
\end{equation}
Therefore at spatial infinity
\begin{equation} \label{318}
\byp \sim -\frac{2ma}{r^{3}} \quad \textrm{as} \quad r \to \infty,
\end{equation}
which is consistent with \eqref{132} since $\mathcal{J}=am$. Furthermore
\begin{equation} \label{319}
\begin{split}
\byp &= O(r^{3}),\text{ }\text{ }\text{ }\text{ }  m^{2} \neq a^{2}, \text{ }\text{ }\text{ as }\text{ }\text{ }r\rightarrow 0,    \\
\byp &= -\frac{2m^{2}a}{(m^{2}+a^{2})^{2}}+O(r), \text{ }\text{ }\text{ }\text{ }  m^{2} = a^{2}, \text{ }\text{ }\text{ as }\text{ }\text{ }r\rightarrow 0.
\end{split}
\end{equation}
This is consistent with \eqref{134}, but not \eqref{133}. The reason for the inconsistency is that the lapse function for the Kerr spacetime does not satisfy the required asymptotics \eqref{77}, whereas the lapse function
for the extreme Kerr spacetime does satisfy the desired asymptotics \eqref{78}.

\section{Appendix C: Boundary Terms}
\label{sec9} \setcounter{equation}{0}
\setcounter{section}{9}

Consider the basic inequality \eqref{48}. Under the hypotheses of Theorem \ref{thm2} this yields
\begin{equation}\label{900}
\overline{m}-\mathcal{M}(\overline{U},\overline{\omega})
\geq\frac{1}{8\pi}\lim_{\overline{r}\rightarrow\infty}
\int_{S_{\overline{r}}}uQ(\overline{\nu}) dA_{\overline{g}}
-\frac{1}{8\pi}\lim_{\overline{r}\rightarrow 0}
\int_{S_{\overline{r}}}uQ(\overline{\nu}) dA_{\overline{g}},
\end{equation}
where $\overline{\nu}$ is the unit normal pointing towards $M_{end}^{+}$ for the coordinate spheres $S_{\overline{r}}$. Here $(\overline{r},\phi,\overline{\theta})$ are spherical coordinates as in \eqref{75},
but with respect to $\overline{g}$. The purpose of this section is to show that
\begin{equation}\label{901}
\lim_{\overline{r}\rightarrow\infty}
\int_{S_{\overline{r}}}uQ(\overline{\nu}) dA_{\overline{g}}
-\lim_{\overline{r}\rightarrow 0}
\int_{S_{\overline{r}}}uQ(\overline{\nu}) dA_{\overline{g}}=8\pi\mathcal{Y}(\overline{\mathcal{J}})(\overline{\mathcal{J}}-\mathcal{J}),
\end{equation}
where $\mathcal{Y}(\overline{\mathcal{J}})=\lim_{r\rightarrow 0}Y^{\phi}$ as in \eqref{133}, \eqref{134}. Thus, the choice $\overline{\mathcal{J}}=\mathcal{J}$ guarantees \eqref{53}.

Recall that
\begin{equation} \label{401}
Q (\cdot)
= \overline{\nabla}^{2}f(\by, \cdot)-  \bk(u \barna f, \cdot)
+ (k-\pi)(w, \cdot)+  (k-\pi)(w,w)\sqrt{1+u^{2}|\nabla f|_{g}^{2}}\!\text{ }udf.
\end{equation}
It is clear from the asymptotics \eqref{31}, \eqref{40} that the first term on the left-hand side of \eqref{901}
vanishes, so we will focus on the second term. In what follows, it will be assumed that
\begin{equation} \label{902}
|k|_{g}+|k(\partial_{\phi},\cdot)|_{g}+|k(\partial_{\phi},\partial_{\phi})|\leq c\text{ }\text{ }\text{ on }
\text{ }\text{ }M.
\end{equation}
Note also that \eqref{76}-\eqref{78} and \eqref{81}-\eqref{83} imply that
\begin{equation} \label{903}
|\pi|_{g}+|\pi(\partial_{\phi},\cdot)|_{g}+|\pi(\partial_{\phi},\partial_{\phi})|\leq c\text{ }\text{ }\text{ on }\text{ }\text{ }M,
\end{equation}
and
\begin{equation} \label{402}
u \to 0, \text{ }\text{ }\text{ }\text{ }\text{ }|\nabla f|_{g} \to 0\text{ }\text{ }\text{ }\text{ as }\text{ }\text{ }\text{ }
r\rightarrow 0.
\end{equation}
Let us now consider terms in \eqref{401} when applied to $\overline{\nu}$. Since
\begin{equation}\label{402-1}
w=\frac{u \nabla_{g} f + u^{-1}\by}{\sqrt{\volg}},
\end{equation}
it follows that
\begin{equation}\label{904}
k(w,\overline{\nu})=u^{-1}k(\overline{Y},\overline{\nu})+O\left(u|\nabla f|_{g}^{2}|k(\overline{Y},\overline{\nu})|+u|k(\nabla f,\overline{\nu})|\right).
\end{equation}
Moreover
\begin{equation}\label{905}
|k(\overline{Y},\overline{\nu})|\leq|Y^\phi||k(\partial_{\phi},\cdot)|_{\overline{g}},
\end{equation}
and
\begin{align}\label{906}
\begin{split}
|k(\partial_{\phi},\cdot)|_{\overline{g}}^{2}&=\overline{g}^{ij}k_{il}\eta^{l}k_{jm}\eta^{m}\\
&=|k(\partial_{\phi},\cdot)|_{g}^{2}-\eta^{l}\eta^{m}w^{i}w^{j}k_{il}k_{jm}
+u^{-2}\eta^{l}\eta^{m}\overline{Y}^{i}\overline{Y}^{j}k_{il}k_{jm}\\
&\leq|k(\partial_{\phi},\cdot)|_{g}^{2}+u^{-2}(Y^{\phi})^{2}|k(\partial_{\phi},\partial_{\phi})|^{2},
\end{split}
\end{align}
so that
\begin{equation}\label{907}
|k(\overline{Y},\overline{\nu})|\leq c(|k(\partial_{\phi},\cdot)|_{g}+u^{-1}|k(\partial_{\phi},\partial_{\phi})|).
\end{equation}
Similarly
\begin{equation}\label{907}
|k(\nabla f,\overline{\nu})|\leq c(|k|_{g}|\nabla f|_{g}+u^{-1}|k(\partial_{\phi},\cdot)|_{g}|\nabla f|_{g}).
\end{equation}
Hence
\begin{equation}\label{908}
k(w,\overline{\nu})=u^{-1}k(\overline{Y},\overline{\nu})+O(|\nabla f|_{g}).
\end{equation}
Analogous computations show that
\begin{equation}\label{909}
k(w,w)=u^{-2}k(\overline{Y},\overline{Y})+O(|\nabla f|_{g}),
\end{equation}
and also
\begin{equation}\label{909.0}
\pi(w,\overline{\nu})=u^{-1}\pi(\overline{Y},\overline{\nu})+O(|\nabla f|_{g}),\text{ }\text{ }\text{ }\text{ }\text{ }
\pi(w,w)=u^{-2}\pi(\overline{Y},\overline{Y})+O(|\nabla f|_{g}).
\end{equation}
We now have
\begin{equation}\label{910}
Q(\overline{\nu})=\overline{\nabla}^{2}f(\by, \overline{\nu})-  u\bk( \barna f, \overline{\nu})
+ u^{-1}(k-\pi)(\overline{Y}, \overline{\nu})+u^{-1}(k-\pi)(\overline{Y},\overline{Y})\overline{\nu}(f)+O(|\nabla f|_{g}).
\end{equation}

According to \eqref{05-1} and \eqref{021}
\begin{equation} \label{403}
\pi(\by, \cdot)
= u\sqrt{1+u^{2}|\nabla f|_{g}^{2}}
\left(\overline{\nabla}^{2}f(\by, \cdot)+ u^{-1}\bk (\by, \cdot) - \pi(\bw,\by) df \right),
\end{equation}
so that
\begin{equation}\label{911}
\pi(\by, \overline{\nu})= u \overline{\nabla}^{2}f(\by, \overline{\nu})
+ \bk(\by, \overline{\nu}) - \pi(\by, \by)\overline{\nu}(f) +O(u|\nabla f|_{g}).
\end{equation}
Combining this with \eqref{910}, and the fact that $\bk(\barna f,\overline{\nu})=0$ (as $\overline{\nu}$
is the normal for an axisymmetric surface), produces
\begin{equation}\label{405}
Q(\overline{\nu})=u^{-1}k (\by, \overline{\nu} )
+ u^{-1}k(\by,\by)\overline{\nu}(f)
- u^{-1}\bk(\by,\overline{\nu})
+ O( | \nabla f |_{g}).
\end{equation}
In sum
\begin{align}\label{406}
\begin{split}
\lim_{\overline{r}\to 0}\int_{S_{\overline{r}}} u Q(\overline{\nu}) dA_{\bg}
&= \lim_{\overline{r} \to 0} \int_{S_{\overline{r}}}
\left( k(\by, \overline{\nu}) + k(\by, \overline{\nu})\overline{\nu}(f)
- \bk (\by, \overline{\nu}) +
O(u | \nabla f|_{g}) \right) dA_{\bg}    \\
&= \lim_{\overline{r} \to 0} \int_{S_{\overline{r}}}
\left( k(\by, \overline{\nu}) + k(\by,\by)\overline{\nu}(f)+ O(u | \nabla f |_{g})\right) dA_{\bg}
- 8  \pi  \mathcal{Y} \overline{\mathcal{J}},
\end{split}
\end{align}
where the last line is obtained from the definition of angular momentum and \eqref{30}.

In order to proceed, it will be necessary to express $\overline{\nu}$ in terms of quantities asasociated
with the metric $g$. Recall that $(M,g)$ is embedded via a graph $t=f(x)$ in the spacetime
$(\overline{M} \times \mathbb{R},\tg=\bg-2Y_{i}dx^{i}dt-\varphi dt^{2})$. Let $\mathcal{S}_{\overline{r}}\subset
M$ be the natural lifting of $S_{\overline{r}}\subset\overline{M}$ to the graph. The goal is to compute
$\overline{\nu}$ in terms of $\xi$, the unit normal to $\mathcal{S}_{\overline{r}}$ pointing towards $M_{end}^{+}$. Observe that an orthonormal
frame for $(\overline{M},\overline{g})$ is given by
\begin{equation}\label{912}
\overline{\nu}=e_{\overline{r}}=e^{\overline{U}-\overline{\alpha}}(\partial_{\overline{r}}-A_{\overline{r}}\partial_{\phi}),\text{ }\text{ }\text{ }\text{ }e_{\overline{\theta}}=\frac{e^{\overline{U}-\overline{\alpha}}}{\overline{r}}
(\partial_{\overline{\theta}}-A_{\overline{\theta}}\partial_{\phi}),
\text{ }\text{ }\text{ }\text{ }e_{\phi}=\frac{e^{\overline{U}}}{\overline{r}\sin\overline{\theta}}\partial_{\phi},
\end{equation}
and that
\begin{equation} \label{407}
X_{i} = e_{i} + e_{i}(f)\partial_{t},\text{ }\text{ }\text{ }\text{ } i= \br, \bth, \phi,
\end{equation}
is a basis for the tangent space of $(M,g)$. Thus, a normal to $\mathcal{S}_{\overline{r}}$ may be written as
\begin{equation} \label{407-1}
\xi':= X_{\overline{r}} + C_{\overline{\theta}}X_{\overline{\theta}} + C_{\phi}X_{\phi}
\end{equation}
for some constants $C_{\overline{\theta}}$, $C_{\phi}$. In order to calculate these constants, note that
\begin{align}\label{408}
\begin{split}
0 &= \tg (\xi', X_{\phi})   \\
&= \tg ( X_{\overline{r}}+C_{\overline{\theta}}X_{\overline{\theta}} + C_{\phi}X_{\phi}, e_{\phi} )  \\
&= \tg (e_{\br}(f) \partial_{t}, e_{\phi})
+ C_{\bth} \tg ( e_{\bth}(f)\partial_{t}, e_{\phi} )
+ C_{\phi}  \\
&= -e_{\br}(f)Y(e_{\phi}) - C_{\overline{\theta}}e_{\bth}(f)Y(e_{\phi}) + C_{\phi}
\end{split}
\end{align}
and
\begin{align}\label{409}
\begin{split}
0 &= \tg ( \xi', X_{\bth} )  \\
&= \tg  (X_{\overline{r}}+C_{\overline{\theta}}X_{\overline{\theta}} + C_{\phi}X_{\phi}, e_{\bth}+e_{\bth}(f)\partial_{t} )
 \\
&= \tg ( e_{\br}+e_{\br}(f) \partial_{t}, e_{\bth}+e_{\bth}(f)\partial_{t} )
+ C_{\overline{\theta}} \tg ( e_{\bth}+e_{\bth}(f) \partial_{t}, e_{\bth}+e_{\bth}(f)\partial_{t} )
+ C_{\phi}e_{\bth}(f)\tg ( e_{\phi}, \partial_{t} )  \\
&= -e_{\bth}(f) Y(e_{\br})-e_{\br}(f)Y(e_{\bth})- \varphi e_{\br}(f)e_{\bth}(f)\\
&
+ C_{\overline{\theta}}(1-2 e_{\bth}(f) Y(e_{\bth})-\varphi e_{\bth}(f)^{2})
- C_{\phi}e_{\bth}(f)Y(e_{\phi})  \\
&= - \varphi e_{\br}(f)e_{\bth}(f)
+ C_{\overline{\theta}}(1-\varphi e_{\bth}(f)^{2})
- C_{\phi}e_{\bth}(f)Y(e_{\phi}),
\end{split}
\end{align}
where in the last line the identity $Y(e_{\br}) = Y(e_{\bth}) = 0 $ is used. Solving for $C_{\overline{\theta}}$ and $C_{\phi}$ yields
\begin{equation}\label{410}
\xi'= X_{\overline{r}}+ \frac{u^{2}e_{\br}(f)e_{\bth}(f)}{1-u^{2}e_{\bth}(f)^{2}} X_{\overline{\theta}} + \frac{e_{\br}(f)Y(e_{\phi})}{1-u^{2}e_{\bth}(f)^{2}} X_{\phi},
\end{equation}
and hence
\begin{equation} \label{411}
\xi=\frac{\xi'}{|\xi'|_{g}}= \sqrt{\frac{1-u^{2}e_{\bth}(f)^{2}}{1-u^{2}|\overline{\nabla} f|_{\overline{g}}^{2}}}
\left(\overline{\nu}+\overline{\nu}(f)\partial_{t}
+\frac{u^{2}\overline{\nu}(f)e_{\bth}(f)}{1- u^{2}e_{\bth}(f)^{2}} X_{\overline{\theta}}
+\frac{\overline{\nu}(f)Y(e_{\phi})}{1-u^{2}e_{\bth}(f)^{2}} X_{\phi} \right).
\end{equation}

Consider now the integrand on the right-hand side of \eqref{406}. Since $k(\partial_{t}, \cdot)=0$ and
\begin{equation} \label{914}
|e_{i}(f)|\leq|\overline{\nabla} f|_{\overline{g}}=\frac{|\nabla f|_{g}}{\sqrt{1+u^{2}|\nabla f|_{g}^{2}}}\leq|\nabla f|_{g},\text{ }\text{ }\text{ }\text{ }Y(e_{\phi})e_{\phi}=\overline{Y},\text{ }\text{ }\text{ }\text{ }dA_{g} = \sqrt{1-u^{2}e_{\bth}(f)^{2}} dA_{\bg},
\end{equation}
it follows that
\begin{equation}\label{916}
\left(k(\by, \overline{\nu}) + k(\by,\by)\overline{\nu}(f)\right) dA_{\bg}
=\left(k(\by, \xi)+O(u|\nabla f|_{g})\right)dA_{g}\text{ }\text{ }\text{ as }\text{ }\text{ }
\overline{r}\rightarrow 0.
\end{equation}
Note also that the area $|\mathcal{S}_{\overline{r}}|$ grows like $\overline{r}^{\!\text{ }-2}$ when $M_{end}^{-}$ is asymptotically flat, and is bounded when $M_{end}^{-}$ is asymptotically cylindrical.
Therefore with the help of the asymptotics \eqref{76}-\eqref{78} and \eqref{81}-\eqref{83}
\begin{align}\label{416}
\begin{split}
\lim_{\overline{r}\to 0}\int_{S_{\overline{r}}} u Q(\overline{\nu}) dA_{\bg}
&= \lim_{\overline{r} \to 0} \int_{\mathcal{S}_{\overline{r}}}
\left( k(\by,\xi) + O(u | \nabla f |_{g})\right) dA_{g}
- 8  \pi  \mathcal{Y} \overline{\mathcal{J}}\\
&= 8\pi \mathcal{Y} (\mathcal{J}-\overline{\mathcal{J}}).
\end{split}
\end{align}

\section{Appendix D: Miscellaneous Formulae}
\label{sec10} \setcounter{equation}{0}
\setcounter{section}{10}

In this section we will compute certain Christoffel symbols used in Section \ref{sec6}, and record how the twist potential encodes angular momentum.

Christoffel symbols will be expressed in terms of the Brill coordinate system \eqref{75}, where $\rho=r\sin\theta$. Observe that components of the inverse metric are given by
\begin{equation}\label{d1}
g^{rr} = e^{2U-2\alpha},\text{ }\text{ }\text{ }\text{ }
g^{\theta \theta} = r^{-2} e^{2U-2\alpha},\text{ }\text{ }\text{ }\text{ }
g^{\phi \phi}
= \rho^{-2}e^{2U} + e^{2U- 2\alpha}(A_{r}^{2} + r^{-2}A_{\theta}^{2}),
\end{equation}
\begin{equation}\label{d2}
g^{r\theta} =0,\text{ }\text{ }\text{ }\text{ }
g^{r \phi} =  -A_{r}e^{2U - 2 \alpha},\text{ }\text{ }\text{ }\text{ }
g^{\theta \phi} =  - r^{-2}A_{\theta}e^{2U - 2 \alpha}.
\end{equation}
It follows that
\begin{equation}\label{d3}
\Gamma_{rr}^{r}=\frac{1}{2}e^{2U-2\alpha}\partial_{r}(e^{-2U+2\alpha}
+\rho^{2}e^{-2U}A_{r}^{2})-e^{2U-2\alpha}A_{r}\partial_{r}(\rho^{2}e^{-2U}A_{r}),
\end{equation}
\begin{equation}\label{d4}
\Gamma_{\theta\theta}^{r}=\frac{1}{2}e^{2U-2\alpha}
\left[2\partial_{\theta}(\rho^{2}e^{-2U}A_{r}A_{\theta})-\partial_{r}(r^{2}e^{-2U+2\alpha}
+\rho^{2}e^{-2U}A_{\theta}^{2})\right]
-e^{2U-2\alpha}A_{r}\partial_{\theta}(\rho^{2}e^{-2U}A_{\theta}),
\end{equation}
\begin{equation}\label{d5}
\Gamma_{\phi\phi}^{r}=-\frac{1}{2}e^{2U-2\alpha}\partial_{r}(\rho^{2}e^{-2U}),
\end{equation}
\begin{equation}\label{d6}
\Gamma_{r\theta}^{r}=\frac{1}{2}e^{2U-2\alpha}\partial_{\theta}(e^{-2U+2\alpha}
+\rho^{2}e^{-2U}A_{r}^{2})
-\frac{1}{2}e^{2U-2\alpha}A_{r}\left[\partial_{r}(\rho^{2}e^{-2U}A_{\theta})
+\partial_{\theta}(\rho^{2}e^{-2U}A_{r})\right],
\end{equation}
\begin{equation}\label{d7}
\Gamma_{r\phi}^{r}=-\frac{1}{2}e^{2U-2\alpha}A_{r}\partial_{r}(\rho^{2}e^{-2U}),
\end{equation}
\begin{equation}\label{d8}
\Gamma_{\theta\phi}^{r}=-\frac{1}{2}e^{2U-2\alpha}A_{r}\partial_{\theta}(\rho^{2}e^{-2U})
+\frac{1}{2}e^{2U-2\alpha}\left[\partial_{\theta}(\rho^{2}e^{-2U}A_{r})
-\partial_{r}(\rho^{2}e^{-2U}A_{\theta})\right].
\end{equation}

We now record how the twist potential encodes angular momentum. Again, consider the coordinate system in \eqref{75}.
An orthonormal basis is given by
\begin{equation}\label{a2}
e_{r}=e^{U-\alpha}(\partial_{r}-A_{r}\partial_{\phi}),\text{ }\text{ }\text{ }\text{ }e_{\theta}=\frac{e^{U-\alpha}}{r}(\partial_{\theta}-A_{\theta}\partial_{\phi}),
\text{ }\text{ }\text{ }\text{ }e_{\phi}=\frac{e^{U}}{r\sin\theta}\partial_{\phi}.
\end{equation}
The twist potential may be calculated in terms of $k$ by
\begin{equation}\label{a3}
\frac{1}{2}\partial_{i}\omega=\epsilon_{ijl}k^{j}_{m}\eta^{l}\eta^{m},
\end{equation}
where $\epsilon_{ijl}$ is the volume element of $g$. It follows that
\begin{align}\label{a4}
\begin{split}
\frac{e^{U-\alpha}}{2r}\partial_{\theta}\omega &=\frac{1}{2}e_{\theta}(\omega)\\
&=-\epsilon(e_{r},e_{\theta},e_{\phi})k(e_{r},e_{\phi})|\eta|^{2}\\
&=-e^{-U} k(e_{r},\partial_{\phi})r\sin\theta,
\end{split}
\end{align}
or rather
\begin{equation}\label{a5}
k(e_{r},\partial_{\phi})=-\frac{e^{2U-\alpha}}{2r^{2}\sin\theta}\partial_{\theta}\omega.
\end{equation}
Hence, if there are only two ends
\begin{align}\label{a6}
\begin{split}
\mathcal{J}&=\frac{1}{8\pi}\int_{S_{\infty}}(k_{ij}-(Tr k)g_{ij})\nu^{i}\eta^{j}\\
&=\lim_{r\rightarrow0}\frac{1}{8\pi}\int_{\partial B(r)}k(\partial_{\phi},e_{r})dA\\
&=\lim_{r\rightarrow0}\frac{1}{8\pi}\int_{\partial B(1)}k(\partial_{\phi},e_{r})e^{-2U+\alpha}r^{2}\sin\theta d\theta d\phi\\
&=-\lim_{r\rightarrow0}\frac{1}{16\pi}\int_{\partial B(1)}\partial_{\theta}\omega d\theta d\phi\\
&=\frac{1}{8}(\omega|_{I_{+}}-\omega|_{I_{-}}).
\end{split}
\end{align}

\end{document}